\documentclass[11pt]{article}
\usepackage{mydef2col}
\usepackage{markArticle}

\title{Valuing American options and Flexible Forwards contracts in time-dependent models}
\shorttitle{Flexible Forwards in time-dependent models}
\author{
\authorstyle{
Leif Andersen\textsuperscript{1}\thanks{e-mail: \url{leif.andersen@bofa.com}} \,
Andrey Itkin\textsuperscript{2}\thanks{e-mail: \url{aitkin@nyu.edu}} \,
and Rakhymzhan Kazbek\textsuperscript{3}\thanks{e-mail: \url{rakhymzhan.kazbek@umu.se}}
}
\newline\newline
\textsuperscript{1}
\institution{Bank of America, New York, USA.} \\
\textsuperscript{2}
\institution{FRE Department, Tandon School of Engineering, New York University, USA.} \\
\textsuperscript{3}
\institution{Department of Mathematics and Mathematical Statistics, UMEA Universitet, Sweden.}
}

\date{\today}

\begin{document}

\maketitle

\lettrineabstract{A flexible forward (FF) is a customized FX hedging instrument that
guarantees a fixed exchange rate while letting the holder choose the delivery date
within a pre-agreed window. It is therefore an American-style option on timing, and
its valuation must respect the volatility skew of the underlying currency pair. We
price FF contracts (and, more generally, American options) under a time-inhomogeneous
Heston model which captures the forward-skew term structure while preserving analytical tractability through a recursive (matrix) Riccati solution for the joint characteristic function. Extending the integral-equation (decomposition) approach to time-dependent
coefficients, we derive a Volterra equation characterizing the early-exercise surface. The expectation in the decomposition formula is evaluated by two complementary spectral methods: a double cosine (COS) expansion of the transition density, and a damped-Sinc (DSINC) local-basis scheme that is more accurate and stays robust when a low Feller ratio or large vol-of-vol induces Gibbs oscillations in the COS series. Benchmarked against a
penalty-iteration MCS-ADI finite-difference solver, both methods price a contract
in about 1-2 seconds, roughly an order of magnitude faster than the finest
finite-difference grid, while DSINC improves median accuracy over COS by about a
factor of twelve. The experiments also show that the early-exercise surface is a
substantially nonlinear function of the variance, contrary to the linear-in-variance
approximation common in earlier work.}


%


\section{Introduction}

A flexible forward (FF) contract (also known as a window forward or optional forward contract) is a customized foreign exchange hedging tool that combines the guarantee of a fixed exchange rate with flexibility on the timing of the transaction. Unlike a standard forward contract that locks in a rate for a single, specific future date, a FF allows the holder to execute the trade anytime within a pre-agreed period. Two parties agree on a fixed exchange rate, a total notional amount, a delivery window, and often minimum and maximum tranche sizes for partial drawdowns. During the delivery window, the holder can call upon the bank to execute the contract for part or all of the total amount at the pre-agreed fixed rate on any business day they choose. Each time a portion is executed, that amount is settled and delivered, with the obligation reduced accordingly. Any unused portion expires worthless at the end of the window.

The guaranteed rate on the FF is typically less favorable than the rate for a standard forward set at the worst point in the window. This is the cost of the optionality, which the bank prices into the rate. The instrument is ideal for businesses with uncertain cash flows. For example, a company expecting EUR payments from multiple clients over a quarter can hedge the total expected amount without knowing the exact dates. It locks in a worst-case rate, removing downside risk. If the spot rate moves favorably during the window, the company can execute on those better days for portions of the hedge, blending its overall rate. However, the contract is binding---the total notional amount must be exchanged by the end of the window, but the holder controls the timing.

A typical use case is a U.S. importer who needs to pay a European supplier €1,000,000 sometime in the third quarter, with uncertain invoice dates. The current spot rate is 1.08 USD/EUR. A standard forward contract locking in a single date (e.g., September 30) could lead to a mismatch if payment is due July 15th. Instead, the importer enters a FF for €1,000,000 with a window of July 1--September 30 at a fixed rate of 1.0950 USD/EUR. If the euro strengthens to 1.12 in August, the importer executes the contract in August, saving money vs. the spot market. If the euro weakens to 1.07 in July, they retain the flexibility to wait, potentially executing later if the rate recovers, or simply executing at their guaranteed 1.0950. They can draw down the amount in multiple payments aligned with actual supplier invoices.

Among advantages, the FF contract eliminates timing risk for uncertain cash flows, guarantees a worst-case rate, retains some opportunistic benefit if the spot market moves favorably, and is more efficient than multiple standard forwards for the same period. Drawbacks include a less favorable rate than a standard forward, greater complexity, and the need for active management during the window. In essence, a FF contract represents a trade-off: the client sacrifices a small amount of the guaranteed forward rate in exchange for control over execution timing.

It can be easily observed that a FF is essentially an American-style option on timing, where the holder chooses the optimal delivery date within a window. Accordingly, this product can be priced using methods suitable for American options; see, e.g., the surveys in \cite{Itkin2025ddj,Andersen2025} and references therein. However, the choice of model must be made with care. In currency pairs like EUR/USD or USD/JPY, there is typically strong institutional demand to hedge against sudden, large moves in one direction. For example, Japanese exporters continuously buy USD/JPY puts, driving up implied volatility for downside strikes and creating a volatility skew. The skew is directly quoted in the market as a Risk Reversal (RR). A 25-Delta RR is the difference in volatility between a 25D call and a 25D put. A positive RR means OTM calls are more expensive than OTM puts, indicating a bullish USD bias; a negative RR indicates a bearish bias. Modeling skew is necessary to price and hedge these structures correctly. A flat volatility model would value them at zero. For vanilla options, using the correct market-implied volatility for each strike is often sufficient. However, for exotics, the dynamics implied by the skew are essential.

If markets for American-style options with the same underlying were liquid, one could theoretically derive a local volatility surface from them. But such markets are rarely liquid, and the inversion formula would be much more complex due to the early exercise boundary. The most common practitioner approach to resolve this is known as deAmericanization, see \cite{Broadie2007} among others. This approach calibrates a local volatility surface $\sigma(t,S)$ to market prices of liquid European options and then uses this $\sigma(t,S)$ in a pricing PDE, tree, or Monte Carlo with early exercise for the American-style instrument. This is a practical compromise, not theoretically pure but widely used. Traditionally, numerical methods such as binomial or trinomial trees are employed to determine this implied volatility (assumed constant), see \cite{Burkovska2018} and references therein, while the semi-analytical approach of \cite{Henry-Labordere2017} provides more tractability despite less accuracy.

\subsection{Choice of model}

In FX derivatives pricing, the primary risk factor and source of optionality is the exchange rate itself. The volatility of the FX spot rate is orders of magnitude more significant for pricing and risk management than the volatility of underlying interest rates for most standard products. FX spot volatility for a major G10 pair typically ranges from 5\% to 15\% annualized and can spike much higher during crises, while interest rate volatility for short-term rates is usually in the range of 0.5\% to 2\% annualized. A 1\% move in spot is a common daily occurrence, while a 1\% move in an interest rate is a massive event. Thus, for vanilla options and many exotics, adding stochastic rates complicates the model for negligible gain in pricing accuracy.

An FX options desk's primary risk sensitivities are delta (hedged with spot or forwards), gamma/vega (hedged with other options), and vanna/volga (sensitivity to skew dynamics). Interest rate sensitivity (rho) exists but is smaller in magnitude and easily hedged linearly using deposits, futures, or cross-currency swaps. Using stochastic rates would introduce cross-gammas that are not only tiny but also difficult and expensive to hedge. The complexity cost outweighs the risk-reduction benefit. Notable exceptions where stochastic rates become relevant include long-dated FX (LDFX) beyond 10 years, but FF contracts typically involve much shorter tenors.

For FF contracts, the flexibility window is typically 1 to 3 months long, though it can be longer for strategic hedges. A 3-month window is market standard for medium-term hedges, very common for hedging confirmed near-term cash flows. Medium-term windows of 6 to 12 months are also popular and liquid, used by corporates to hedge forecasted cash flows. Longer-term windows of 1 to 2 years (occasionally up to 5 years) are used for hedging longer-term forecasted exposures, though liquidity diminishes and the cost of flexibility increases with tenor. While FFs can be structured for almost any tenor, the core of the market is the 3-month flexible window in the 6-month to 1-year forward space, as it matches quarterly cash flow forecasting cycles.

\subsection{Pricing FF under the time-inhomogeneous Heston model} \label{sec:t_inhom}

As discussed, pricing FF contracts with typical maturities up to 1-2 years can be effectively handled by assuming deterministic interest rates. A comparison of FF option prices under deterministic and stochastic rates in \cite{gregori_flexible_2023} confirms that the difference is relatively small for tenors up to one year. However, capturing the FF volatility skew requires a model incorporating at least stochastic volatility.

The first idea would be to proceed with the well-known time-homogeneous Heston model \cite{Heston_32}, since it is affine and tractable. However, this model has several disadvantages for our purposes. First, it produces a symmetric smile/skew that becomes flatter at longer maturities $T$, since the ATM skew decays as approximately $1/\sqrt{T}$. Therefore, it cannot generate persistent steep skews at longer maturities without unrealistic parameters. Second, as in any time-homogeneous model, the fixed mean-reversion level $\theta$ limits the flexibility to calibrate to the FF term structure.

A better choice would be the 3/2 model introduced in \cite{Heston_32, Lewis:2000}, a non-affine stochastic volatility model whose analytical tractability has been extensively studied; see \cite{CarrSun, ItkinCarr2010, Drimus2012}, among others. For the same parameter magnitude, it produces a steeper skew than the Heston model, generating heavy tails via excess kurtosis in the return distribution. It is also analytically tractable and would generate realistic forward skew for FF contracts by embedding leverage effect enhancement. However, it has fewer parameters than the time-inhomogeneous Heston model and is more challenging for calibration. Moreover, the characteristic function (CF) of this model is expressed via Gamma and Kummer confluent hypergeometric functions \cite{as64}, so computational time is primarily driven by evaluating complex-valued special functions, which is slow as standard mathematical packages lack fast routines.

Therefore, a possible alternative we consider in this paper is a \textbf{time-inhomogeneous Heston model}. This model is capable of providing an exact calibration to today's vanilla surface, preserving steep skews across all maturities via time-dependent correlation $\rho(t)$, and generating a flexible forward-starting skew term structure via $\theta(t)$ and vol-of-vol $\xi(t)$ - critical for FF contracts.

At the same time, such a model has many degrees of freedom, so future calibrations may be inconsistent. To resolve this, we adopt a hybrid approach where we consider the Heston model with time-dependent mean-reversion, vol-of-vol, and correlation, but treat them as having known functional forms (ansatz) of time $t$ with unknown constant coefficients. This hybrid approach combines calibration flexibility with the parsimony and structural integrity of parametric forms, maintaining analytical tractability while reducing overfitting risk and providing smooth forward evolution. A particularly effective hybrid parameterization treats $\theta(t)$ as a continuous function (e.g., cubic spline) while modeling $\rho(t)$ and $\xi(t)$ as piecewise constants over intervals aligned with traded option tenors, such as [1M, 3M, 6M, 1Y, 2Y].

Our choice is motivated by three observations: (a) long-term volatility $\theta(t)$ evolves smoothly over time, reflecting gradual shifts in economic conditions; (b) correlation $\rho(t)$ can undergo discrete regime shifts, aligning with observed market phases; (c) vol-of-vol $\xi(t)$ exhibits clustering and regime-dependent behavior. Importantly, as we show below, this modeling approach retains full analytical tractability.

\subsection{Pricing Methodology}

The valuation of American options within complex models typically admits no closed-form solution, making numerical approaches standard.

\paragraph{Least Squares Monte Carlo.} For models with deterministic interest rates and stochastic variance $v_t$, the LSM approach (see \cite{tsitsiklis2001regression, stentoft2004assessing, broadie1997pricing, andersen2000simple}) simulates forward paths of $(S_t, v_t)$ and works backward, using regression to estimate continuation values. However, several well-known challenges remain: (a) LSM provides only a lower bound on the true option price; (b) obtaining an upper bound requires a separate, computationally expensive dual approach \cite{andersen2004primal}; (c) poor choice of basis functions leads to significant underestimation; (d) look-ahead bias can overestimate continuation values; (e) avoiding bias by simulating sub-paths becomes prohibitively expensive; (f) regression errors propagate backward; (g) deep OTM paths pose particular difficulties. As applied to pricing FF contracts, see \cite{HolblLovric2019}, perhaps the first paper focusing on mark-to-model valuations of open FX forwards.

\paragraph{PDE/Tree Methods with Early Exercise.} Finite-difference methods are an industry standard for pricing American/Bermudan options for models with low-dimensional state spaces, offering high accuracy, convergence guarantees, efficient Greeks computation, and good handling of early exercise constraints. For details, see \cite{achdou2005computational, chiarella2009numerical} and the thorough survey in \cite{Andersen2025}. Tree-based methods have also drawn attention, \cite{GiriboneLigato2022, GiriboneRaviola2019}.

\paragraph{An Integral Equation Approach.}

Various proxy approaches to pricing FF contracts have been developed, based on semi-analytical approximations of American put option prices, as in \cite{geske1984american}. These often yield satisfactory results, particularly when domestic and foreign interest rates differ significantly, and correctly capture pricing asymptotics for extreme spot levels \cite{HokTse2024}. However, when rates are close and the strike is near-the-money, proxy pricing can deviate by dozens of basis points \cite{krifka2025approximate}---an error of the same order as typical sales margins, and therefore not ignorable.

For more sophisticated models with multiple assets or stochastic factors, solving multi-dimensional free-boundary PDEs becomes computationally intensive. Consequently, an alternative semi-analytical approach has been developed for several time-homogeneous models, based on a change of variables formula introduced in \cite{Peskir2005, Peskir2007} and applied to models including Heston \cite{Chiarella2005}, Merton jump-diffusion \cite{Chiarella2009}, and the 3/2 model \cite{DetempleKitapbayev2017}.

As applied to time-inhomogeneous jump-diffusion models and models with arbitrary drift, an extension has been developed in \cite{ItkinKitapbayev2025r, ItkinCF2025j}, covering situations where either the joint transition density or joint CF is known in closed form, or the GIT method of \cite{CarrItkin2020jd, ItkinMuravey2024jd, Itkin2024jd, Itkin2025ddj, ItkinLiptonMuraveyBook} is used to efficiently compute the early exercise boundary. The main idea is to solve Volterra integral equations where the early exercise boundary is determined explicitly and separately from the option price, providing better accuracy and speed than PDE methods where this is done implicitly and simultaneously.

This method addresses key problems: (1) extending the decomposition approach to general jump-diffusion and \LY models; (2) handling cases where closed-form transition densities are unavailable by leveraging CFs via, e.g., the COS method \cite{FangOosterlee2008}; and (3) generalizing to multidimensional diffusions.

In this paper, we propose a time-inhomogeneous Heston model for FF contracts, a special case of the multidimensional diffusion framework in Section 3 of \cite{ItkinCF2025j}. In contrast to the brief numerical example provided there for the time-inhomogeneous 3/2 model, we develop a complete pricing methodology for this stochastic FF model. We conduct a thorough numerical analysis, benchmarking against a finite-difference method with respect to precision and efficiency.

The remainder of this paper is organized as follows. \Cref{modelSpec} details the model specification, including the formal definition of the FF contract in \cref{contract} and the computation of the joint CF for time-inhomogeneous coefficients in \cref{condCF}. \Cref{IEapproach} presents the Integral Equation approach for pricing American options, covering the computation of the joint transition density in \cref{jtp} and the exercise boundary in \cref{compEB}. \Cref{pdeSol} describes a benchmark solution using modern finite-difference methods. \Cref{numEx} provides a numerical comparison between the semi-analytical solution via our IE approach and the PDE benchmark. \Cref{sec:conclusion} offers concluding remarks.

\section{Model specification} \label{modelSpec}

Consider the Heston stochastic volatility model with several time-dependent parameters
\begin{align} \label{model}
dx_t &= \left( r_d(t) - r_f(t) - \dfrac{1}{2}v_t \right) dt + \sqrt{v_t} dW_t^S, \\
dv_t &= \kappa\left(\theta(t) - v_t\right) dt + \xi(t) \sqrt{v_t} dW_t^v, \nonumber \\
dW_t^S dW_t^v &= \rho(t) dt, \nonumber
\end{align}
where $x_t = \log (S_t/K)$ is the spot log-price in domestic/foreign terms, $K$ is the option strike, $W^{(1)}$ and $W^{(2)}$ are two standard correlated Brownian motions under domestic risk-neutral measure $\mathbb{Q}^d$ with the constant correlation coefficient $\rho(t) \in [-1,1]$, $\kappa$ is the mean-reversion speed, $\xi(t)$ is the volatility of variance $v_t$ (vol-of-vol), $\theta(t)$ is the mean-reversion level (the long-term run), $r_d, r_f$ are the instantaneous domestic and foreign interest rates. If the so-called Feller condition $\Fe = 2 \kappa \theta(t) > \xi^2(t)$ is satisfied for all times $t \in [0,T]$, the process $v_t$ is strictly positive, $v_t \in [0,\infty)$; otherwise its behavior at the origin should be additionally identified, see e.g., \cite{f54,CarrLinetsky2006,lucic_boundary_2012}.

We supplement the model with the following parameter specifications
\begin{itemize}
\item $\kappa > 0, v_0 > 0$: \emph{constant} mean reversion speed and the initial variance at time $t=0$;
\item $\theta(t)$: \emph{continuous} function of time (e.g., exponential decay: $\theta(t) = a + be^{-ct}$);
\item $r_d(t), r_f(t)$: the domestic and foreign instantaneous interest rates are either \emph{continuous} or \emph{discrete} functions of time;
\item $\sigma(t), \rho(t)$: \emph{piecewise constant} on intervals $[t_{k-1}, t_k)$, $k=1,\dots,N$, $0 < t_1 < \dots < t_N = T$.
\end{itemize}

\subsection{Specification of the FF contract} \label{contract}

To proceed, we need to specify the terms of a FF contract. Let
\begin{itemize}
    \item $N_{f}$ be the aggregate notional in foreign units;
    \item $[T_{1}, T_{2}]$ be the time window to exercise the contract;
    \item $\eta$: $+1$ if the FF buyer wants to buy the foreign currency, or $-1$ if the buyer wants to sell it.
\end{itemize}
Let us also denote
\begin{equation}
D_{x}(t,T)  = e^{-\int_{t}^{T} r_{x}(u)du}, \qquad x = d, f,
\end{equation}
to be the deterministic discount factor in both currencies.

Under the terms of the FF, the buyer has the right to freely design a notional drawdown schedule inside the window, where the drawn notional gets converted at the quoted exchange rate $K$, rather than the prevailing spot exchange rate $S_t$. One caveat is that the entire notional $N_{f}$ must be drawn at or before time $T_{2}$. If the buyer decides to draw down $n(t)dt$ during the time period $[t, t+dt]$, then it must be the case that
\begin{equation} \label{L1}
    \int_{T_{1}}^{T_{2}} n(t)dt = N_{f}.
\end{equation}

From a cash-flow perspective, the draw of $n(t)dt$ foreign units will be converted to a receipt of $\eta\, n(t)K dt $ domestic units. Since the fair exchange amount would have been $\eta\, n(t) S_t dt$, the time $t$ value to the buyer associated with the draw $n(t)$ at time $t$ is, in domestic units $\eta\, n(t)(K - S_t)dt$.

From \eqref{L1}, the total time-$t_0$ value in domestic units associated with a draw strategy $n(\cdot)$ is thus
\begin{equation}
V_{n}(0) = \eta\, \mathbb{E} \left( \int_{T_{1}}^{T_{2}} D_{d}(0,t)n(t)(K - S(t))dt \right).
\end{equation}
If $n(t)$ were deterministic, i.e., determined at time $t_0$, we would write the right-hand side as
\begin{equation} \label{L3}
\eta \left( \int_{T_{1}}^{T_{2}} n(t)(K D_{d}(0,t) - S_0 D_{f}(0,t))dt \right).
\end{equation}
However, importantly, the FF allows for dynamic payment schedules determined on the fly by the buyer, though always subject to \eqref{L1}.

The option seller would always assume that the option buyer acts so as to maximize\footnote{Going forward we just set $\eta=1$; the case $\eta=-1$ is handled analogously.} the payout value of the FF contract, so the fair value of the contract is
\begin{equation} \label{L4}
V(0) = \mathbb{E} \left( \int_{T_{1}}^{T_{2}} D_{d}(0,t) n^{*}(t) (K - S_t) dt \right),
\end{equation}
where $n^{*}(t)$ is the optimal notional drawdown strategy, i.e., the one that maximizes $V_{n}(0)$.

While estimating $n^{*}$ above may appear formidable, we can use backward induction to show that the optimal strategy must be of the “bang-bang” type: either draw nothing or draw everything. That is,
\begin{equation}
    n^{*}(t) = N_{f} \delta(\tau - t),
\end{equation}
where $\tau \in [T_{1}, T_{2}]$ is the optimal time to exercise the FF. Therefore, we can simplify the valuation expression to\footnote{In practice, FF buyers may not act this way, since the primary goal of FF's is to hedge against unknown business-driven flows of foreign currency. The option seller, to be conservative, would still assume optimal exercise behavior.}
\begin{equation}
V(0) = \sup_{\tau \in [T_{1}, T_{2}]} \mathbb{E} \left( D_{d}(0, \tau) (K - S_{\tau}) \right).
\end{equation}

The time $\tau$ may be characterized by an exercise boundary
\begin{equation}
\tau = \inf \{ t \in [T_{1}, T_{2}] : S_t < S^*(t) \}.
\end{equation}
Since exercise is mandatory within the window $[T_1,T_2]$, at $t = T_2$ the boundary effectively disappears or becomes infinite. There is no continuation region left; regardless of how high or low the spot price is, the contract must be settled.

In the limit $S^*(T_2^-)$ (approaching the end of the window from the left), the optimal EB typically converges to a specific economic threshold rather than the strike price. For an FF long position, the boundary $S^*(T_2^-)$ is defined by the ratio of the interest rates:
\begin{equation} \label{SB_T2}
S^*(T_2^-) = K \dfrac{r_d(T_2)}{r_f(T_2)}.
\end{equation}
This result comes from the no-arbitrage condition of early exercise. Indeed, if one holds the contract for an infinitesimal time $dt$ longer:
\begin{itemize}
\item \emph{Cost of waiting:} the loss of interest that could have been earned on the domestic currency: $r_d K dt$.

\item \emph{Benefit of waiting:} the gain of interest on the foreign asset not yet purchased: $r_f S_t dt$.
\end{itemize}
At the optimal boundary just before expiration, these two must balance. Setting $r_d K = r_f S_t$ gives the boundary. Notice that $S^*(T_2-)$ can be either above or below $K$, unlike the proper American put\footnote{For a regular American put with strike $K$, we would have $S^*(T_2) = K$ and $S^*(T_2-) = K \min \left(1, r_d(T_2)/r_f(T_2)\right)$.}.

The steps that led to \eqref{SB_T2} can also be used to establish a necessary condition for exercise at other values of $t$, giving an upper bound for the EB:
\begin{equation}
S^*(t) \leq K \dfrac{r_d(t)}{r_f(t)}, \quad t \in\left[T_1, T_2\right).
\end{equation}
Speaking of bounds, we can define the time-$t_0$ American premium for the FF as the difference between the true FF price and the best possible \emph{deterministic} strategy established at time $t_0$:
\begin{equation}
AP(0) = V(0) - \tilde{V}(0) \geq 0, \qquad \tilde{V}(0) = N_f \max _{t \in\left[T_1, T_2\right]} \left(K D_d(0, t) - S_0 D_f(0, t)\right).
\end{equation}
In the definition of $\tilde{V}(0)$, we relied on \eqref{L3} and \eqref{L4}. In practice, it is often the case that:
\begin{itemize}
    \item $\tilde{V}(0)$ is an edge case: $\tilde{V}(0) \approx N_f \max \left(K D_d\left(0, T_1\right)-S_0 D_f\left(0, T_1\right), K D_d\left(0, T_2\right)-S_0 D_f\left(0, T_2\right)\right)$.
    \item The premium $AP(0)$ is very small compared to ordinary American-style options.
\end{itemize}
This is a consequence of the FF’s payout being linear in the underlying, in contrast to the convex payout of a Put or Call. Consequently, FFs generally exhibit low Vega.  Nevertheless, the presence of embedded timing optionality, combined with the volatility skew of the underlying market, allows an implied volatility to be associated with an FF strike, making it strike-dependent, i.e., there is a volatility skew in the FFs.

\subsection{Joint characteristic function} \label{condCF}

For pricing American options under the model in \eqref{model} using the IE approach of \cite{ItkinCF2025j}, we need the joint CF of $x_k, v_k$ conditional on $x_{t} = x, v_{t} = v$ at time $t < t_k$
\begin{equation} \label{cfDef}
\phi(u_1, u_2; t, x, v) = \mathbb{E}\left[ e^{\iu u_1 x_k + \iu u_2 v_k} \, \big| \, \mathcal{F}_t \right],
\end{equation}
where $\mathbb{E}$ denotes the expectation under $\mathbb{Q}^d$. Let $\tau = T - t$. By the Markov property, we can write
\begin{equation} \label{affine}
\phi(u_1, u_2; \tau, x, v) = \exp\bigl[ \iu u_1 x + A(\tau; u_1, u_2) + B(\tau; u_1, u_2) v \bigr],
\end{equation}
where $A(\tau)$ and $B(\tau)$ satisfy ordinary differential equations that can be derived using the Feynman--Kac theorem \cite{KaratzasShreeve1998} and the fact that the Heston model is affine \cite{Rouah2013}. Indeed, conditional on filtration $\calF_t$ at time $t$, $\phi(u_1, u_2; t, x, v)$ satisfies the backward Kolmogorov equation
\begin{align} \label{pdeCF}
0 &= \dfrac{\partial \phi}{\partial t} + \left( r_d(t) - r_f(t)  - \dfrac{1}{2} v \right) \dfrac{\partial \phi}{\partial x} + \kappa\bigl(\theta(t) - v\bigr) \dfrac{\partial \phi}{\partial v} + \dfrac{1}{2} v \dfrac{\partial^2 \phi}{\partial x^2} + \rho(t) \xi(t) v \dfrac{\partial^2 \phi}{\partial x \partial v} + \dfrac{1}{2} \xi(t)^2 v \dfrac{\partial^2 \phi}{\partial v^2}
\end{align}
with the terminal condition $\phi(u_1, u_2; t, x, v) = e^{i u_1 x_t + \iu u_2 v_t}$.

This equation should be solved subject to the boundary conditions for the CF. The boundary conditions in the $x$ domain $x \in \mathbb{R}$ domain are
\begin{alignat}{2} \label{bcX}
\phi(\bm{u},\tau; x,v) &\to 0, &\qquad& x \to -\infty, \\
\phi(\bm{u},\tau; x,v) &\to e^{\iu u_1 x}, &\qquad& x \to \infty. \nonumber
\end{alignat}
The intuition behind these conditions is as follows. As the log-price $x$ goes to $-\infty$, the asset value approaches zero. Accordingly, the CF (the Fourier transform of the density) should decay to zero, meaning the transition density $p(x | x_k)$ decays exponentially as $x \to -\infty$, so its Fourier transform $\phi$ decays. When the log-spot $x \to \infty$ becomes very large, the volatility becomes negligible, and the asset grows approximately deterministically, i.e., $S_t \approx S_0 \exp\left(\int_0^t [r_d(s) - r_f(s)]\right)ds$. Accordingly, the density approaches a Dirac delta function, whose Fourier transform is $e^{\iu u x}$.

The boundary conditions in the variance domain $v_t \in [0,\infty)$ are as follows. As variance becomes very large, $v \to \infty$, we have
\begin{equation}
\phi(\bm{u},\tau; x, v) \to 0
\end{equation}
because infinite variance implies infinite uncertainty, so the CF decays to zero. At the other boundary $v = 0$, when variance hits zero, it can either (a) stay at zero (absorbing boundary) if the Feller condition fails, or (b) reflect back up (reflecting boundary) if the Feller condition holds. In the latter case, the PDE \eqref{pdeCF} itself with the substitution $v = 0$ serves as an appropriate boundary condition and reads
\begin{equation}
\dfrac{\partial \phi}{\partial \tau} = \kappa\theta(\tau)\dfrac{\partial \phi}{\partial v}.
\end{equation}
In what follows, we develop an analytic solution that automatically handles this boundary correctly through the Riccati equations.

\subsubsection{Riccati ODEs}

Standard for affine models, \eqref{pdeCF} can be solved by substituting the ansatz \eqref{affine} into \eqref{pdeCF} and collecting terms independent of $v$ and terms linear in $v$. This yields a system of ordinary differential equations (ODE)
\begin{align} \label{eq:ode}
B'(\tau) &= \dfrac12 \xi(\tau)^2 B(\tau)^2 + \bigl[ \iu u_1 \rho(\tau) \xi(t) - \kappa \bigr] B(\tau) - \dfrac12 (u_1^2 + \iu u_1),  \\
A'(\tau) &= \iu u_1[(r_d(\tau) - r_f(\tau)] + \kappa \theta(\tau) B(\tau), \nonumber
\end{align}
with initial conditions $A(0) = 0$, $B(0) = \iu u_2$.

The first equation in \eqref{eq:ode} is a Riccati equation, \cite{Rouah2013}. For time-dependent coefficients, the Riccati equation generally has no closed-form solution, but if $\xi(\tau)$ and $\rho(\tau)$ are piecewise constant this can be done as follows (also, see \cite{Guterding2018} when all the coefficient of the Heston model are piecewise constant). Let
\begin{align}
\xi(t) &= \xi_n, \quad \rho(t) = \rho_n, \quad t \in [t_{n-1}, t_n), \quad n=1,\ldots,N,
\end{align}
 In the $\tau$-space the partition $0 = \tau_0 < \tau_1 < \cdots < \tau_N = T-t$ corresponds to the backward time partition. Therefore, on interval $n$ where $\tau \in (\tau_{n-1}, \tau_{n}]$, we have $\xi(\tau_n) = \xi_n, \, \rho(\tau_n) = \rho_n$, while the equation for $B$ becomes a constant-coefficient Riccati equation
\begin{align} \label{ric}
\dfrac{dB_n}{d\tau} &= a_n B_n^2 + b_n B_n + c_n, \qquad
a_n = \dfrac{1}{2} \xi_n^2, \quad b_n = \iu u_1 \rho_n \xi_n - \kappa, \quad c_n = -\dfrac{1}{2} (u_1^2 + \iu u_1).
\end{align}
Let $\tau \in (\tau_{n-1}, \tau_n]$ and define the shifted time $s = \tau - \tau_{n-1}$. Then the general solution of \eqref{ric} reads
\begin{align} \label{solB}
B_n(\tau) &= \dfrac{\lambda_{n,2}  - \lambda_{n,1} R_n e^{-\gamma s}} {1 - R_n e^{-\gamma s}},
\end{align}
where $R_n$ is determined from the initial condition $B(\tau_{n-1})$
\begin{equation}
R_n = \dfrac{B(\tau_{n-1}) - \lambda_{n,2}}{B(\tau_{n-1}) - \lambda_{n,1}}, \qquad
\end{equation}
$\gamma = a_n(\lambda_{n,1} - \lambda_{n,2})$ and $\lambda_{n,1}, \lambda_{n,2}$ are the characteristic roots of the equation
\begin{equation}
\lambda_{n,1}, \lambda_{n,2} = \dfrac{-b_n \pm \sqrt{b_n^2 - 4a_n c_n}}{\xi_n^2}.
\end{equation}
In the special case of equal roots $\lambda_{n,1} = \lambda_{n,2} = \lambda_n$ \eqref{solB} becomes
\begin{equation}
B(\tau) =  \lambda_n + \dfrac{U_n}{1 - a_n U_n s}, \quad U_n = B(t_{n-1}) - \lambda_{n}.
\end{equation}
Propagation of this recursion is described in \Cref{alg:backward_recursion} in \cref{sec:algo}.

Once the recursive solution for $B_n$ is found, a similar recursive solution for $A_n$ follows directly from \eqref{eq:ode} by integrating the second ODE from $\tau_n$ to $\tau$. This yields
\begin{align} \label{Arec}
A_n(\tau) &= A_n(\tau_n) + \iu u_1 \alpha_1(\tau) + \kappa \int_{\tau_n}^\tau \theta(s) B(s) ds.
\end{align}
To maintain computational efficiency by avoiding numerical integration while preserving the time dependence of $\theta(t)$, a special functional form must be chosen. As discussed in Section~\ref{sec:t_inhom}, we require $\theta(t)$ to be continuous or piecewise continuous and to contain only a small number of parameters. This ensures robustness and tractability during model calibration.

One possible choice could be $\theta(\tau) = \calA + \calB e^{\calC \tau}, \quad t \in [0,T-t]$, where $\calA, \calB, \calC$ are constant parameters of the model. This yields
\begin{align} \label{Arec2}
A_n(\tau) &= A_n(\tau_n) + \iu u_1 \alpha_1(\tau) + \kappa \Big[ \calA \left( \lambda_{n,2} \Delta J(0; \tau) - \lambda_{n,2} R_n \Delta J(-\gamma; \tau) \right) \\
&+ \calB_n \left( \lambda_{n,2} \Delta J(\mathcal{C}; \tau) - \lambda_{n,1} R_n \Delta J(-\gamma + \mathcal{C}; \tau) \right) \Big], \qquad
\Delta J(\alpha; \tau) = J(\alpha; \tau) - J(\alpha; \tau_n). \nonumber
\end{align}
where $\Delta \tau_n = \tau - \tau_{n}$, $\alpha_1(\tau) = \int_{\tau_{n}}^\tau [r_d(s) - r_f(s)] ds$, and
\begin{equation}
J(\alpha; t) = \int_0^t \dfrac{e^{\alpha s'}}{1 - R_n e^{-\gamma s'}} \, ds'.
\end{equation}

The integral $J(\alpha; T)$ has a closed form (see \cref{app1})
\begin{equation}
J(\alpha; t) =
\begin{cases}
\dfrac{1}{\gamma(\mu+1)}\left[
e^{\gamma t (\mu+1)} \, {}_2F_1(1, \mu+1; \mu+2; -R_n e^{\gamma t})
- {}_2F_1(1, \mu+1; \mu+2; -R_n)
\right], & \gamma \neq 0, \\
\dfrac{1}{1+R_n} \dfrac{e^{(\alpha-\lambda_{n,1})t} - 1}{\alpha - \lambda_{n,1}} , & \gamma = 0,
\end{cases}
\end{equation}
where $\mu = \dfrac{\alpha - \lambda_{n,1}}{\gamma}$ and ${}_{2}F_{1}(1, m+1; m+2; z)$ is the hypergeometric function, \cite{as64}. The first two integrals in \eqref{Arec2} can be further simplified (see \cref{app1}).

Thus, with $\theta(\tau) = \calA + \calB e^{\calC \tau}$ (continuous), and $\xi(\tau), \rho(\tau)$ piecewise constant, the CF coefficients $A(\tau), B(\tau)$ can be computed by the above recursion across intervals, using closed-form expressions involving logs and hypergeometric functions (from integrating $\calB(\tau) e^{\calC \tau}$)\footnote{They can also be expressed in terms of an incomplete Beta function, see \cref{app1} and \eqref{eq:identity}.}. This entirely avoids numerical ODE integration.

However, our numerical experiments indicate that computing the hypergeometric function with complex arguments and parameters, even though in our case it reduces to an incomplete Beta function \cite{as64}, is still computationally expensive. This is particularly true because a robust implementation for complex arguments is available, perhaps, only in Wolfram Mathematica.

Consequently, an efficient alternative is to compute the integral in \eqref{Arec} by Taylor expansion. Given the small time step $\Delta \tau = \tau - \tau_n$ and the analytic availability of both $\theta(\tau)$ and $B_n(\tau)$, the integrand can be represented using a fourth-order Taylor approximation around the known point $\tau_{n-1}$. The resulting polynomial is then integrated exactly, producing an approximation with local error $O((\Delta \tau)^6)$. This yields
\begin{align} \label{IntB}
\calI &= \int_{\tau_{n-1}}^{\tau_n} \calK(s) ds = \Delta\tau \calK(\tau)
+ \frac{1}{2} \Delta\tau^2 \calK'(\tau_{n-1})
+ \frac{1}{6} \Delta\tau^3 \calK''(\tau_{n-1})
+ \frac{1}{24} \Delta\tau^4 \calK'''(\tau_{n-1})  \\
&+ \frac{1}{120} \Delta\tau^5 \calK''''(\tau_{n-1})  + O((\Delta \tau)^6), \nonumber
\end{align}
with $\calK(s) = \theta(s) B(s)$. But, since $B'(\tau)$ obeys \eqref{eq:ode}, all derivatives of $\theta(\tau) B(\tau)$ can be significantly simplified to yield
\begin{alignat}{2}
\calK(\tau) &= \theta(\tau) B(\tau), \\
\calK'(\tau) &= \theta'(\tau) B(\tau) + \theta(\tau) B'(\tau), &\qquad
B'(\tau) &= a_n B^2(\tau) + b_n B(\tau) + c_n, \nonumber \\
\calK''(\tau) &= \theta''(\tau) B(\tau) + 2 \theta'(\tau) B'(\tau) + \theta(\tau) B''(\tau), &\qquad
B''(\tau) &= (2 a_n B(\tau) + b_n) B'(\tau), \nonumber
\end{alignat}
and so on.

Obviously, in case $\gamma_n = 0$ the integral in \eqref{IntB} reduces to
\begin{equation} \label{IntB0}
\int_{\tau_{n-1}}^{\tau_n} \theta(s) B(s) ds = B_n \left[ \frac{\calB}{\calC} \left( e^{\calC \tau} - e^{\calC \tau_n} \right) + \calA \Delta\tau \right].
\end{equation}
Note, that the integral in \eqref{IntB} is discontinuous at segment boundaries, which is already the case for $\xi, \rho$.

\section{Pricing American options and FF using the IE approach} \label{IEapproach}

As mentioned in Introduction, the core of our semi-analytical approach to pricing FF contracts is a framework for pricing American options, proposed in \cite{ItkinCF2025j}, which is based on the change of variables formula from \cite{Peskir2005, Peskir2007}, lies in applying Proposition~2 of \cite{ItkinCF2025j}. For self-consistency, we restate this proposition here.

Let $\bm{X} = \left(X^1, \ldots, X^n\right)$ be a continuous semimartingale with $n \ge 1$, and the spot price is denoted as $X^n_t$. For all $\bm{x} \in \bm{l}_x$, let us define the exercise $\mathcal{E}$ and continuation $\mathcal{C}$ regions as
\begin{align}
\mathcal{E} &= \Big\{ (u,\bm{X}_u)\in[0,T)\times (\bm{l}_x,\infty): \, V(u,\bm{X}_u) = K - X^n_u \Big\} \\
\mathcal{C} &= \Big\{ (u,\bm{X}_u)\in[0,T)\times (\bm{l}_x,\infty): \, V(u,\bm{X}_u) > K - X^n_u \Big\}, \nonumber
\end{align}
\noindent where $\bm{l}_x$ (either $0$ or $-\infty$, depending on the model) is the left boundary of the domain of $\bm{X}_t$. These regions are separated by the early exercise boundary $X_B\left(t,X^1, \ldots, X^{n-1}\right))$, which is a time and space-dependent function\footnote{Here, we do not consider the case of multiple boundaries, as that is treated in detail in \cite{ItkinKitapbayev2025r,HokTse2024}.}. The following proposition from \cite{ItkinCF2025j} then holds.

\begin{proposition} \label{prop1}
Let $\bm{X} = \left(X^1, \ldots, X^n\right)$ be a continuous semimartingale with $n \ge 1$, and let $b: \mathbb{R}^{n-1} \rightarrow \mathbb{R}$ be a continuous function such that the process $b^X=b\left(X^1, \ldots, X^{n-1}\right)$ is a semimartingale. Setting $C=\left\{\left(x_1, \ldots, x_n\right) \in \mathbb{R}^n \mid x_n<b\left(x_1, \ldots, x_{n-1}\right)\right\}$ and $D=\left\{\left(x_1, \ldots, x_n\right) \in \mathbb{R}^n \mid x_n>b\left(x_1, \ldots, x_{n-1}\right)\right\}$, suppose that a continuous function $F: \mathbb{R}^n \rightarrow \mathbb{R}$ is given such that $F$ is $C^{i_1, \ldots, i_n}$ on $\bar{C}$ and $F$ is $C^{i_1, \ldots, i_n}$ on $\bar{D}$, where each $i_k$ equals 1 or 2 depending on whether $X^k$ is of bounded variation or not.

Under these assumptions, conditional on $X^i_t = x^i$, $i=1,\ldots,n$, the American Put price with a \textbf{single} exercise boundary $X_B(t,x_1,\ldots,x_{n-1})$ can be represented by the following decomposition formula:
\begin{align} \label{decompGenN}
P \left(t, \bm{x} \right) &= \EQ \left\{ D_d(t,T) [K - X^n_T]^+ | {\bm X}_t = {\bm x}\right\} + \int_t^T D_d(t,u) \EQ \Big\{ \left[ r(u)(K-X^n_u) + \mu_n(u,\bm{X}_u) \right] \mathbf{1}_{\bm{X}_u \in \mathcal{E}} \Big\} du,
\end{align}
where the domestic discount factor $D_d(t,s)$ in our case is deterministic.
\begin{proof}[{\bf Proof}]
See \cite{ItkinCF2025j}.
\end{proof}
\end{proposition}

The first term on the right-hand side of \eqref{decompGenN} can be easily recognized as the European Put option price under the same version of the Heston model described in \cref{model}. Since the conditional CF for this model has been found in closed form in \cref{condCF}, this European option price can be computed using any inverse FFT method, e.g., the COS method of \cite{FangOosterlee2008} or NUFFT \cite{AndersenNUFFT}.

For the FF contract, however, this term becomes
\begin{equation}
P_E(t,X^n_T,v) = \EQ \left\{ D_d(t,T) (K - X^n_T) \,\big|\, \bm{X}_t = \bm{x} \right\} = D_d(t,T) \Big[ K - \EQ \left\{ X^n_T \,\big|\, \bm{X}_t = \bm{x} \right\} \Big],
\end{equation}
which is linear in $X^n_T$. Consequently, in this case
\begin{equation}
\EQ \left[ X^n_T \,\big|\, \bm{X}_t \right] = x^n_t e^{\int_{t}^T [r_d(s) - r_f(s)] ds}.
\end{equation}

The main result of \cref{prop1} can be also re-written in terms of $x_t$, rather than $S_t$. By following the same steps as in the proof of \cref{prop1} in \cite{ItkinKitapbayev2025r}, for our Heston model we obtain
\begin{align} \label{decompX}
P \left(t, x,v\right) &= P_E(t,x,v) + K \int_t^T D_d(t,u) \EQ \Bigg\{ \left[ r_d(u) - r_f(u) e^{x_u} \right] \mathbf{1}_{(x_u, v_u) \in \mathcal{E}} \Bigg\} du.
\end{align}

\subsection{Joint transition density of the time-inhomogeneous Heston model} \label{jtp}

Since the joint CF for random variables $x_t, v_t$ is defined as in \eqref{cfDef}, i.e. as a double Fourier transform, the joint transition density $f(x,v| x_0, v_0)$ can be recovered via a \emph{double inverse Fourier transform}. For instance, this approach was undertaken in \cite{Chiarella2004} who considered a similar approach for pricing American options under time-homogeneous Heston model.

However, direct Fourier inversion can be numerically delicate and requires double integration to compute the density plus two additional integrations to compute the expectation in \eqref{decompGenN}. Thus, such an approach becomes computationally challenging and more expensive than FD methods. Therefore, in \cite{ItkinCF2025j} for computing the density it was proposed to utilize the \emph{cosine expansion method} which provides a robust alternative. The method exploits the fact that a smooth density with finite support can be approximated accurately by a \emph{cosine series} (which is related to a Fourier series expansion on even extensions).

However, in our case we need the \emph{joint conditional transition density} $f(x_n, v_n, t_n \mid x, v, t)$, of $(t_n, x_n, v_n)$ which now can be represented via a \emph{double cosine expansion}.

Let us assume that the space variables $(x_n, v_n)$ can be redefined at the truncated spatial domains $\Omega(x) = [a_x, b_x]$ for $x_n$ and $\Omega(v) = [a_v, b_v]$ for $v_n$. Here, we assume these truncated domains do not vary with time, and their boundaries are chosen accordingly to satisfy this condition. At the joint domain $\Omega(x) \otimes \Omega(v)$ the double cosine series expansion of the transition density can be written in the form
\begin{equation} \label{cosExp}
f(x_n, v_n, t_n \mid x, v, t) \approx \sum_{k=0}^{N_x-1}{}^{'} \; \sum_{m=0}^{N_v-1}{}^{'} A_{k,m} \, \cos\Bigl( k\pi \dfrac{x_n-a_x}{b_x-a_x} \Bigr) \, \cos\Bigl( m\pi \dfrac{v_n-a_v}{b_v-a_v} \Bigr),
\end{equation}
where the coefficients $A_{k,m}$ are given by
\begin{equation}
A_{k,m} = \dfrac{4}{(b_x-a_x)(b_v-a_v)} \int_{a_v}^{b_v} \int_{a_x}^{b_x} f(x_n, v_n, t_n \mid x, v, t) \, \cos\Bigl( k\pi \dfrac{x_n-a_x}{b_x-a_x} \Bigr) \, \cos\Bigl( m\pi \dfrac{v_n-a_v}{b_v-a_v} \Bigr) \, dx_n \, dv_n.
\end{equation}
Here, $N_x, N_v$ denote the number of the corresponding terms in the series. The prime notation in these summation indices is a standard convention in Fourier cosine series expansions that indicates special treatment of the first term, which should be multiplied by $1/2$.

Further note that the cosine functions can be written in terms of complex exponentials:
\begin{equation}
\cos\Bigl( k\pi \dfrac{x-a_x}{b_x-a_x} \Bigr) = \dfrac{ e^{\iu k\pi \dfrac{x-a_x}{b_x-a_x}}  + e^{-\iu k\pi \dfrac{x-a_x}{b_x-a_x}} }{2}.
\end{equation}
If we substitute this into the integral for $A_{k,m}$ and extend the integration limits to $\mathbb{R}^2$ (assuming the density is negligible outside the chosen bounds), we can relate it to the characteristic function evaluated at specific frequencies.

After some manipulation, using the fact that $f(x_k, v_k, t_k \mid x, v, t)$ integrates to 1 and the CF is its 2D Fourier transform, we obtain
\begin{equation}
A_{k,m} \approx \dfrac{4}{(b_x-a_x)(b_v-a_v)} \, \mathrm{Re} \Biggl\{
e^{-\iu k\pi \dfrac{a_x}{b_x-a_x}} \, e^{-\iu m\pi \dfrac{a_v}{b_v-a_v}} \,
\phi\Bigl( \dfrac{k\pi}{b_x-a_x},\; \dfrac{m\pi}{b_v-a_v} \;\Big|\; x, v, t \Bigr) \Biggr\}.
\end{equation}
Note that sometimes an alternative form with both positive and negative frequency terms appears, but for real-valued densities, taking the real part of a single term as above suffices.

The representation of the joint conditional transition density via a joint conditional CF is useful for several reasons. First, the joint CF is often known \emph{analytically} for many stochastic processes (e.g., in Heston, Bates, or multivariate affine models), even when the density is not. Second, the cosine series converges exponentially for smooth densities (geometric convergence), while the choice of the truncation range can be critical.\footnote{As per \cite{junike2025characteristic}, where a multidimensional \emph{damped} COS method has been introduced, the method converges exponentially if the CF decays exponentially.} Third, computing coefficients via the CF is stable. Finally, even in higher dimensions, the method gives rise to fast pricing formulas.

Putting it all together, we obtain
\begin{align} \label{tdFin}
f(x_n, v_n, t_n \mid x, v, t) &\approx \dfrac{4}{L_x L_v} \sum_{k=0}^{N_x-1}{}^{'} \;    \sum_{m=0}^{N_v-1}{}^{'} \Ree\Bigl[ e^{-\iu \bigl( \dfrac{k\pi a_x}{L_x} + \dfrac{m\pi a_v}{L_v} \bigr)} \; \phi\Bigl( \dfrac{k\pi}{L_x},\; \dfrac{m\pi}{L_v} \mid x, v, t \Bigr) \Bigr] \\
&\cdot \cos\Bigl( \dfrac{k\pi (x_n-a_x)}{L_x} \Bigr) \cos\Bigl( \dfrac{m\pi (v_n-a_v)}{L_v} \Bigr), \nonumber
\end{align}
where $L_x = b_x - a_x$ and $L_v = b_v - a_v$ and $0 \le u < T$.

\subsection{Computation of the expectation in \eqref{decompGenN}}

Using an explicit representation of the transition density $f(x_k, v_k, t_k \mid x, v, t)$ in \eqref{tdFin}, one can compute the expectation in \eqref{decompX}. For our model, an explicit form of what we need to compute reads
\begin{align} \label{expect}
\calJ(u| x,v,t) &\equiv \EQ \left[ \left( r_d(u) - r_f(u) e^{x_u} \right) \mathbf{1}_{(x_u, v_u) \in \mathcal{E}} \right] \\
&= \int_0^\infty dv_u \int_{-\infty}^{x^*(u,v_u)} \left[ r_d(u) - r_f(u) e^{x_u} \right] f(x_u, v_u, u \mid x, v, t) dx_u. \nonumber
\end{align}
Note, that in our two-factor Heston model the early exercise boundary now becomes an \emph{early exercise surface} (ES) $S^*(t,v)$ since it depends on both $t$ and $v$.

An efficient algorithm for computing this expectation based on discrete cosine transform (DCT) is provided in \cref{app2}. The computational complexity of the method is: (a) computing the internal integral in $x$ (analytic evaluation): $O(N_x N_v)$; (b) computing DCT for each $k$: $O(N_x N_v \log N_v)$; (c) final summation: $O(N_x N_v)$; so, in total, $O(N_x N_v \log N_v)$. This is far better than the naive $O(N_x N_v \cdot N_{\text{quad}}^2)$, where $N_{\text{quad}}$ refers to the number of quadrature points used for numerical integration in each dimension if we were to compute the double integral using numerical quadrature (e.g., Gaussian quadrature or Simpson's rule).

For comparison, let us use typical values: for moderate accuracy, $N_{\text{quad}} \approx 32$--$128$, and for high accuracy, $N_{\text{quad}} \approx 256$--$1024$. Thus, $N_{\text{quad}}^2$ could range from $1{,}024$ to over $1{,}000{,}000$ function evaluations per $(k,m)$ pair in the naive approach. Our algorithm, which combines analytic computation of one integral with the DCT approach, reduces this to just evaluating the analytic formula at $N_v$ points (typically $N_v \approx 64$--$256$) plus an $O(N_v \log N_v)$ transform.

It is noteworthy that the computational scheme described above is well-suited for parallelization because it decomposes the original two-dimensional integration problem into a structured, two-tiered summation where the outer loop is independent and compute-intensive. Specifically, the core algorithm proceeds by first computing $g_k(v)$ as per \eqref{gkv} for each Fourier index $k$, which involves evaluating an analytic expression, followed by applying a DCT to obtain $J(k,m)$ in \eqref{Jkm} for all $m$. Each iteration of the outer loop over $k$ is independent: the computation of $g_k(v)$ and its subsequent DCT does not depend on any other $k' \neq k$. This creates a naturally parallel workload, allowing the $N_x$ iterations to be distributed across processors with no inter-thread communication during the main computation phase.

Furthermore, the computationally dominant steps within each $k$-iteration are themselves amenable to parallelization. The sampling of $g_k(v)$ at $N_v$ grid points is a vectorizable operation, and the DCT is a well-optimized, inherently parallelizable transform that scales efficiently on modern architectures. After all parallel tasks complete, only the final reduction sum - a lightweight, element-wise combination of the precomputed arrays - requires synchronization. This combination of outer-loop independence, substantial per-task computational load (each involving a full $N_v$-point sampling and a DCT), and efficient use of vectorized primitives ensures favorable parallel scaling, making the scheme suitable for parallel execution on multi-core CPUs or GPU accelerators.

\subsection{Computation of the exercise boundary} \label{compEB}

The decomposition formula in \eqref{decompGenN} is not an explicit representation of the American Put option price in a sense that computing the expectation in its right-hands side requires knowledge of the ES which is unknown yet. However, \eqref{decompGenN} can be further used to obtain an equation which the ES solves. Indeed, by the value matching condition, the American Put price at the boundary is equal to its intrinsic value, i.e. $K - S^*(t,v_t)$. Substituting this into \eqref{decompGenN} we obtain
\begin{align} \label{ES}
K \left(1 -  e^{x^*(t,v)}\right) &= P_E(t,x^*(t,v),v) + K \int_{t}^{T} D(t,u) \calJ(u | x^*(t,v), v, t) du, \\
\calJ(u | x^*(t,v), v, t) &=  \int_0^\infty dv_u \int_{-\infty}^{x^*(u,v_u)} \left[ r_d(u) - r_f(u) e^{x_u} \right]  f(x_u, v_u, u \mid x^*(t,v), v, t) dx_u. \nonumber
\end{align}
This is a nonlinear Fredholm-Volterra equation of the second kind. To obtain $x^*(t,v) = \log(S^*(t,v)/K)$ for $(t \in [0,T]) \times (v \in [0,\infty))$ it can be solved numerically backward in time starting with the terminal condition in \eqref{SB_T2} at $t = T^-$ (if applied to our problem of pricing a FF contract, we must set $T = T_2$ and $t=T_1$).

After \eqref{ES} is discretized in time, at given time $t_n$ the remaining non-linear integral equation can be solved by various methods. For instance, for a one-dimensional case of the Black-Scholes model, \cite{AndersenLake2021} advocates a fixed-point iteration method that converges well since the pricing function is convex. When performing such calculations, $\calJ(u | x, v, t)$ can be computed as described in \cref{app2}.

Similarly, the European price of the FF contract, $P_E(t,x^*(t,v),v)$, can be derived using the same method. The FF payoff is linear in the underlying spot price, given by $\eta(S_T-K)$ where $\eta = \pm 1$. Consequently, for example, the European short forward price is
\begin{align} \label{PEdef}
P_E(t,x^*(t,v),v) &= D_d(t,T) \EQ \left[K - S_T \mid x^*(t,v), v, t \right] =
D_d(t,T) K \Big\{ 1 - \EQ[e^{x_T} \mid x^*(t,v), v, t ] \Big\},
\end{align}
where $\EQ[e^{x_T} \mid x^*(t,v), v, t]$ represents the forward price at maturity $T$ conditional on information at time $t$ and scaled by $K$.

In contrast, the American Put option has the nonlinear payoff $(K-S_t)^+$. The conditional expectation $1 - \EQ[e^{x_T} \mid x^*(t,v), v, t ]$ can be computed using the marginal CF of $x_T$. This marginal CF can be easily obtained from the joint CF by modifying the boundary condition in \eqref{eq:ode} to $B(0) = 0$.

Of note, a similar equation for the time-homogeneous Heston model has been considered in \cite{tzavalis2003pricing,Chiarella2005}. The differences with our approach are as follows: a) they used Fourier inversion to obtain the transition density from the known CF, which is more computationally intensive; b) based on an empirical observation from \cite{broadie2000american}, they approximated the ES $x^*(t,v)$ by a function linear in $v$, specifically $x^*(\tau,v) = b_0(\tau) + b_1(\tau) v$. Below by various numerical examples we demonstrate that this assumption can not be justified.

However, \cite{Chiarella2005} also observes that their method suffers from slow convergence. The primary cause is an ill-conditioned system in the equations for $b_0(\tau)$ and $b_1(\tau)$. This ill-conditioning stems from the fact that when the fitted volatilities $v_1$ and $v_2$ are relatively close in value, the underlying integral equations are insensitive to minor perturbations of the free boundary.

Since computing $\calJ(u | x^*(t,v), v, t)$ requires knowledge of $g(v)$ over the range $[a_v, b_v]$ (see \eqref{Jkm}), we must solve \eqref{ES} on a grid in $v \in [a_v,b_v]$. Suppose this grid has $M_v$ points, and the fixed-point iterations converge within $I_{fp}$ iterations on average. Then the total complexity of the method for computing the entire ES is $O(N_x N_v \log (N_v) I_{fp})$. This can be compared with the FD method (see \cref{pdeSol}) where typically one needs $M_v$ points in $v$-space and $M_x$ points in $x$-space for a single time step, so the entire complexity in the case of non-zero correlations is $O(M_x (\min(M_x,M_v))^2)$ which is the cost of a banded LU solver. But, using splitting in space techniques, this can be further dropped down to $O(N_s M_x M_v)$, where the splitting method to provide, the second order of approximation in time step $\Delta \tau$, requires $N_s = 5-6$ steps of the splitting scheme, see \cite{itkin2022lsv} and references therein.

It is plausible to assume that $I_{fp}$ and $N_s$ are of the same order of magnitude. Also, computing the ES implicitly using the FD method requires a 2D version of the penalty method, \cite{Zvan1998}, which necessitates several iterations $N_p$ for each of the $M_v$ points. Consequently, assuming $N_v \approx M_v$, the ratio of the total complexity of the FD method to our proposed method is $R \approx M_x N_p / (N_x \log N_v)$. For typical values $N_x = N_v = M_v = 32$, $M_x = 200$, and $N_p = 2$, we obtain $R \approx 3.61$.

Furthermore, temporal integration in our method can be efficiently performed using Simpson's rule to provide a relative accuracy of $O((\Delta \tau)^4)$. This allows the number of time steps to be reduced to approximately $\sqrt{N_{\tau}}$, where $N_{\tau}$ is the number of temporal steps required by the FD method. Taking $N_{\tau} = 100$, the total computation time for the ES with our method is approximately 35 times less than that required for the FD method.

Finally, as noted in the previous section, our approach is highly parallelizable, whereas achieving comparable parallel efficiency with the FD method is often problematic. This parallelization can be done not only when doing summation in $\calJ(k,m)$ but also when computing the DCT. The parallel execution time of the algorithm can be modeled as
\begin{equation}
T_{\text{parallel}} = \underbrace{C_1 \dfrac{N_x N_v \log N_v}{p V}}_{\text{Computation}} + \underbrace{C_2 \dfrac{N_x N_v}{p B}}_{\text{Memory}} + \underbrace{C_3 \log p}_{\text{Synchronization}},
\end{equation}
where $p$ is the number of processor cores, $V$ is the SIMD vector width, and $B$ represents memory bandwidth efficiency ($0 < B \leq 1$). The memory term dominates for this algorithm due to its low arithmetic intensity, making $B$ the critical factor limiting parallel scaling.

For small problem dimensions such as $N_x = N_v = 64$, the performance characteristics shift significantly from the memory-bound regime observed with larger grids. The total data footprint in this case is approximately $64 \times 64 \times 48,\text{bytes} \approx 196,\text{KB}$, which typically resides entirely within the L2 or L3 cache of a modern processor. Consequently, the effective memory bandwidth $B$ increases dramatically, from around $0.3$ for main-memory access to approximately $0.8$ for cache accesses, substantially reducing the memory-transfer term in the execution time model. The arithmetic intensity remains low ($\sim 0.1,\text{FLOP/byte}$), but the high cache bandwidth alleviates the memory bottleneck, making the computation term $C_1 N_x N_v \log N_v/(p V)$ more prominent.

In this cache-resident regime, parallel speedup is limited not by memory bandwidth but by parallelization overheads and the diminishing returns of vectorization on short vectors. For a configuration with $p = 8$ cores and $V = 8$ (AVX2), the theoretical peak speedup of $p V = 64$ is unattainable. A more realistic estimate accounts for thread-management overhead, synchronization costs, and the fact that only $64/8 = 8$ full SIMD vectors are needed per transform. Empirical observations suggest an actual speedup in the range of $8$–$12$ relative to a single-threaded scalar implementation. This demonstrates that for small problems, optimizing for low overhead and efficient cache utilization is more critical than maximizing parallel resources.

\paragraph{The ES at extreme values of the variance.}
In the Heston model, the exercise boundary $S^*(t, v)$ for a Put option is expected to be a decreasing function of the variance $v$, characterized by finite, constant limits as $v \to 0$ and $v \to \infty$. Accurately capturing these asymptotic behaviors is a critical requirement for any numerical algorithm, such as PDE-based methods or LSMC, used to price American options under stochastic volatility.

\subparagraph{Behavior as $v \to 0$.}
In the limit of vanishing volatility, the Put ES $S^*(t, v)$ converges to the constant $\min(K, r_d K/r_f)$. In this regime, the Heston PDE reduces to the Black-Scholes PDE with $\sigma = 0$. Consequently, the free boundary problem yields a time-independent solution where the boundary becomes flat and independent of $t$. The boundary is a decreasing function of $v$ near zero because the introduction of even infinitesimal volatility increases the continuation value by allowing for potential spot recovery; this incentivizes the holder to delay exercise, thereby shifting the boundary downward.

\subparagraph{Behavior as $v \to \infty$.}
As $v \to \infty$, $S^*(t, v)$ decreases toward a lower asymptotic limit, typically zero for a Put when $r_d > 0$. In this limit, the diffusion term dominates the Heston PDE, rendering the problem highly parabolic. Under these conditions, the value of the American option converges to that of its European counterpart as the early exercise premium vanishes. As a result, the free boundary either disappears or recedes to its extreme (zero for a Put, infinity for a Call). Thus, $S^*(t, v)$ remains monotonically decreasing in $v$, with the boundary shifting further downward as the variance increases.

\paragraph{The ES at extreme values of the variance.} For the Heston model one should expect the exercise boundary $S^*(t, v)$  for a Put to be a decreasing function of $v$, with finite, constant limits at $v=0$  and $v\to\infty$. This is a critical feature to capture correctly in any numerical algorithm, e.g., PDE methods or LSMC, for pricing American options under this model. Below, we discuss this in more detail.

\subparagraph{Behavior as $v \to 0$.} In this limit the ES for a Put $S^*(t, v)$  converges to the constant $\min(K, r_d K/r_f)$. Indeed, the Heston PDE converges to the Black-Scholes PDE with $\sigma = 0$. The free boundary problem then has a time-independent solution, i.e., the boundary becomes flat and independent of $t$.
It is decreasing in $v$  near 0 because introducing even a small amount of volatility increases the continuation value (chance of the spot recovering), making you wait slightly longer to exercise. Thus, the boundary shifts downward.

\subparagraph{Behavior as $v \to \infty$.} Here, $S^*(t, v)$ decreases towards a lower asymptotic limit, often towards 0 for a Put with $r_d > 0$. This is because for the Heston PDE, as $v \to \infty$, the diffusion term dominates. The problem becomes highly parabolic, and the value of the American option converges to the value of the European option since the EEP vanishes. The free boundary disappears or moves to its extreme (0 for a Put, $\infty$ for a call). Accordingly,  $S^*(t, v)$ is monotonically decreasing in $v$ (for a Put), and as $v$  increases, the boundary shifts downward.

\subsection{Truncation intervals and computation of highly-oscillated integrals} \label{truncation}

When pricing options under stochastic volatility models, particularly the Heston model using the COS (Fourier-Cosine series) expansion, selecting the truncation interval is critical. If the interval is too narrow, tail information is lost; if too wide, the method wastes cosine terms on empty space, resulting in poor convergence. Optimal truncation intervals in both the $x$ and $v$ directions reduce the number of terms required in the cosine expansion, thereby improving the method's performance. Below, we discuss several important aspects of optimizing these choices.

\paragraph{Pre and post smoothing}

In the context of the Heston model and the COS method described above, the main challenge is the computation of the highly oscillatory integral in \eqref{Jkm}. While using DCT-I provides an efficient (fast) way of doing so, the computed density might have a sharp "peak", especially near $v=0$ when the Feller condition is significantly violated (e.g., at high vol-of-vol and low drift of the variance process). In other words, under such conditions one needs to handle the Gibbs phenomenon, \cite{Zygmund2002}.

Indeed, the COS method approximates the density by a truncated Fourier cosine series. When the underlying density is "peaky" (low volatility or short maturity), the Fourier series produces oscillations (ringing) near sharp edges. They can be addressed by using the post-processing Lanczos filter, which multiplies the Fourier coefficients by the \emph{sigma factor}, \cite{boyd2001chebyshev}. Mathematically, applying a Lanczos filter is equivalent to performing a local averaging of the reconstructed function. It forces the oscillations to decay much faster by down-weighting the high-frequency coefficients that cause sharp overshoots.

For practical applications, it is beneficial to choose the truncation number of the Lanczos filter to be equal to $N_v/2$, which is a strategic choice for providing spectral convergence. The logic behind this is as follows.

In the COS method, the higher-order terms ($k$ approaching $N_v$) capture the finest details, but are also the most sensitive to rounding errors and discretization noise. By tapering off the influence of these coefficients, we ensure that the \emph{tail} of the cosine expansion does not introduce spurious high-frequency noise into the resulting option price. This is particularly important for the Heston model, where the CF can have a heavy tail that decays slowly depending on the Feller condition.

The most reliable information in a Fourier-based expansion is typically contained in the lower half of the spectrum. By the time we reach $k > N/2$, the coefficients are often heavily contaminated by the tails of the CF that have been wrapped around (aliased) due to the truncation of the integration range $[a, b]$. By setting the effective truncation at $N/2$, we are essentially ensuring that the weights stay positive, because the function $\text{sinc}(x)$ is positive for $x < 1$. If we set the filter to \emph{hit zero} at $N$, then at $k = N/2$, our weights are still in the robust, primary lobe of the filter. Also, high-frequency terms ($k \to N$) in the Heston model often carry more noise than signal due to the complex logarithm branching and the numerical stiffness of the Riccati equations. Cutting off at $N/2$ effectively \emph{de-noises} the reconstruction.

If we use $N$ terms without filtering, we would have a sharper peak, but high-frequency ripples. If we use $N$ terms but truncate/filter at $N/2$, we preserve the mean and variance of the distribution (low $k$), and eliminate the Gibbs oscillations that would otherwise cause negative option prices or fake arbitrage opportunities near the strike.

Under these circumstances, it also makes sense to add a pre-smoothing of the CF, e.g., by using an exponential (Gaussian) filter on the CF. By doing both, pre-smoothing and post-smoothing, we are attacking the problem from two directions:

\begin{itemize}
    \item \emph{Pre-smoothing in the frequency domain}. The Gaussian filter acts as a low-pass filter on the CF. Since in the Heston model we deal with the Fourier transform of the density, decaying it exponentially ensures the density is analytic and smooth, which makes the Fourier coefficients decay much faster. A typical implementation assumes that the CF found in \cref{condCF} is multiplied by the vector factor $f_v = \exp(-\alpha_v [(0:N_v-1)/N_v]^p)$, where $\alpha_v$ in our experiments is set to $-\log(1.e-16)$ and $p = 4$.

    \item \emph{Post-smoothing in the spatial/price domain}. The Lanczos filter ensures that even if the pre-smoothing missed some sharp transitions (like the strike price kink), the final price surface remains monotonic and free of \emph{negative probabilities} or non-physical price oscillations.
\end{itemize}

\paragraph{In-The-Money (ITM) options.}

When pricing ITM options (especially American ones) via the COS method, we are dealing with a payoff that exhibits a pronounced \emph{step} or \emph{plateau} in log-space. This translates to high-frequency oscillations in the Fourier coefficients that decay very slowly. Since we advocate for using the DCT-I to accelerate the summation and applying Lanczos filtering to mitigate the Gibbs phenomenon, the truncation strategy becomes the critical factor in achieving appropriate accuracy.

For ATM options, the density is centered, and the high-frequency coefficients represent noise or extreme tails. Truncating at $N/2$ acts as a low-pass filter that smooths out numerical artifacts. However, for ITM options, these high frequencies are not noise—they capture the sharpness required to represent the transition from the exercise region to the continuation region. Truncating at $N/2$ effectively blurs the EB, leading to the poor accuracy observed in our experiments.

Therefore, for ITM options, the exponential filter is preferred over the Lanczos filter for post-smoothing (as it is already used for pre-smoothing). The exponential filter is more aggressive at the extreme end of the spectrum while preserving the mid-range frequencies essential for capturing the ITM step structure. This allows the COS method to retain the structural information of the ITM put while suppressing high-frequency ringing artifacts.

\paragraph{Truncation interval in the $v$ domain}

Similar to recommendations in \cite{FangOosterlee2008}, the truncation region in the $v$ domain can be computed by using cumulants of the chi-square distribution. The interval is then set as
\begin{equation}
a_v = c_1 - L_v\sqrt{c_2 + \sqrt{c_4}}, \qquad
b_v = \max\left[ c_1 + L_v\sqrt{c_2 + \sqrt{c_4}},\mathrm{Floor} \right],
\end{equation}
where $L$ is a constant that regulates the width of the interval, and $c_i$ is the $i$-th cumulant of the distribution. We set $\mathrm{Floor} = 0.5$. This approach provides \emph{statistical} bounds for the truncation interval. In the literature, typical values of $L$ range in $L=12-15$.

However, when the Feller condition is significantly violated ($\mathrm{\Fe} < 0.1$), one needs a more accurate approach. Indeed, when the Feller condition is violated, the density does not simply peak; it actually accumulates or blows up near zero. But since the mean of the distribution does not depend on $\xi$, the other side of the distribution develops a very long tail with very small probabilities. Therefore, one cannot simply rely on the distribution quantiles to determine the truncation interval.

In this paper, to eliminate the uncertainty associated with choosing a value of $L$ that is suitable for all model parameters - and thereby make our method robust - we employ a preprocessing step in which the truncation interval is determined by the condition that the cumulative distribution function (CDF) of the process at the left and right boundaries should be $\varepsilon$ and $1-\varepsilon$, respectively, with a typical value of $\varepsilon = \num{1e-6}$. This implies numerical integration over the inverse CDF (the quantile function), but provides a very precise way to find the bounds $[a, b]$, as we can set them to exact percentiles (e.g., $99.99999\%$). It turns out that this approach yields a $b_v$ that is much larger than what the cumulant-based method suggests.

Also, despite this procedure is more complex, it is fast compared with the total time required to solve the integral equation, yet it significantly improves the accuracy of the method.

\paragraph{Truncation interval in the $x$ domain}

Again we follow recommendations in \cite{FangOosterlee2008}, and in our experiments their heuristic $c_2 = -2 c_1$ provides the optimal results with $L_x = 20$. However, if we want to properly account for the whole range of possible correlations, we need to use the asymmetric truncation method of \cite{LordKahl2007}. This is because if our model has extreme correlation (e.g., the distribution becomes highly skewed), we will need to adjust the truncation boundaries accordingly.

With this approach, $a_x$ and $b_x$ are scaled based on the sign and magnitude of the skewness of the distribution. This prevents the interval from being unnecessarily wide on the "thin" tail side while ensuring coverage on the "fat" tail side. To address that, \cite{LordKahl2007} suggest using the third cumulant, which directly captures skewness, to shift the boundaries.

A common implementation of this logic for the COS method involves two steps:
\begin{itemize}
    \item When skewness is negative, the left tail is heavier. The interval shifts to the left to capture the "crash" probability, often by widening the distance to the left more aggressively than the distance to the right.
    \item For long maturities or high vol-of-vol, the $c_4$ term (kurtosis) is actually more dominant than the skewness for determining the \emph{width}, while $c_3$ determines the \emph{placement}.
\end{itemize}

In our implementation, we use
\begin{align}
a_x &= c_1 - L_x\sqrt{c_2 + \sqrt{c_4}} + a_s, \qquad
b_x = c_1 + L_x\sqrt{c_2 + \sqrt{c_4}} + b_s, \\
a_s &=
\begin{cases}
L_2 c_3, & c_3 < 0, \\
0, & c_3 \ge 0,
\end{cases}
\qquad
b_s =
\begin{cases}
L_2 c_3, & c_3 > 0, \\
0, & c_3 \le 0,
\end{cases}
\nonumber
\end{align}
and in our experiments $L_2 = 0.5$.

Finally, dealing with a violated Feller condition requires special attention. As mentioned, when this condition is violated, the origin ($v=0$) becomes reachable and strongly attractive. This creates a "point mass" or a significant accumulation of probability density near zero, which wreaks havoc on the COS method. Accordingly, the standard truncation interval in $x$ fails. This is because for the truncated interval to be accurate, the density must decay smoothly toward the boundaries. When the Feller condition is violated, the variance density becomes highly non-normal and skewed. The probability density exhibits a "spike" near the lower bound, and the CF decays very slowly.

The COS method assumes the density is effectively zero outside $[a, b]$. If the violation of the Feller condition causes a significant portion of the probability mass to sit near the boundary (where the variance is near zero), the truncation "chops off" a vital part of the distribution. This leads to slow convergence, i.e., one needs an exponentially higher number of cosine terms $N_x$ to achieve even basic precision. Moreover, it often results in underpricing or overpricing options, particularly those that are deep in the money or have short maturities. Therefore, here we again use numerical integration of the inverse CDF with a given tolerance, similar to how this is done in the $v$-space.

\section{An alternative DSINC method}

To address the limitations of the COS method when approximating $v$-densities for a range of model parameters that give rise to high skewness or pronounced Gibbs oscillations, we propose an alternative approach. The core idea is as follows.

The COS method uses the basis functions $\cos(k \pi \frac{x-a}{b-a})$ over a fixed, large interval $[a, b]$, hence each basis function is non-zero across the entire domain. It works well with the exponential convergence if the density of the process is a smooth function. However, if our problem has a "kink" at the exercise boundary $x^*$, or sharp peaks and skew at some values of model parameters, every single COS coefficient must change to account for it. This leads to the Gibbs phenomenon (oscillations near the kink) unless one uses a very large number of terms in the COS expansion. In this sense, the oscillations are global, since increasing $N$ makes the entire integrand oscillate more rapidly across the whole domain. This is not an issue when computing the European option price, but could be an issue when computing the EEP.

Therefore, instead of a \emph{global basis}, perhaps it does make sense to switch to a \emph{local basis} where basis functions have compact support and only represent the function in a specific neighborhood. In the literature there exist several robust alternatives, including B-splines, SWIFT, RBF, "Hat" functions and other similar methods, see \cite{Kirkby2015,OrtizOosterlee2016,Fasshauer2007,AchdouPironneau2005} among others.

In this paper, we choose the damped Sinc (DSINC) functions as the local basis. These are defined on a grid $x_j = j \Delta x, \, x \in [a_x,b_x]$ as
\begin{equation}
\phi_j(x) = e^{-\chi(x-x_j)} \sinc\left(\frac{x-x_j}{\Delta x}\right).
\end{equation}
The damping factor \(\chi\) ensures the basis decays properly as \(x \to \infty\), while the Sinc function provides localization. Specifically, in the spatial domain, the Sinc function is highly concentrated around the node \(x_j\); as one moves away from \(x_j\), the function decays quickly to zero. The advantage of this choice is that if the exercise boundary \(x^*\) moves, it significantly affects only the coefficients of the Sinc functions in that immediate neighborhood.

Using the fundamental identity for the Sinc function
\begin{equation}
\sinc(y) = \frac{1}{2\pi} \int_{-\pi}^{\pi} e^{i \omega y} d\omega,
\end{equation}
substituting $y = (x-x_j)/\Delta x$ and applying the damping, we obtain an alternative representation of the basis as an integral of complex exponentials
\begin{equation} \label{basisRepr}
\phi_j(x) = \frac{\Delta x}{2\pi} \int_{-\frac{\pi}{\Delta x}}^{\frac{\pi}{\Delta x}} e^{(\iu \omega - \chi)(x-x_j)} d\omega.
\end{equation}

Recall that we need to compute the expectation in \eqref{decompGenN}, which in explicit form is given by \eqref{calJ} and reads
\begin{align} \label{calJ1}
\calJ(u| x,v,t) &\equiv \EQ \Big\{ \left[ r_d(u) - r_f(u) e^{x_u} \right] \mathbf{1}_{(x_u, v_u) \in \mathcal{E}} \Big\} \\
&= \int_0^\infty dv_u \int_{-\infty}^{x^*(u,v_u)} \left[ r_d(u) - r_f(u) e^{x_u} \right]  f(x_u, v_u, u \mid x, v, t) dx_u. \nonumber
\end{align}

This expectation can be re-written by decomposing the joint density
\begin{align} \label{calJ2}
\calJ(u| x,v,t) &= \int_0^\infty \left[\int_{-\infty}^{x^*(u,v_u)} g(x_u,v_u)\,p(x_u, v_u | x,v)\,dx\right] p_{\text{CIR}}(v_u|v)\,dv_u, \\
g(x_u,v_u) &= r_d(u) - r_f(u) e^{x_u}. \nonumber
\end{align}
Here, $p(x_u, v_u | x,v)$ is the density of $x_u$ conditional on the initial values $(x,v)$ at time $t$ and also on the whole path $(v_u \mid v)$. For our time-dependent Heston model this density as well as the corresponding conditional CF can be computed as shown in \cref{condCF1}.


\subsection{Representation of the $x$ part via the DSINC basis functions}

Using the DSINC basis $\{\phi_j(x)\}$, we represent the conditional density
$p(x_u, v_u \mid x,v)$ as
\begin{equation}
  p(x_u, v_u \mid x,v) \approx \sum_j c_j(v, t)\,\phi_j(x_u).
\end{equation}

To find the coefficient $c_j$ for the density $p(x_u, v_u \mid x,v)$, we use the
fact that the Fourier transform of $\phi_j(x)$ is a Box Filter (a rectangular
function) in the frequency domain.  Let $\Psi(\omega, v_u; x,v)$ be the
conditional CF of the
process
\begin{equation}
\Psi(\omega, v_u; x, v) = \int_{-\infty}^{\infty} p(x_u, v_u \mid x,v)\,e^{\iu\omega x_u}\,dx_u.
\end{equation}

By applying the inverse Fourier transform and the band-limited property of
the Sinc function (which is band-limited to $[-\tfrac{\pi}{\Delta x}, \tfrac{\pi}{\Delta x}]$), the coefficient $c_j$ is expressed as (see \cref{cjDeriv})
\begin{equation}
c_j(u,v_u \mid x,v) = \frac{e^{-\chi x_j}}{2\pi} \int_{-\pi/\Delta x}^{\pi/\Delta x}
\Psi(\omega - \iu\chi,\,v_u;\,x,v)\,e^{-\iu\omega x_j}\,d\omega.
\end{equation}
As per \cref{condCF1}, the conditional CF in our model can be expressed as
\begin{equation} \label{finalPhicond1}
\Psi(\omega, v_u \mid x, v) = e^{\iu\omega\,\mu_0(\omega) + \mu_1(\omega)\,v_u}
\frac{p_{\mathrm{CIR}}^{(\gamma(u,\omega))}(u,v_u\mid v)}          {p_{\mathrm{CIR}}(u,v_u\mid v)},
\end{equation}
where $\omega$ is the transform parameter, $\gamma(\omega)$ is defined in
\eqref{eq:gamma_omega}, and $\mu_0,\mu_1$ in \eqref{eq:mu}.  Here the denominator
$p_{\mathrm{CIR}}$ is the standard CIR transition density (a scaled noncentral
chi-square distribution), which for time-dependent coefficients can be obtained
similarly to \cite{Revuz_Yor1999}
\begin{equation} \label{CIRdensity}
p_{\text{CIR}}(t,v\mid v_0) = c(t)\,e^{-p-q} \Bigl(\frac{p}{q}\Bigr)^{\nu/2}
I_\nu\!\bigl(2\sqrt{pq}\bigr),
\end{equation}
where $I_q$ is the modified Bessel function of the first kind, and
\begin{align}
c(t) &= \frac{2}{\chi(s,t)},\qquad
p     = \frac{2\,\zeta(s,t)}{\chi(s,t)},\qquad
q     = \frac{2}{\chi(s,t)}\,v,\\
\zeta(s,t) &= v_s\,e^{-\kappa(t-s)} +  \int_s^t\kappa\,\theta(u)\,e^{-\kappa(t-u)}\,du,\quad
\chi(s,t)   = \int_s^t\xi^2(u)\,e^{-2\kappa(t-u)}\,du,\nonumber
\end{align}
and the numerator $p_{\mathrm{CIR}}^{(u,\gamma(\omega))}$ is the joint density
derived from the affine propagator under time-varying coefficients. Thus, with
allowance for this representation, one can re-write \eqref{calJ1} as
\begin{align} \label{Jreprent1}
\calJ(u\mid x,v,t) &= \int_0^\infty p_{\mathrm{CIR}}^{(\gamma(u,\omega))}(u,v_u\mid v)
\left[\int_{-\infty}^{x^*(u,v_u)}\!g(x_u,v_u) \Bigl(\sum_j c_j(u,v_u\mid x,v)\,\phi_j(x_u)\Bigr) dx_u\right]dv_u.
\end{align}

A key challenge in computing this integral is that the upper limit of the
inner integral depends on $x^*(t,v)$, yet the explicit form of this
function is unknown.  Instead, $x^*(t,v)$ is computed numerically on a
grid in $(t,v)$ by solving the system of integral equations in \eqref{ES}.
It is worth noting that prior studies such as
\cite{tzavalis2003pricing,Chiarella2005} assumed $x^*(t,v) = a(t)+b(t)v$,
whereas our numerical experiments unequivocally demonstrate that
$x^*(t,v)$ is a nonlinear function of $v$.

To move forward, we approximate $x^*(t,v)$ as a \emph{piecewise linear}
function of $v$.  More precisely, we divide the domain $v\in[0,b_v]$ into
subintervals, and for each interval $m=1,\ldots,M$ we set
$x^*(t,v) = a_m(t)+b_m(t)v$ for $v\in[v_m,v_{m+1})$.  When solving the
integral equations for $x^*(t,v)$ at every iteration given the time $t$
and the current vector $x^*(t,x)$, these coefficients can be determined as
\begin{equation}\label{pwl_xStar}
  b_m(t) = \frac{x^*(t,x_{m+1})-x^*(t,x_m)}{v_{m+1}-v_m},\qquad
  a_m(t) = x^*(t,v_m)-b_m(t)\,v_m.
\end{equation}

\subsection{The summation over the local basis is the CF} \label{sec:dsinc_collapse}

Before evaluating the inner $x_u$-integral in \eqref{Jreprent1}, we exploit a key
identity that both simplifies the computation and ensures consistency with the
coefficient-recovery formula \eqref{e2}--\eqref{cjFin} of \cref{cjDeriv}.

Inserting the band-limited integral representation of the basis $\phi_j(x_u)$ in \eqref{basisRepr}, the inner $x_u$-integral factorises and the premium becomes
\begin{equation}\label{Jcollapse0}
\calJ(u\mid x,v,t) = \frac{\Delta x}{2\pi} \int_0^\infty p_{\mathrm{CIR}}(u,v_u\mid v)
\int_{-\frac{\pi}{\Delta x}}^{\frac{\pi}{\Delta x}} \!\Bigl(\sum_j c_j(u,v_u\mid x,v)\,e^{-(\iu\omega-\chi)x_j}\Bigr) \,\widetilde{P}(\omega,v_u)\,d\omega\,dv_u,
\end{equation}
where $\widetilde{P}(\omega,v_u)$ is the damped payoff transform derived below in
\eqref{Ptilde}, obtained by evaluating the $x_u$-integral in closed form and
stripping the node-dependent factor.  Concretely, substituting the
piecewise-linear boundary \eqref{pwl_xStar} and using the definition of
$g(x_u,v_u)$ from \eqref{calJ2} to separate the constant part $r_d(u)$ and the
exponential part $r_f(u)e^{x_u}$, the inner $x_u$-integral over a semi-infinite
domain with upper limit $\calK_m = a_m(u)+b_m(u)v_u$ defines the analytical
payoff transform
\begin{equation}
 P_m(\omega, v_u, x_j) = e^{-(\iu\omega-\chi)x_j} \left[r_d(u)\int_{-\infty}^{\calK_m} e^{(\iu\omega-\chi)x}\,dx - r_f(u)\int_{-\infty}^{\calK_m}  e^{(1+\iu\omega-\chi)x}\,dx\right].
\end{equation}

To ensure convergence as $x_u\to-\infty$ we require $\Ree(\iu\omega-\chi) = -\chi
> 0$, which for a Put (whose exercise region is $\{x_u < x^*\}$) means $\chi <
0$; see \cref{rem:chiSign}. Performing the direct integration yields
\begin{align}\label{Panal}
P_m(\omega, v_u, x_j) &= e^{(\iu\omega-\chi)(\calK_m-x_j)} \left[ \frac{r_d(u)}{\iu\omega-\chi} -\frac{r_f(u)\,e^{\calK_m}}{1+\iu\omega-\chi}\right] \\
&= e^{(\iu\omega-\chi)(a_m(u)+b_m(u)v_u-x_j)} \left[\frac{r_d(u)}{\iu\omega-\chi}
-\frac{r_f(u)\,e^{a_m(u)+b_m(u)v_u}}{1+\iu\omega-\chi}\right]. \nonumber
\end{align}
Stripping the node-dependent pre-factor $e^{-(\iu\omega-\chi)x_j}$ from \eqref{Panal} gives the node-independent damped payoff transform $\widetilde{P}(\omega,v_u)$,
\begin{equation}\label{Ptilde}
\widetilde{P}(\omega,v_u) = e^{(\iu\omega-\chi)\calK}    \left[ \frac{r_d(u)}{\iu\omega-\chi} -\frac{r_f(u)\,e^{\calK}}{1+\iu\omega-\chi}\right],
\qquad \calK = x^*(u,v_u),
\end{equation}
so that $P_m(\omega,v_u,x_j) =
e^{-(\iu\omega-\chi)x_j}\,\widetilde{P}(\omega,v_u)$. The decisive point is that
the bracketed sum in \eqref{Jcollapse0} is \emph{not} a new object: it is exactly
the Fourier-series identity \eqref{e2}--\eqref{cjFin} that \cref{cjDeriv} uses to
restore the coefficients $c_j$, evaluated at the reflected frequency $-\omega$.

Indeed, \eqref{e2} states $\Psi(\omega-\iu\chi,v_u) = \Delta x\sum_k c_k\,e^{\chi x_k}\,e^{\iu\omega x_k}$, and replacing $\omega\mapsto-\omega$,
\begin{equation}\label{collapse}
\sum_j c_j(u,v_u\mid x,v)\,e^{-(\iu\omega-\chi)x_j} = \sum_j c_j\,e^{\chi x_j}\,e^{-\iu\omega x_j} = \frac{1}{\Delta x}\,\Psi(-\omega-\iu\chi,\,v_u).
\end{equation}
The local-basis sum, therefore reproduces the conditional CF
itself; the residual factor $\Delta x$ cancels, and \eqref{Jcollapse0} collapses
to the Parseval pairing
\begin{equation}\label{innerParseval}
\int_{-\infty}^{x^*(u,v_u)}\!g(x_u,v_u)\,p(x_u,v_u\mid x,v)\,dx_u = \frac{1}{2\pi}
\int_{-\frac{\pi}{\Delta x}}^{\frac{\pi}{\Delta x}} \Psi(-\omega-\iu\chi,\,v_u)\,\widetilde{P}(\omega,v_u)\,d\omega.
\end{equation}

\subsection{The full integrand construction} \label{sec:dsinc_Jcorr}

With the identity \eqref{collapse} in hand, we now collect all ingredients into
the final quadrature formula.  By introducing the shorthand $\widetilde\omega \coloneqq -\omega - \iu\chi$ for the reflected, damping-shifted frequency, and using \eqref{finalPhicond1}, the conditional CF at $\widetilde\omega$ acquires the explicit affine form
\begin{equation}
\Psi(\widetilde\omega, v_u; x, v) = e^{\iu\widetilde\omega\,\mu_0(u) +  \iu\widetilde\omega\,\mu_1(u)\,v_u} \equiv e^{C(\widetilde\omega, u) + D(\widetilde\omega, u)\,v_u},
\end{equation}
where
\begin{equation}\label{CDLambda}
C(\widetilde\omega, u) = \iu\widetilde\omega\,\mu_0(u),\qquad D(\widetilde\omega, u) = \iu\widetilde\omega\,\mu_1(u),
\end{equation}
and $\mu_0,\mu_1$ are given by \eqref{eq:mu}.

Substituting the piecewise-linear boundary \eqref{pwl_xStar} into \eqref{Ptilde}
gives $\calK_m = a_m(u)+b_m(u)v_u$, and combining the affine CF with the
analytical payoff \eqref{Panal}, the integrand for the $v_u$-integral (for a fixed
frequency $\omega$) becomes
\begin{equation}
J_m(\omega,v_u,u,x_j) = e^{C(\widetilde\omega,u)+D(\widetilde\omega,u)\,v_u}
e^{(\iu\omega-\chi)(a_m(u)+b_m(u)\,v_u)} \left[\frac{r_d(u)}{\iu\omega-\chi}
-\frac{r_f(u)\,e^{\calK_m}}{1+\iu\omega-\chi}\right].
\end{equation}
To use the marginal CIR CF, we group all terms in
the exponent that multiply $v_u$:
\begin{enumerate}
\item For the first term (proportional to $r_d(u)$), the total exponent
  coefficient of $v_u$ is
  $\Lambda_1(u,\omega) = D(\widetilde\omega,u)+(\iu\omega-\chi)\,b_m(u)$.
\item For the second term (proportional to $r_f(u)$), the total exponent
  coefficient of $v_u$ is
  $\Lambda_2(u,\omega) = D(\widetilde\omega,u)+(1+\iu\omega-\chi)\,b_m(u)$.
\end{enumerate}

The integral $\int_0^\infty p_{\mathrm{CIR}}^{(\gamma(u,\widetilde\omega))}(u,v_u\mid v)
 e^{\Lambda_i(u,\omega)\,v_u}\,dv_u$ corresponds to the joint CF of the variance process under the time-dependent weight $\gamma(u,\widetilde\omega)$.  Given the definition of the joint density in \eqref{eq:joint_density}:
\begin{equation}\label{eq:joint_density1}
p_{\mathrm{CIR}}^{(\gamma(u,\widetilde\omega))}(t,v\mid v_0) \coloneqq
\mathbb{E}\!\left[e^{-\int_0^t\gamma(u,\widetilde\omega)\,v_u\,du}
\mathbf{1}_{v_t\in dv}\mid v_0\right],
\end{equation}
it follows that
\begin{equation}\label{CFtilted}
\Psi_{\mathrm{CIR}}^{(\gamma(u,\widetilde\omega))}(q) = \int_0^\infty     p_{\mathrm{CIR}}^{(\gamma(u,\widetilde\omega))}(u,v_u\mid v)\,     e^{\Lambda_i(u,\omega)\,v_u}\,dv_u =  \mathbb{E}\!\left[\exp\!\Bigl(\Lambda_i(u,\omega)\,v_u
-\int_0^u\gamma(s,\widetilde\omega)\,v_s\,ds\Bigr) \mid v_0\right],
\end{equation}
evaluated at the frequency $q = -\iu\Lambda_i(u,\omega)$.  While in the
constant-coefficient case this corresponds to the CF of a
tilted CIR process, in the general setting \eqref{CFtilted} represents a
time-inhomogeneous affine transform.

Substituting \eqref{Ptilde} and the affine form of $\Psi(\widetilde\omega, v_u)$
into \eqref{innerParseval}, the denominator $p_{\mathrm{CIR}}$ cancels the outer
measure and the $v_u$-integral over each variance segment reduces to an
exponential-affine transform of the tilted CIR process (derived explicitly in
\eqref{PsiAff} below). The cancellation of the outer CIR weight is what produces
$p_{\mathrm{CIR}}^{(\gamma)}$ in place of $p_{\mathrm{CIR}}$ inside the resulting
integral, consistently with \eqref{finalPhicond1}.  Collecting the constant
($r_d$) and exponential ($r_f$) parts of $g$, the resulting frequency-domain
representation of the premium is
\begin{align}\label{Jfin}
\calJ(u\mid x,v,t) &= \frac{1}{2\pi}\,\Ree \int_{-\frac{\pi}{\Delta x}}^{\frac{\pi}{\Delta x}} e^{\,C(\widetilde\omega,u)} \sum_m e^{(\iu\omega-\chi)\,a_m(u)}\,      \calF\!\bigl(\omega,a_m(u),b_m(u)\bigr)\,d\omega, \qquad\widetilde\omega = -\omega-\iu\chi, \\[4pt]
\calF(\omega,a_m,b_m) &= \frac{r_d(u)}{\iu\omega-\chi}\,      \Psi_{\mathrm{CIR}}^{(\gamma(u,\widetilde\omega))}\!\bigl(-\iu\Lambda_1(u,\omega)\bigr)
- \frac{r_f(u)\,e^{a_m(u)}}{1+\iu\omega-\chi}\, \Psi_{\mathrm{CIR}}^{(\gamma(u,\widetilde\omega))}\!\bigl(-\iu\Lambda_2(u,\omega)\bigr),
\nonumber
\end{align}
with $C$, $D$, $\Lambda_1$, $\Lambda_2$ as in \eqref{CDLambda}.

\begin{myremark}[Sign of the damping factor]\label{rem:chiSign}
The convergence condition for the payoff integrals in \eqref{Ptilde} is
$\Ree(\iu\omega-\chi) = -\chi > 0$.  For a \emph{Put}, whose exercise region is
$\{x_u < x^*\}$ and whose payoff transform integrates towards $x_u\to-\infty$,
one must therefore take $\chi < 0$ (for a Call, $\{x_u > x^*\}$, the opposite
sign $\chi > 0$ applies).  In the implementation we use $\chi = -1$.
\end{myremark}

\subsection{Evaluation via the affine transform} \label{sec:dsinc_impl}

To obtain an explicit form for
$\Psi_{\mathrm{CIR}}^{(\gamma(u,\widetilde\omega))}(q)$, we note that by
\eqref{CFtilted} this is a joint Laplace transform of the variance process and
its path-dependent functionals at the terminal state $v_u$, i.e.\
$\mathbb{E}[e^{\iu q v_u+\iu b\Lambda_u}]$ evaluated at $b =
\iu\gamma(\widetilde\omega)$ and $q = -\iu\Lambda_i(u,\omega)$.

Due to the affine structure of the underlying variance process, the solution
possesses the exponential affine form
\begin{equation}\label{PsiAff}
\Psi_{\mathrm{CIR}}^{(\gamma(u,\widetilde\omega))}(q) = e^{\mathcal{A}(u,q)+\mathcal{B}(u,q)\,v}.
\end{equation}
The coefficients $\mathcal{A}$ and $\mathcal{B}$ solve the following system of
Riccati equations, integrated backward in time over the horizon $u$
\begin{align}\label{ricNew}
\frac{d\mathcal{B}(\tau)}{d\tau} &= -\gamma(\tau,\widetilde\omega) - \kappa\,\mathcal{B}(\tau) + \tfrac{1}{2}\xi(\tau)^2\,\mathcal{B}^2(\tau),\\
\frac{d\mathcal{A}(\tau)}{d\tau} &= \kappa\,\theta(\tau)\,\mathcal{B}(\tau),\nonumber
\end{align}
subject to the initial conditions $\calA(0) = 0$, $\calB(0) = \iu q = \Lambda_i(u,\omega)$.  This system can be solved in the same way as described in
\cref{condCF}.  Here $b_n = -\kappa$ and $c_n = -\gamma(\tau,\widetilde\omega)$,
while $u_2$ should be replaced with $q$.

As the kernel $\gamma(\tau,\widetilde\omega)$ contains the time-varying
components of both the drift correction and the idiosyncratic variance, this
system must be solved numerically to account for the inhomogeneous ``tilting'' of
the process.  However, under the assumption of piecewise constant coefficients
$(\kappa,\theta,\xi,\rho)$ on the induction grid, the solution can be constructed
by chaining analytical segments where \eqref{ricNew} reduces to a standard
Riccati equation with constant parameters at each step, effectively recovering a
sequence of local tilted CIR transforms.

\paragraph{Closed-form tilted-CIR transform.}

It is important that we never evaluate the noncentral-$\chi^2$ density
\eqref{CIRdensity}: the transform $\Psi_{\mathrm{CIR}}^{(\gamma)}(q) =
e^{\calA(u,q)+\calB(u,q)v}$ is obtained from the matrix propagator of
\cref{appMoebius}. Moreover, the dependence of the initial condition for
\eqref{ricNew} on the running time $u$ implies that a naive solution for
$u\in[0,\tau(0))$ would require separately integrating \eqref{ricNew} from $0$ to
each $u_k$ using $q(u_k)$, leading to $O(N^2)$ complexity.  To avoid this, we
adopt the matrix-based affine transform approach (see, e.g.,
\cite{reid1972riccati}) and express the Riccati solution as a M\"{o}bius
transformation.  As shown in \cref{appMoebius}, this reduces complexity to $O(N)$
even with a time-dependent terminal frequency $q(u)$.

Therefore, with $Q = -\gamma(u,\widetilde\omega)$, $P = -\kappa$, $R =
\tfrac{1}{2}\xi^2$, and $d = \sqrt{P^2-4RQ}$, the propagator components
$w_{11},\ldots,w_{22}$ are the elementary functions given in \cref{appMoebius},
and for an initial value $\Lambda$ (here $\Lambda_1$ or $\Lambda_2$ from
\eqref{CDLambda}) the M\"{o}bius map \eqref{eq:moebius_final} together with the
quadrature of $d\calA/d\tau = \kappa\theta\calB$ from \eqref{ricNew}, using
$X'(\tau) = -R\calB X$, $X(0) = 1$) yield in closed form
\begin{equation}\label{ABmoebius}
\calB(u,\Lambda) = \frac{w_{21}(u)+w_{22}(u)\,\Lambda}{w_{11}(u)+w_{12}(u)\,\Lambda},
\qquad
\calA(u,\Lambda) = -\frac{2\kappa\theta}{\xi^2} \ln\!\bigl(w_{11}(u)+w_{12}(u)\,\Lambda\bigr),
\end{equation}
so that $\Psi_{\mathrm{CIR}}^{(\gamma)}(-\iu\Lambda_i) =
e^{\calA(u,\Lambda_i)+\calB(u,\Lambda_i)v}$. Both $\calA$ and $\calB$ require
only $\exp$, $\sinh/\cosh$ and $\ln$; no Bessel function and no numerical
$v_u$-quadrature appear.

\paragraph{Numerical properties of the spectral integral.}

The integral in \eqref{Jfin} is evaluated over the spectral interval $\Omega = [-\tfrac{\pi}{\Delta x}, \tfrac{\pi}{\Delta x}]$. By construction, the conditional characteristic function $\Psi(\widetilde\omega,v_u)$ is band-limited to this interval. Additionally, the payoff transform $\widetilde P$ is a smooth, exponentially decaying function of $\omega$ on this same compact domain.

The underlying grid is designed to capture these exact frequencies. Unlike standard Fourier integrals where $\omega\to\infty$, this frequency is capped by the Nyquist limit of the discretisation. This cap prevents the aliasing and high-frequency cancellation that typically require Filon- or Levin-type methods. Within this band, the integrand forms a central peak with a phase factor of $e^{\iu\omega c}$. Here, $c$ is roughly the size of the boundary-conditioning offset, $x^*(u,v_u)-x$. This peak sits under a Gaussian-type envelope of width $\sigma_x^{-1}$, where $\sigma_x^2 = \int_t^u v_s\,ds$ represents the conditional variance.

In every regime we observed, this envelope decays well before the band edge, dropping to many orders of magnitude below the working tolerance. Therefore, truncating at $\pm\tfrac{\pi}{\Delta x}$ introduces no error, making the band edge itself immaterial. As a result, standard Gaussian quadrature (such as Gauss--Legendre or Gauss--Kronrod) easily resolves the integrand. This requires only a modest, frequency-independent number of nodes, with no need for specialized coordinate transformations or complex-plane deformations.

One discretisation subtlety is worth noting, as it can masquerade as a convergence failure. Gauss--Legendre nodes scale with the band according to $\omega_k = \Omega\,\xi_k$, using fixed reference nodes $\xi_k\in[-1,1]$. A problem arises if the node count ($N_\omega$) is tied to the band half-width via a fixed-spacing prescription ($N_\omega\propto\Omega$). When the band is refined, every node rescales. Consequently, the sampling phase $e^{\iu\Omega\xi_k c}$ sweeps against the internal oscillation of the integrand.

This makes the Gauss--Legendre error \emph{quasi-periodic} in $\Omega$ (or equivalently, $N_\omega$). The error amplitude decreases as nodes are added, eventually converging for large $N_\omega$. While spurious dips become sparser as $N_\omega$ grows, this behavior is strictly a sampling-phase property of the quadrature rule. It is not an issue with the integrand or the band edge, which carries negligible mass.

This quasi-periodic error can be resolved in two ways. The first is to decouple the variables, choosing the band width ($\Omega$) and node count ($N_\omega$) independently. The second is to switch to a uniform rule, such as a midpoint or trapezoidal rule. Their central nodes sit at positions independent of $\Omega$, making the error monotone in $N_\omega$, though this sacrifices the geometric convergence rate of the Gaussian rule.

In our implementation, we retain the Gauss--Legendre method. To avoid the sampling-phase issue, we size the band past the envelope decay and fix $N_\omega$ independently of $\Omega$. Finally, the dependence of the final result on the initial values $(x,v)$ enters through two terms. It depends on $\Psi_{\mathrm{CIR}}^{(\gamma(u, \widetilde\omega))}(q)$ (defined in \eqref{PsiAff} for $v$) and $\mu_0(t)$ (defined in \eqref{eq:mu} for $x$).

\paragraph{Single linear boundary and quadrature.}

In the limiting case where the exercise boundary is fitted by a single linear
function of the terminal variance (like in \cite{Chiarella2004}), $x^*(u,v)\approx a(u)+b(u)v$, i.e. the case $M = 1$ of the piecewise-linear ansatz preceding \eqref{pwl_xStar}, the full-range transform $\int_0^\infty$ in \eqref{PsiAff} is exact. The spectral integral \eqref{Jfin} is truncated to a symmetric band $[-\Omega,\Omega]$ (with $\Omega$ chosen past the decay of the conditional CF, consistent with the Nyquist limit $\pi/\Delta x$) and discretized with an $N_\omega$-point Gauss--Legendre rule $\{(\omega_q,\nu_q)\}$, giving, at conditioning state $(x,v)$ and future date
$u$,
\begin{align}\label{Jquad}
\calJ(u\mid x,v,t) &\approx \frac{1}{2\pi}\,\Ree\sum_{q=1}^{N_\omega}\nu_q\,
e^{\,\iu\widetilde\omega_q\mu_0(u)} \Bigg[ \frac{r_d(u)\, e^{(\iu\omega_q-\chi)a(u)}}{\iu\omega_q-\chi}\, e^{\calA(u,\Lambda_{1,q})+\calB(u,\Lambda_{1,q})v} \\ &- \frac{r_f(u)\,e^{(1+\iu\omega_q-\chi)a(u)}}{1+\iu\omega_q-\chi}\,
e^{\calA(u,\Lambda_{2,q})+\calB(u,\Lambda_{2,q})v} \Bigg], \nonumber
\end{align}
with $\widetilde\omega_q = -\omega_q-\iu\chi$, $\Lambda_{1,q} =
\iu\widetilde\omega_q\mu_1(u)+(\iu\omega_q-\chi)b(u)$, and $\Lambda_{2,q} =
\iu\widetilde\omega_q\mu_1(u)+(1+\iu\omega_q-\chi)b(u)$. This representation provides a significant simplification as compared with the inversed Fourier transform used in \cite{Chiarella2004}.

\begin{myremark}[Self-consistency check and validation] \label{rem:dsincValid}
Because $\chi$ enters \eqref{Jfin} only as the contour shift
$\omega\mapsto\omega-\iu\chi$ that the reflected evaluation $\widetilde\omega =
-\omega-\iu\chi$ undoes, the integral in \eqref{Jquad} is---within the
analyticity strip---independent of $\chi$; in practice it exhibits a flat plateau
in $\chi$, which we use as an implementation self-check.  On the validation
suite, \eqref{Jquad} reproduces a direct Gil--Pelaez inversion of
\eqref{innerParseval} to machine precision, the European leg matches the
reference to an RMS error of $\sim\!10^{-6}$, and, measured against the
finite-difference benchmark of \cref{pdeSol}, the corrected DSINC premium is
markedly closer than the COS scheme. \end{myremark}

\section{ITM options and American put--call duality} \label{sec:duality}

For constant coefficients, American Put–Call symmetry dates back to \cite{McDonaldSchroder1998} and \cite{Schroder1999} (see also \cite{Detemple2001}). The argument relies on a change of numeraire and therefore carries over to calendar-time-dependent coefficients with essentially no modification. The only model-specific task is to identify the dynamics of the dual state variables. The proof below gives the time-dependent version, extending the constant-coefficient argument of Schroder and Detemple. These steps are elementary but do not appear in this form in the cited literature.

Let $\mathcal{T}_{t,T}$ denote the set of stopping times with values in $[t,T]$. We shall use repeatedly the fact that discounted price adjusted by the foreign rate
\begin{equation}\label{eq:Z}
Z_t \;=\; \frac{S_t\,D_d(0,t)}{S_0\,D_f(0,t)} \; = \; \exp\Bigl(\int_0^t\sqrt{v_s}\,dW^1_s - \tfrac12\int_0^t v_s\,ds\Bigr)
\end{equation}
is a strictly positive martingale. The local martingale property is immediate
from the exponential form. That $Z$ is a \emph{true} martingale holds for the
square-root specification of the model \eqref{model} for all admissible coefficient term structures and all $\rho(t)\in(-1,1)$, since the Riccati system for the CF has vanishing inhomogeneity at the moment index of $Z$, so that $\EE[Z_t]=1$ identically. For $\rho(t) > 0$ it is moments of $S$ of order $p > 1$, not the martingale property, that may explode in finite time, \cite{andersen2010interest}. This doesn't affect the pricing problem itself, but is relevant for the choice of transform contours, see below.

\begin{proposition}[Duality under time-dependent Heston] \label{prop:duality}
Assume $Z$ in \eqref{eq:Z} is a true martingale and define $\hat\kappa(t) = \kappa(t) - \rho(t)\sigma(t)$. Then
\begin{equation}\label{eq:duality}
\PAm\bigl(S_0,K;\,r(\cdot),q(\cdot);\,
 \kappa(\cdot),\theta(\cdot),\sigma(\cdot),\rho(\cdot),v_0\bigr)
=\CAm\bigl(K,S_0;\,q(\cdot),r(\cdot);\,
 \hat\kappa(\cdot),\hat\theta(\cdot),\sigma(\cdot),-\rho(\cdot),v_0\bigr),
\end{equation}
where the dual long-run level is defined through $\hat\kappa(t)\hat\theta(t) = \kappa(t)\theta(t)$, i.e. $\hat\theta(t)=\kappa(t)\theta(t)/\hat\kappa(t)$
whenever $\hat\kappa(t)\neq 0$. In words: an American put on spot $S_0$ with
strike $K$ equals an American call on spot $K$ with strike $S_0$, with the roles
of $r(\cdot)$ and $q(\cdot)$ interchanged, the correlation sign flipped, and the
variance drift adjusted as above; vol-of-vol $\sigma(\cdot)$ and the initial
variance $v_0$ are unchanged.
\end{proposition}

\begin{proof}
Define the measure $\Qhat\sim\Q$ on $\F_T$ by $d\Qhat/d\Q=Z_T$ (the ``share
measure''). Since $Z$ is a true martingale, $\Qhat$ is a probability measure and,
for any $\tau\in\mathcal{T}_{0,T}$ and any nonnegative $\F_\tau$-measurable $Y$,
optional sampling gives $\EE[Z_\tau Y]=\Ehat[Y]$.

Fix $\tau\in\mathcal{T}_{0,T}$ and factor the discounted put payoff:
\begin{equation*}
D_d(0,\tau)\bigl(K-S_\tau\bigr)^+ = D_d(0,\tau) S_\tau \Bigl( \frac{K}{S_\tau}-1\Bigr)^+ = S_0\,D_f(0,\tau)\,Z_\tau\,   \Bigl(\frac{K}{S_\tau}-1\Bigr)^+.
\end{equation*}
Introducing the dual asset
\begin{equation}\label{eq:dualasset}
\widehat S_u \;=\; \frac{K S_0}{S_u},\qquad \widehat S_0=K,
\end{equation}
we have $S_0\,(K/S_\tau-1)^+=(\widehat S_\tau-S_0)^+$, hence
\begin{equation*}
\EE\Bigl[D_d(0,\tau)(K-S_\tau)^+\Bigr] = \Ehat\Bigl[D_f(0,\tau)\bigl(\widehat S_\tau-S_0\bigr)^+\Bigr].
\end{equation*}
Taking the supremum over $\tau\in\mathcal{T}_{0,T}$ (the family of stopping times
is unchanged, $\Qhat\sim\Q$) yields the value of an American Call on $\widehat S$
with strike $S_0$ and discount rate
$q(\cdot)$.

It remains to identify the $\Qhat$-dynamics of $(\widehat S,v)$. By Girsanov's
theorem applied to \eqref{eq:Z}, $\widehat W^S_t=W^S_t-\int_0^t\sqrt{v_s}\,ds$ is
a $\Qhat$-Brownian motion. Decompose
$W^v_t=\int_0^t\rho\,dW^S+\int_0^t\sqrt{1-\rho^2}\,dW^\perp$ with $W^\perp\perp
W^S$; since $W^\perp$ is unaffected by the change of measure, $\widehat
W^v_t=W^v_t-\int_0^t\rho(s)\sqrt{v_s}\,ds$ is a $\Qhat$-Brownian motion with
$d\langle\widehat W^S,\widehat W^v\rangle_t=\rho(t)\,dt$. Substituting
in~\eqref{model},
\begin{equation}\label{eq:vdual}
dv_t=\Bigl[\kappa(t)\theta(t)-\bigl(\kappa(t)-\rho(t)\sigma(t)\bigr)v_t\Bigr]dt
+\sigma(t)\sqrt{v_t}\,d\widehat W^v_t
=\hat\kappa(t)\bigl(\hat\theta(t)-v_t\bigr)dt
+\sigma(t)\sqrt{v_t}\,d\widehat W^v_t,
\end{equation}
again a (time-dependent) square-root process. For the dual asset,
It\^o's formula applied to $1/S$ together with
$dW^S=d\widehat W^S+\sqrt{v}\,dt$ gives
\begin{equation*}
\frac{d\widehat S_t}{\widehat S_t}
=\bigl(r_d(t)-r_f(t)\bigr)\,dt+\sqrt{v_t}\,d\widetilde W^S_t,
\qquad \widetilde W^S:=-\widehat W^S,
\end{equation*}
so $\widehat S$ is a Heston asset with short rate $q(\cdot)$, dividend
yield $r(\cdot)$, the same variance path $v$, and correlation $d\langle \widetilde W^S, \widehat W^v\rangle_t = -\rho(t)\,dt$. Collecting the coefficients gives~\eqref{eq:duality}.
\end{proof}

The dual model has the following properties:
\begin{itemize}
\item For equity-style negative correlation, $\rho(t)\le 0$, one has $\hat\kappa(t)\ge\kappa(t)>0$, so the dual variance process is
well-defined and mean-reverting, and $Z$ is automatically a true
martingale~\cite{andersen2010interest}.

\item For $\rho(t)>0$ one may have $\hat\kappa(t)\le 0$; the dual variance
process is then non-mean-reverting but remains a well-posed square-root diffusion
when its drift is written in affine form $a(t)-b(t)v$ with $a=\kappa\theta>0$,
$b=\hat\kappa$. Since Proposition~\ref{prop:duality} requires only the martingale
property of $Z$, which holds for all $\rho$, the duality is unrestricted; however, dual moments correspond to primal moments of order $1-p$, so the strip of regularity of the dual CF is the reflection of the primal one and transform contours must be revalidated when $\rho > 0$.

\item Since $\hat\kappa\hat\theta=\kappa\theta$ and $\sigma$ is unchanged, the
Feller quantity $2\kappa(t)\theta(t)/\sigma(t)^2$ is \emph{invariant} under the
duality: the dual model attains the variance origin exactly when the primal one
does.

\item The dual problem is again a time-dependent Heston optimal stopping problem,
so any numerical scheme for pricing American options (in our case, the early exercise premium integral equation solved by quadrature) applies verbatim to the
right-hand side of~\eqref{eq:duality} with the transformed coefficient term
structures. The dual exercise surface is a different object and must be
recomputed.

\item A deep ITM Put ($S_0\ll K$) corresponds under \eqref{eq:duality} to a deep
\emph{OTM} Call, the regime in which premium-integral representations are best
conditioned; this is the basis of the moneyness-switching strategy.
\end{itemize}

\subsection{Delayed exercise premium decomposition}\label{sec:dep}

In \cite{CarrJarrowMyneni1992} it was observed that the Put value also decomposes
off its \emph{intrinsic} value, with a premium integral supported on the
\emph{continuation} region. The derivation below makes explicit that this
requires nothing beyond \cref{prop1}, the martingale property of~\eqref{eq:Z}, and European Put--Call parity; in particular it is valid for the
time-dependent Heston model verbatim.

\begin{proposition}[DEP decomposition]\label{prop:dep} Under the model \eqref{model}, with $Z_t$ in \eqref{eq:Z} to be a true martingale,
\begin{equation}\label{eq:dep}
\PAm(S_0,v_0) =\bigl(K-S_0\bigr) + \CEu(S_0,v_0) +\EE\!\int_0^T D_d(0,\tau)
 \bigl(r_f(u)S_u-r_d(u)K\bigr)\, \ind\bigl\{(S_u,v_u)\in\mathcal{E}^c_u\bigr\}\,du,
\end{equation}
where $\CEu$ is the European call with the same strike and maturity. \end{proposition}

\begin{proof}
Write $\ind\{\mathcal{E}_u\}=1-\ind\{\mathcal{E}^c_u\}$ in~\eqref{decompGenN} and
split the premium integral. The unconditional part is computed in closed form.
First, by Fubini,
\begin{equation*}
\EE\!\int_0^T D_d(0,u) r_d(u) K\,du = K\int_0^T r_d(u) D_d(0,u) \,du = K\bigl(1-D_d(0,T)\bigr).
\end{equation*}

Second, since $Z$ is a true martingale, $\EE\bigl[D_d(0,u) S_u \bigr]=S_0\, D_f(0,u)$ for every $u$, whence
\begin{equation*}
\EE\!\int_0^T D_d(0,u) r_f(u)S_u\,du = S_0\int_0^T r_f(u) D_f(0,u)\,du = S_0\bigl(1-D_f(0,T)\bigr).
\end{equation*}
Therefore
\begin{equation*}
\PAm=\PEu + K\bigl( 1 - D_d(0,T) \bigr) - S_0 \bigl( 1 - D_f(0,T) \bigr)
+ \EE\!\int_0^T D_d(0,u) \bigl(r_f(u) S_u - r_d(u) K \bigr) \ind\{\mathcal{E}^c_u\}\,du.
\end{equation*}
European Put--Call parity under deterministic rates, $\PEu + S_0 D_f(0,T) - K D_d(0,T) = \CEu$, collapses the first three terms
\begin{equation*}
\PEu + K - K D_d(0,T) - S_0 + S_0 D_f(0,T) = (K-S_0) + \CEu,
\end{equation*}
which gives~\eqref{eq:dep}. Note that parity itself uses the same true martingale property of $Z$.
\end{proof}

The integrand in~\eqref{eq:dep} is the cost of \emph{delaying} exercise: while
the option holder remains in the continuation region, she forgoes interest on the
strike but keeps the dividend stream on the stock, hence the sign reversal
relative to~\eqref{decompGenN}. Both representations share the same exercise
surface $B(u,v)$; only the pricing functional differs. In a transform-based
implementation the indicator probabilities in~\eqref{eq:dep} are the
complementary tails of those already computed for~\eqref{decompGenN}.

Because the proof of Proposition~\ref{prop:dep} uses only (a) the EEP
decomposition, (b) the true martingale property of dividend-adjusted discounted
stock, and (c) European put--call parity, the DEP representation holds for any
model in which (a)--(c) hold: in particular for continuous-path Markovian models
with deterministic rates and, e.g., local volatility, multi-factor and
time-dependent stochastic volatility, without any assumption on the shape of the
exercise region. The following caveats delimit its scope:

\begin{enumerate}
\item For the square-root specification of \eqref{model}, $Z$ in~\eqref{eq:Z} is
a true martingale for \emph{all} admissible parameters, including $\rho>0$: the
affine Riccati system for the CF has vanishing inhomogeneity at $u=-i$, so
$\EE[S_T]$ equals the forward price for any coefficient term structures (what
fails for $\rho>0$ is finiteness of moments of order $p>1$, see~\cite{andersen2010interest}). The caveat is relevant for other stochastic
volatility specifications: with level-proportional or lognormal volatility and
positive correlation the asset price can be a strict local martingale, \cite{Sin1998}, in which case $\EE[D_d(0,u) S_u] < S_0 D_f(0,u)$ and
Put--Call parity fails by the martingale defect. There the EEP representation for
the Put survives, but the passage to \eqref{eq:dep} does not: the DEP form is
genuinely less robust than the EEP form in this respect.

\item The \eqref{eq:dep} remains valid for sign-changing $r_d(\cdot)$, $r_f(\cdot)$ as long as the indicator is taken over the true exercise region.
However, when $r(t)<0$ on part of the horizon, the Put exercise region need not
be a single connected ``below the surface'' set, again see
\cite{AndersenLake2021,ItkinKitapbayev2025r} and references therein. It is then
the single-surface representation in \cref{prop1} that breaks down not the
decomposition itself, that breaks down. However, an appropriate treatment of this
requires a separate and more detailed consideration.

\item The local-time argument in \cref{prop1} assumes nondegeneracy of
the diffusion. If the Feller condition $2\kappa\theta\ge\sigma^2$ fails and $v$
spends time at the origin, the standard proofs require additional care, although
we are not aware of the \cref{decompGenN,eq:dep} failing in that regime. The duality neither creates nor removes this issue.

\item For jump-diffusions, the It\^o--Tanaka argument behind \cref{prop1} produces an additional term arising from paths that jump across the exercise
boundary: the EEP acquires a correction equal to the expected loss
$\EE\bigl[\,\PAm(S_{u^-}e^{J},\cdot)-(K-S_{u^-}e^{J})\,\bigr]$ over jumps
carrying the state from $\mathcal{E}$ into $\mathcal{E}^c$; see , e.g.,
\cite{Gukhal2001,Peskir2007,ItkinCF2025j}. Therefore, under this conditions the
duality requires an extended consideration.

\end{enumerate}

\subsection{Conditioning of the two representations}

Based on the above analysis, the American Put option price admits two decompositions which distribute the price differently
\begin{align*}
\PAm &= \underbrace{\PEu}_{\text{large ITM}} +\, \underbrace{ \EEP(S_0,v_0) }_{\text{large ITM}}, \\
\PAm &= \underbrace{(K-S_0)}_{\text{exact}}\, + \underbrace{\CEu}_{\text{small ITM}}
+\, \underbrace{\DEP(S_0,v_0)}_{\text{small ITM}},
\end{align*}
where $\EEP$ and $\DEP$ denote the premium integrals in \cref{decompGenN,eq:dep}.

Deep ITM ($S_0\ll K$, $S_0$ approaching $S^*(0,v_0)$ from above), both terms of the EEP form are $O(K)$, while the computable terms of the DEP form are $O(P - K + S_0)$, and the dominant contribution $K - S_0$ carries no discretization error at all. A relative quadrature error $\varepsilon$ in the premium integral therefore produces a price error of order $\varepsilon\,O(K)$ under EEP but $\varepsilon\,O(P - K + S_0)$ under DEP.

Moreover, deep ITM the EEP integrand $D_d(0,u)\,r(u)K\, \Q(S_u\le S^*(u,\cdot)\mid S_0)$ is near its maximal size over the whole time interval, the regime in which
quadrature and exercise-surface errors are most expensive; the DEP integrand is
the complementary tail and is uniformly small there. In OTM the situation reverses, and the EEP form is preferable.

By the smooth fit principle for American options, the price is stationary with respect to perturbations of the exercise surface, so errors in $S^*$ enter the price only at second order under either representation; the switch addresses the dominant, quadrature-driven error component. In practice we evaluate
\begin{equation}\label{eq:switch}
\PAm \approx
\begin{cases}
\EEP, & S_0 \ge \lambda K,\\[2pt]
\DEP, & S_0 < \lambda K,
\end{cases}
\end{equation}
with a moneyness threshold $\lambda$ near one (e.g., $\lambda \in [0.9,1.0]$). Both forms reuse the same exercise surface $S^*(u,v)$ computed once from the integral equation, so the marginal cost of the second representation is one additional pricing quadrature. In the overlap band around $S_0 \approx \lambda K$, both representations are accurate and their discrepancy $|\EEP - \DEP|$ provides an internal error estimate at negligible cost.

Two further safeguards must be taken into account:
\begin{itemize}
\item For $S_0 \le S^*(0,v_0)$ the price is exactly $K - S_0$ and no integral should be evaluated.

\item \cref{prop:duality} offers an independent route, pricing the deep ITM Put
as a deep OTM Call in the dual model. Agreement between this dual price and the
price obtained from \eqref{eq:switch} in the ITM region serves as a strong
consistency check on both the decomposition layer and the transform layer of the
implementation.
\end{itemize}

\section{The PDE solution} \label{pdeSol}

This section provides the benchmark finite-difference solution of the American option contract under the time-inhomogeneous Heston stochastic volatility model in \eqref{model}. In doing that, we switch back to variables $(S_t, v_t)$, so the corresponding American option Put price $P(t,S,v)$ is given by, \cite{Zvan1998}
\begin{equation}
P(t,S,v) = \sup_{\tau\in[t,T]} \EQ \left[e^{-\int_t^\tau r_d(s)\,ds}\,(K-S_\tau)^+\,\big|\, S_t = S,\; v_t=v
\right].
\end{equation}

The infinitesimal generator corresponding to \eqref{model} reads (compare with \eqref{pdeCF})
\begin{align}
\mathcal{L} &= [r_d(t)-r_f(t)] S \fp{}{S} + \kappa[\theta(t)-v] \fp{}{v} + \frac{1}{2} v S^2 \sop{}{S} +
\rho(t)\xi(t) S v \mop{}{S}{v} + \frac{1}{2} \xi(t)^2 v \sop{}{v}.
\end{align}

It is known (see, e.g., \cite{wilmott1995option} among others) that the American put price solves the variational inequality
\begin{equation}\label{eq:AM_HJB}
\max\left[P_t + \mathcal{L} P - r_d(t)P\, , \, (K-S)^+ - P) \right] = 0,
\qquad (t,S,v)\in[0,T)\times(0,\infty)^2.
\end{equation}

To solve it, following \cite{Zvan1998}, we replace \eqref{eq:AM_HJB} by the nonlinear penalized PDE
\begin{equation}\label{eq:AM_penalty}
P_t + \mathcal{L} P - r_d(t) P  + p\,\max\big[(K-S)^+ - P, 0\big] = 0,
\end{equation}
where $p>0$ is the penalty parameter. As $p \to \infty$, the solution of \eqref{eq:AM_penalty}
converges to the solution of the original HJB problem in \eqref{eq:AM_HJB}.

The \eqref{eq:AM_penalty} has to be solved subject to the terminal condition
\begin{equation}
P(T,S,v) = (K-S)^+,
\end{equation}
and the following boundary conditions set at the truncated spatial interval $(S,v) \in [0,S_{\max}]\times[0,V_{\max}]$. For the boundary condition in $S$ we use a Dirichlet condition at the spot lower boundary $S=0$, \cite{Rouah2013}
\begin{equation}
P(t,0,v)=K,
\end{equation}
and a Neumann condition at the spot upper boundary $S = S_{\max}$
\begin{equation}
P_S(t,S_{\max},v) = 0.
\end{equation}

In the variance domain, at the upper boundary we set a Neumann condition, \cite{Rouah2013}
\begin{equation}
P_v (t,S,v_{\max}) = 0,
\end{equation}
while the low boundary $v = 0$ needs a special treatment. Again, it is well-known that if the Feller condition $2 \kappa \theta > \xi^2$ is satisfied, the process $v_t$ is strictly positive, $v_t \in [0,\infty)$, and the boundary condition at $v=0$ is not required. In other words, the PDE itself with $v=0$ substituted into it serves as an appropriate boundary condition, i.e.
\begin{equation}
P_t + [r_d(t) - r_f(t)] S P_S - r_d(t) P = 0.
\end{equation}
Otherwise, if the Feller condition is violated, the behavior of $v_t$ at the origin should be additionally identified, see e.g., \cite{f54, CarrLinetsky2006, lucic_boundary_2012}, for instance, one can proceed with either absorbing or reflecting boundary condition.

\subsection{Temporal and spatial discretizations}

We discretize the spatial variables $S$ and $v$ on a non-uniform grid. Non-uniform meshes achieve higher accuracy with substantially fewer grid points, reducing computational cost (see, e.g., \cite{haentjens2015adi,ItkinBook}). Grid refinement is concentrated near the strike $S=K$ and along the degenerate boundary $v=0$.

The region near $v=0$ requires particular care: the second-order terms in \eqref{eq:AM_HJB} vanish there, rendering the equation convection-dominated. Local mesh refinement in this region improves both the stability and accuracy of the finite-difference approximation. For our numerical experiments, we truncate the computational domain to $(S,v)\in[0,8K]\times[0,5]$, a choice that renders truncation error negligible.

The construction of the non-uniform grid and the corresponding FD approximations follow \cite{haentjens2015adi} exactly; therefore, these technical details are omitted here.

Once the spatial discretization of the penalized problem \eqref{eq:AM_penalty} is complete, \eqref{eq:AM_penalty} is transformed into a semi-discrete system:
\begin{equation} \label{eq:semi_discrete_system}
\frac{d\mathbf{P}}{dt} = A\mathbf{P} + \mathbf{b} + p(\mathbf{P})\left(\mathbf{P}^* - \mathbf{P}\right),
\end{equation}
where $A$ is a sparse matrix, $\mathbf{P}$ and $\mathbf{b}$ denote the solution vector and the boundary contribution vector, respectively, and $p(\mathbf{P})$ is a diagonal penalty matrix with entries
\begin{equation} \label{eq:penalty_matrix}
p_{ii}(\mathbf{P}) =
\begin{cases}
p, & \text{if } \mathbf{P}_i < \mathbf{P}_i^*,\\
0, & \text{otherwise}.
\end{cases}
\end{equation}
Because the penalty matrix depends on the unknown solution, \eqref{eq:semi_discrete_system} is nonlinear.

To advance the semi-discrete system \eqref{eq:semi_discrete_system} in time, we apply an alternating direction implicit scheme of the modified Craig--Sneyd type. This class of schemes is particularly well suited for multidimensional parabolic problems with mixed derivative terms and has been widely used in the numerical valuation of options under stochastic volatility models\footnote{Note, that the MCS scheme is not positivity preserving due to its explicit treatment of the mixed derivatives term. To get a fully positivity preserving scheme one can use a method developed in \cite{itkin2022lsv}.}.

Let $0 = t_0 < t_1 < \cdots < t_N = T_2$ denote a partition of the time interval $[0, T]$ with time steps $\Delta t^n = t_n - t_{n-1}$, and let $\mathbf{P}^n$ denote the numerical approximation to the semi-discrete solution vector at time $t_n$. Following the spatial discretization described in above, the matrix $A$ is decomposed as
\begin{equation}
A = A_0 + A_1 + A_2,
\end{equation}
where $A_0$ corresponds to the mixed derivative term, while $A_1$ and $A_2$ contain the contributions from the first- and second-order derivatives in the $S$- and $v$-directions, respectively. The vector $\mathbf{b}$ is decomposed analogously as $\mathbf{b} = \mathbf{b}_0 + \mathbf{b}_1 + \mathbf{b}_2$, where each component accounts for the boundary contributions associated with the corresponding operator. The mixed derivative operator $A_0$ is treated explicitly in time, whereas the directional operators $A_1$ and $A_2$ are treated implicitly in an alternating fashion. This splitting yields a sequence of one-dimensional linear systems at each time step.

As shown in \cite{haentjens2015adi}, with an appropriate choice of the MSC parameters, the scheme is second-order accurate in time and unconditionally stable for two-dimensional parabolic problems with mixed derivatives in the time-homogeneous case.

\subsection{Penalty iteration for the ADI-MCS time stepping}

We incorporate the penalty formulation \eqref{eq:AM_penalty} into the ADI-MCS time-stepping scheme described in \cite{heidarpour2018spread}.

As demonstrated in \cite{azimzadeh2016weakly}, direct control (LCP-type) discretizations may lead to singular matrix iterates and potential failure of policy iteration unless problem-dependent modifications of the control set are introduced. In contrast, the penalized formulation yields strictly diagonally dominant matrices and guarantees convergence of Howard's policy iteration without requiring any ad hoc removal of controls. Moreover, numerical evidence indicates that penalized schemes perform at least as efficiently as their direct control counterparts, while being more robust.

To preserve computational efficiency within the ADI framework, the penalty term is introduced only in the final directional step of the ADI corresponding to $\hat{\mathbf{Y}}_2$, see \Cref{alg:penalized_policy_mcs} in \cref{sec:algo}. The resulting nonlinear system at each time step is solved using a penalized policy (Howard) iteration algorithm. Owing to the diagonal structure of the penalty matrix and the tridiagonal structure of the directional operators, each policy iteration remains inexpensive, and typically only a small number of iterations is required for convergence. Additionally, while the MCS method requires solving four tridiagonal systems per time step, the extra cost due to the policy iteration is limited since the penalty matrix is diagonal and the number of iterations is small.

\subsection{Validation}

To validate the penalized policy iteration approach, we developed a comprehensive test suite consisting of 432 parameter configurations. These configurations span strikes $K \in {95, 100, 105}$, maturities $T \in {0.25, 1, 2}$, interest rates $r \in {0.025, 0.05}$, dividend yields $q \in {0.025, 0.05}$, initial variances $v_0 \in {0.05, 0.09}$, and correlations $\rho \in {-0.7, 0, 0.7}$, with the vol-of-vol fixed at $\xi = 0.5$. Implemented on a grid with resolution $N_s = 200$, $N_v = 100$, and $N_t = 200$, the MCS-Policy method yields a European option root mean square error of $9.72 \times 10^{-4}$ and a worst-case error of $3.35 \times 10^{-3}$ relative to semi-analytical benchmarks. For American options, the proposed method and the benchmark exhibit close agreement, evidenced by a mean absolute difference of $2.68 \times 10^{-3}$ and a maximum difference of $1.43 \times 10^{-2}$. The early exercise premium differs by no more than $7.51 \times 10^{-3}$. Collectively, these findings confirm that the MCS-Policy approach generates option values consistent with the established method.

Additionally, we compare our results with those reported in~\cite{haentjens2015adi}. The experiments span a broad range of parameters: mean-reversion speeds $\kappa \in [0.3, 5]$, long-run variances $\theta \in [0.0348, 0.16]$, vol-of-vol $\xi \in [0.04, 1.0]$, correlations $\rho \in [-0.9, 0.6]$, interest rates $r \in [0.01, 0.1]$, and maturities $T \in [0.25, 15]$. The initial variance covers both low- and high-volatility regimes ($v_0 \in {0.05, 0.0625, 0.25}$), including challenging cases that violate the Feller condition.

Across all 31 benchmark cases, the proposed scheme demonstrates good agreement with reference solutions, achieving an overall RMSE of $8.745 \times 10^{-3}$ on a grid with $N_s = 500$, $N_v = 250$, and $N_t = 1000$. In moderate parameter regimes with shorter maturities, accuracy improves substantially, with RMSE values as low as $10^{-4}$. For longer maturities, errors increase, yielding a worst-case error on the order of $10^{-1}$, which is the primary contributor to the overall RMSE. These larger deviations are predominantly associated with the parameter sets in Table~5 of~\cite{haentjens2015adi}, which were originally used in studies of European put options. As shown in our first set of tests, the proposed method converges reliably against available semi-analytical solutions for European options. Nevertheless, computing the early exercise premium (EEP) remains more demanding, leading to notable deviations in certain cases. To the best of our knowledge, these parameter configurations have not been extensively tested for American options prior to this work. By contrast, the results in Tables~2–4 of~\cite{haentjens2015adi} exhibit clear convergence, as these correspond to well-established benchmark cases for American and Bermudan options under the Heston model, where RMSE values consistently reach $10^{-4}$. On average, the nonlinear solver requires approximately two policy iterations per time step.

\section{Numerical example} \label{numEx}

This section presents numerical examples illustrating the methodology from the previous sections. The supporting Python code is available at
\href{https://github.com/rakhymzhan11/DSINC-AMERICAN-HESTON}{Github}.

The primary objective of this section is the computation of the EB by solving \eqref{ES} as this described in \cref{compEB}. The American option of short FF price could be subsequently found using this EB. The procedure involves calculating the process's joint  transition density (as outlined in \cref{tdFin}) and then substituting both the density and the EB into equation \eqref{decompGenN}. Evaluating the resulting integrals provides the European option value and the EEP.

In terms of computational cost, finding $P_E(t,X^n_T,v)$ is far less intensive than finding the EB. Indeed, for the FF contract $X_T^n$ is just a forward value. For American options it can be found by using a marginal CF of $x_T$. This disparity in speed arises because the EB computation relies on an iterative root solver. Therefore, the total computation time is dominated by the EB, and the cost of obtaining $P_E(t,X^n_T,v)$ and two remaining integrals (once the EB is found) is negligible in comparison.

Without loss of generality, the model’s time-dependent coefficients are defined as
\begin{align} \label{ex}
\theta(t) &= B_\theta e^{- C_\theta (t-t)} + A_\theta, \qquad
\xi(t) = \sum_{i=1}^N \mathbf{1}_{t < t_i} \xi_i, \qquad
\sigma(t) = \sum_{i=1}^N \mathbf{1}_{t < t_i} \rho_i, \\
r_d(t) &= B_\rd e^{- C_\rd t} + A_\rd, \qquad
r_f(t) = B_\rf e^{- C_\rf t} + A_\rf, \nonumber
\end{align}
where $A_k, B_k, C_k, \, k \in [\theta, r_d, r_f]$, $\xi_i, \rho_i$, and $t_i$ (for $i=1,\dots,N$) are constants. The specific values of these parameters, along with other test inputs, are provided in Table~\ref{tab1}. Their temporal evolution is illustrated in Fig.~\ref{FinputData}.
\begin{table}[!htb]
\begin{center}
\begin{tabular}{|c|c|c|c|c|c|c|c|c|c|c|c|c|c|c|c|}
\hline
\rowcolor[rgb]{ .792,  .929,  .984}
$A_\theta$ & $B_\theta$ & $C_\theta$ & $A_\rd $ & $B_\rd$ & $C_\rd$ &
$A_\rf$ & $B_\rf$ & $C_\rf$ & $\kappa$ & $v_0$ & $T$ & $T_1$ & $N$ & $K$ & $S$ \\
\hline
0.4 & 0.2 & -5.0 & 0.02 & 0.01 & -1.0 & 0.015 & 0.01 & -0.5 & 0.5 & 0.5 & 0.25 & 1/12 & 3 & 0.5 &  0.5 \\
\hline
\end{tabular}
\vspace{\baselineskip}

\begin{tabular}{|c|r|r|r|}
\hline
\rowcolor[rgb]{ .792,  .929,  .984}
$i$ & 1 & 2 & 3 \\
\hline
$t_i$ & 1/12 & 2/12 & T \\
\hline
$\xi_i$  & 0.80 & 0.70 & 0.60 \\
\hline
$\rho_i$  & -0.40 & -0.45 & -0.50 \\
\hline
\end{tabular}
\caption{Model parameters for the numerical test 1.}
\label{tab1}
\end{center}
\end{table}

\begin{figure}[!htp]
\begin{center}
\hspace*{-0.3in}
\includegraphics[width=0.6\textwidth]{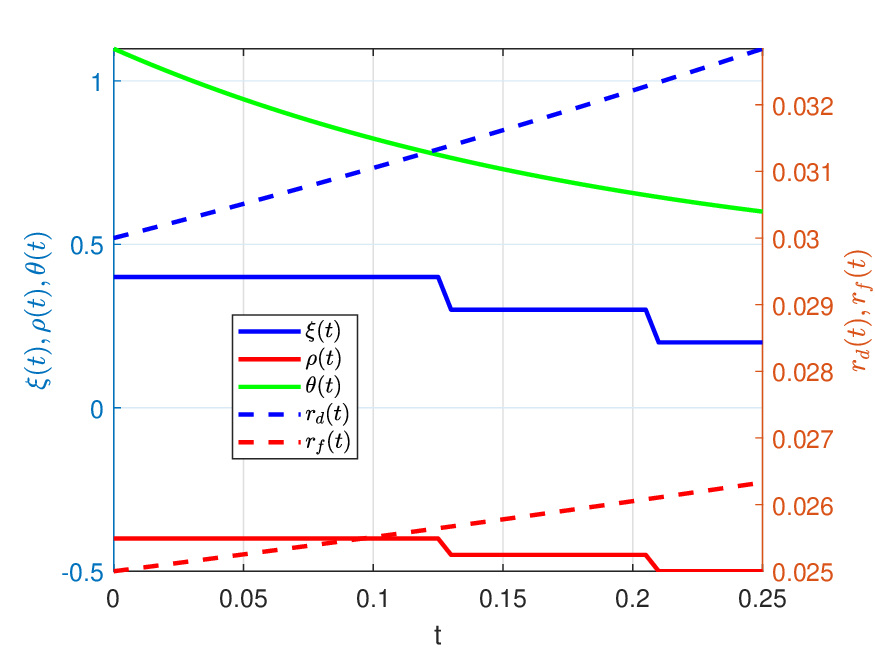}
\end{center}
\caption{Time-dependent parameters of the Heston model used in the numerical example.}
\label{FinputData}
\end{figure}

To better resolve the sharp drop in the EB over the interval $t \in [T^- \Delta t, T^-]$, we use a non-uniform grid in time
\begin{equation}
\Omega := t_{i+1} = T \dfrac{e^{\beta \nu_i} - 1}{e^{\beta} - 1}, \quad \nu_i = i/(N_t-1), \quad i=0,\ldots,N_t-1,
\end{equation}
with $N_t=51$ and $\beta=0.1$.

All computations were performed in Matlab on a system equipped with two Intel Quad-Core i7‑4790 CPUs operating at 3.80 GHz. We solved the system of nonlinear equations for $x^*(t, v_t)$  using either the Matlab solver \texttt{fzero} or \texttt{lsqnonlin}, with the latter yielding better performance for our problem.
Because all components of the system are available in analytical form, we also computed the solution gradients analytically. The solver typically converges to machine precision within 10--13 iterations, with an elapsed time of approximately 42 ms per time step for a FF contract and 44 ms for American options\footnote{Since the EEP is generally small for the FF contracts, this requires a higher tolerance of the solver to resolve it. Therefore, despite getting the European price of the FF doesn't take time, the elapsed time in both cases is close to each other.}.

\subsection{American options}

To demonstrate the applicability of our approach to American options, the next tests use $S = 100$, $K = 100$ and the same model parameters as in Table~\ref{FinputData} with $N_x = N_v = 64$. We mark the results obtained by the IE+COS approach as IECOS, and those by the IE+DSINC method - as DSINC.

\paragraph{Test 1.} In this test, we check whether IECOS method can replicate the Black-Scholes American Put price when the Heston model parameters take extreme values that reduce the model to the constant-coefficient Black-Scholes case.
Accordingly, we set the Heston model parameters to: $\kappa = 200$, $v_0 = \theta = 0.05$, $\xi = 0.001$, $\rho = 0$. This parameter set corresponds to a nearly constant volatility, effectively reducing the Heston model to a Black–Scholes framework with variance $v_0$. Hence, the American option price should match the Black–Scholes price.

The BS European price obtained by using two different methods (FD and quadratures) is 4.36565, which coincides with the Black-Scholes Put price with $\sigma = \sqrt{v_0}$. The FD EEP is 0.119 and the American price is 4.37764. By using IECOS method, we obtain an American Put price of 4.3770 consisting of a European price of 4.3656 and an early exercise premium (EEP) of 0.0115. This result confirms the robustness of our method, even in this very limiting case. To notice, in this particular case, we have $\Fe$ = 2.e7, so the Feller ratio is high, and the distribution of the log-asset price is much more "normal" and tightly clustered around the mean. This justifies our choice of $L_{x,b} = 25$, otherwise the truncation interval is too large, and the COS method may struggle to resolve the sharp peak of the density, leading to oscillations), or to increase $N_x$ (e.g., to $N_x = 128$).

As mentioned in \cref{truncation}, the method's robustness may dependent on the range of parameter values. In particular, when the variance varies only slightly - a regime close to the Black–Scholes model — the following numerical issues may arise:
\begin{enumerate}[itemsep=8pt]
\item In all Black--Scholes--like tests, the Feller condition is usually violated. As a result, the variance distribution develops a spike near zero, leading to heavy oscillations close to the zero boundary (Gibbs phenomenon) in the Fourier representation.

\item Let us assume $\xi = 0.01$). The statistical spread of the Heston variance process with such a small $\xi$ is very narrow, implying a small truncation interval. Yet the numerical requirements of the COS method differ: when $\xi$ is very small, the variance process resembles a sharp spike (approaching a Dirac delta) centered at $v_0$. In the Fourier domain, such a spike corresponds to a slowly decaying CF.

    The COS method samples the CF at frequencies $u_k = k\pi/(b - a)$. If $b-a$ is too small, the frequency step $\Delta u$ becomes large, causing sparse sampling that may miss the curvature of the CF --- an aliasing effect. To maintain accuracy as $\xi \to 0$, one must therefore use a wider interval than the statistical moments suggest, ensuring sufficiently dense frequency sampling.

\item Widening $b$ too far, however, introduces a trade-off typical of the COS method: we are forced to balance truncation error against integration (sampling) error, effectively facing a Nyquist--Shannon type constraint. In the COS method, the number of terms $N$ and the interval width $b-a$ are linked through the highest sampled frequency $u_N = N\pi/(b - a)$.
    Increasing $b$ without increasing $N$ makes the sampling coarser ($\Delta u$ grows), which can fail to resolve high-frequency oscillations in the CF of the Heston model. Conversely, reducing $b$ to its statistically justified width keeps $\Delta u$ small but may cut off significant tails of the density, leading to truncation error.

\item Because small $\xi$ produce a sharply peaked variance distribution, its Fourier-series representation suffers from Gibbs-type oscillations, akin to approximating a discontinuous profile with smooth waves, resulting in \emph{ringing} at the edges.
\end{enumerate}

\paragraph{Test 2.} In this test we also use IECOS and set $\kappa = 0.5, \xi = 0.8, \theta = v_0 = 0.05$ and vary correlation, so $\rho \in [-1,1]$. The results are presented in Table~\ref{corr} and coincide with those computed by the FD method with an accuracy better than 1\%.
\begin{table}[htbp]
\centering
\scalebox{0.8}{
\begin{tabular}{|l|r|r|r|r|r|r|r|r|r|r|r|}
\toprule
\rowcolor[rgb]{ .949,  .808,  .937} $\rho$ & \textbf{-1} & \textbf{-0.9} & \textbf{-0.7} & \textbf{-0.5} & \textbf{-0.3} & \textbf{0} & \textbf{0.3} & \textbf{0.5} & \textbf{0.7} & \textbf{0.9} & \textbf{1} \\
\hline
\textbf{EEB} & 0.011762 & 0.012141 & 0.012945 & 0.01375 & 0.014594 & 0.015987 & 0.017565 & 0.018744 & 0.020061 & 0.021513 & 0.021749 \\
\bottomrule
\end{tabular}%
}
\caption{The EEB as a function of $\rho$ with $\theta = v_0 = 0.05, \xi = 0.8, \kappa = 0.5$, other parameters as in Table~\ref{FinputData}.}
\label{corr}%
\end{table}%

\paragraph{Test 3.} In this test we us DSINC and assume constant model parameters, so set $C_\rd = C_\theta = C_\theta = 0$, and set $\xi_2 = \xi_1, \xi_3 = \xi_1, \rho_3 = \rho_1, \rho_2 = \rho_1$ and $A_{\rd} =  0.03, A_{\rf} = 0.025$. The computed EB is presented in \cref{amer1eb}(a). Note, that we assume $r_d >0,  r_f > 0$ as otherwise multiple exercise boundaries could occur which we don't consider in this paper (in more detail, see \cite{AndersenLake2021,ItkinKitapbayev2025r}). The American Put option price is 13.6749, which consists of the European Put price 13.6633 plus the EEP of 0.01157.

Since this test employs constant coefficients, the European price component can be validated against the standard Heston COS method, which also returns a value of 13.6633. The two results agree to within 1e-7. The distinction lies in how the CF (CF) is computed: our method calculates the joint CF of $(x_t, v_t)$ step-by-step in time using \eqref{IntB} (with $u_2 = 0$), whereas the standard COS method uses the closed-form solution for the CF at maturity $t=T$.

As noted in the Introduction, previous studies, \cite{tzavalis2003pricing, Chiarella2005} building on empirical observations from \cite{broadie2000american}, approximated the early-exercise surface $S^*(\tau,v)$ as a linear function of $v$, specifically $S^*(\tau,v) = b_0(\tau) + b_1(\tau) v$. To examine this assumption, Fig.~\ref{amer1eb}(b) presents $S^*(t_i, v)$ for multiple time points $t_i$ derived from our experiments. While the relationship appears approximately linear in $v$ at certain times, it clearly deviates from linearity at others. Moreover, given the time-dependent nature of the parameters, any linear approximation would be at best piecewise over distinct time intervals.

For a better visibility, we introduce a measure of average convexity of the EB which is
\begin{equation}
\mathcal{N}(t) = \frac{\int_{v_{\min}}^{v_{\max}} |S^*(t,v)) - L(t,v)| \, dv}{(v_{\max} - v_{\min}) \left( \max_{v} S^*(t,v) - \min_{v} S^*(t,v) \right)},
\end{equation}
where $L(t,v)$ is a linear approximation (a straight line which connects $S^*(t,v_{\max})$ and $S^*(t,v_{\min})$.  This measure has several advantages, namely: a) For a linear dependence $\mathcal{N}(t) = 0$; b) It is scale invariant and dimensionless; c) It can handle any shape; d) It will judge an upward-bowing curve and a downward-bowing curve of the same magnitude as identically nonlinear; e) Unlike the second derivative, it does not blow up at the boundaries, and unlike $R^2$, it does not rely on the variance of the curve, which is why it preserves the symmetry.

\Cref{convexity} presents $\mathcal{N}(t)$ as a function of time obtained in this test. To see the impact of $v_0$ and $\rho$ on this convexity, we also run same test but with $v = 0.05$ (panel 'a') and $\rho = 0.8$ (panel 'b'). It can be seen that the impact is small while not negligible.
\begin{figure}[!htp]
\begin{center}
\hspace*{-0.3in}
\subfloat[]{\includegraphics[width=0.5\textwidth]{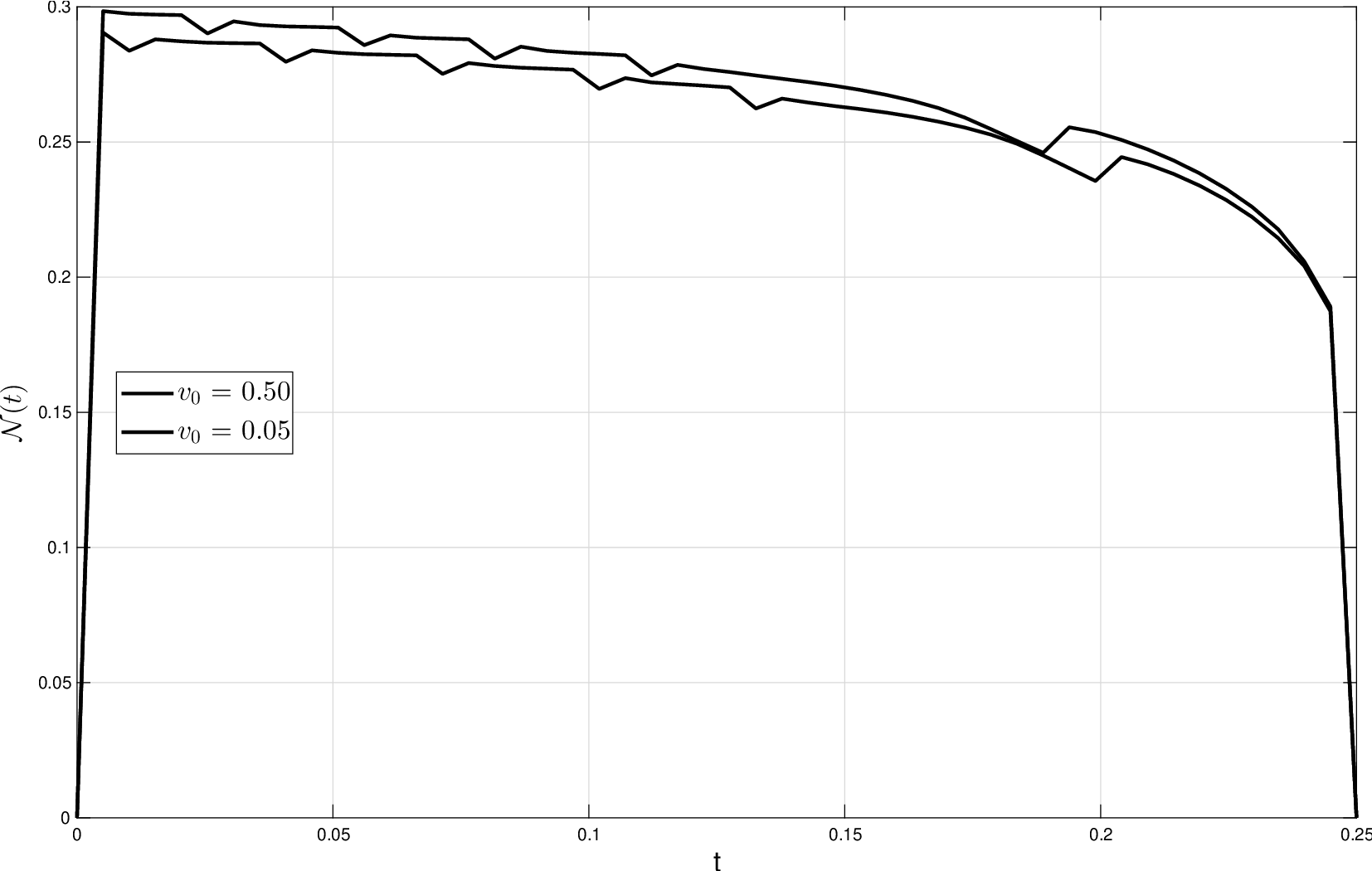}}
\subfloat[]{\includegraphics[width=0.5\textwidth]{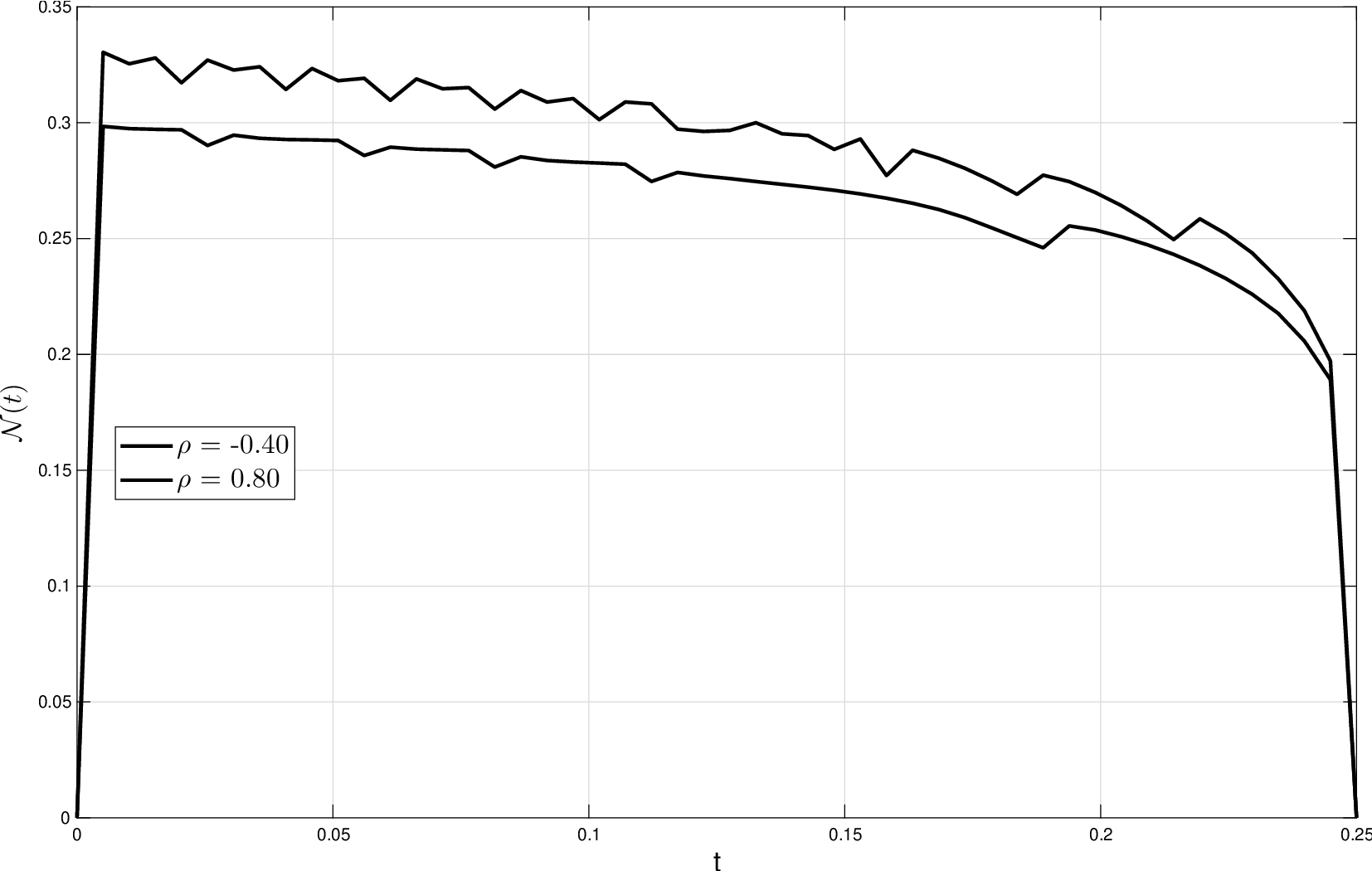}}
\end{center}
\caption{Convexity $\mathcal{N}(t)$ obtained in Test~3 for two values of $v_0$:  0.5 and 0.05 - panel (a); and two values of $\rho$: -0.4 and 0.8 - panel (b).}
\label{convexity}
\end{figure}

\paragraph{Test 4.} In this test, we restore time-inhomogeneity by reverting to the parameters listed in Table~\ref{FinputData}. The computed ES is relatively close to that from the previous test; therefore, rather than showing absolute values, we present the absolute difference between the results obtained here and those from the constant-parameters test. Accordingly, Fig.~\ref{amer2eb}(a) displays the difference in the two ES, while Fig.~\ref{amer2eb}(b) shows $\Delta S^*(t_i, v)$ for multiple time points $t_i$ derived from this experiment. The American Put option price is 13.57417, which consists of the European Put price 13.5647 plus the EEP of 0.00944.

\begin{figure}[!htp]
\begin{center}
\hspace*{-0.3in}
\subfloat[]{\includegraphics[width=0.5\textwidth]{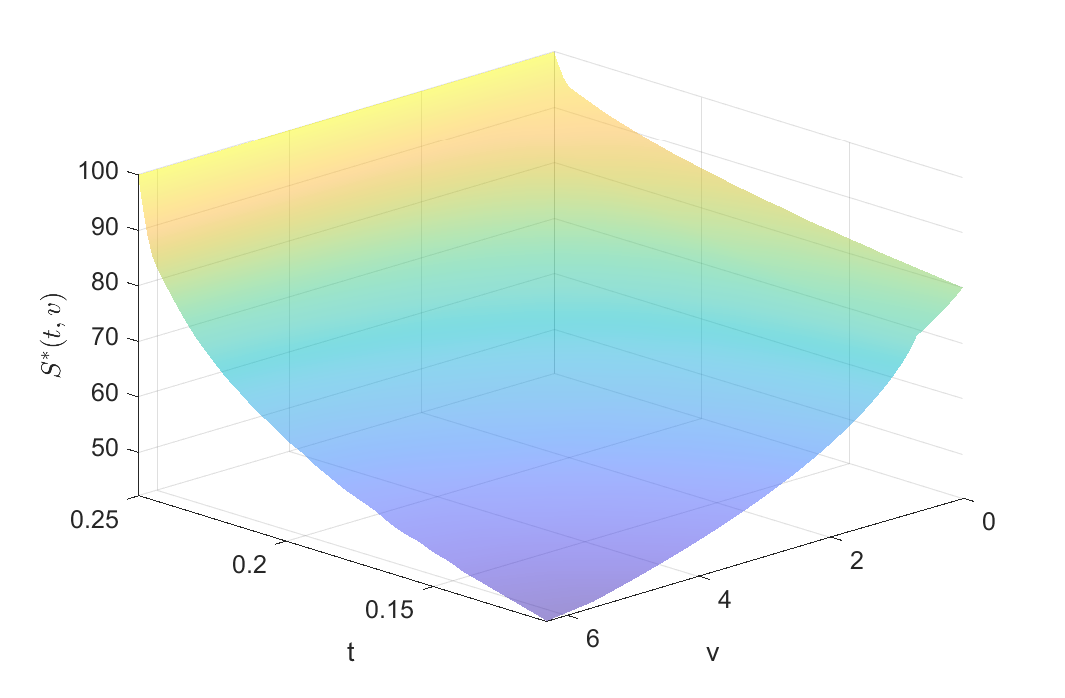}}
\subfloat[]{\includegraphics[width=0.5\textwidth]{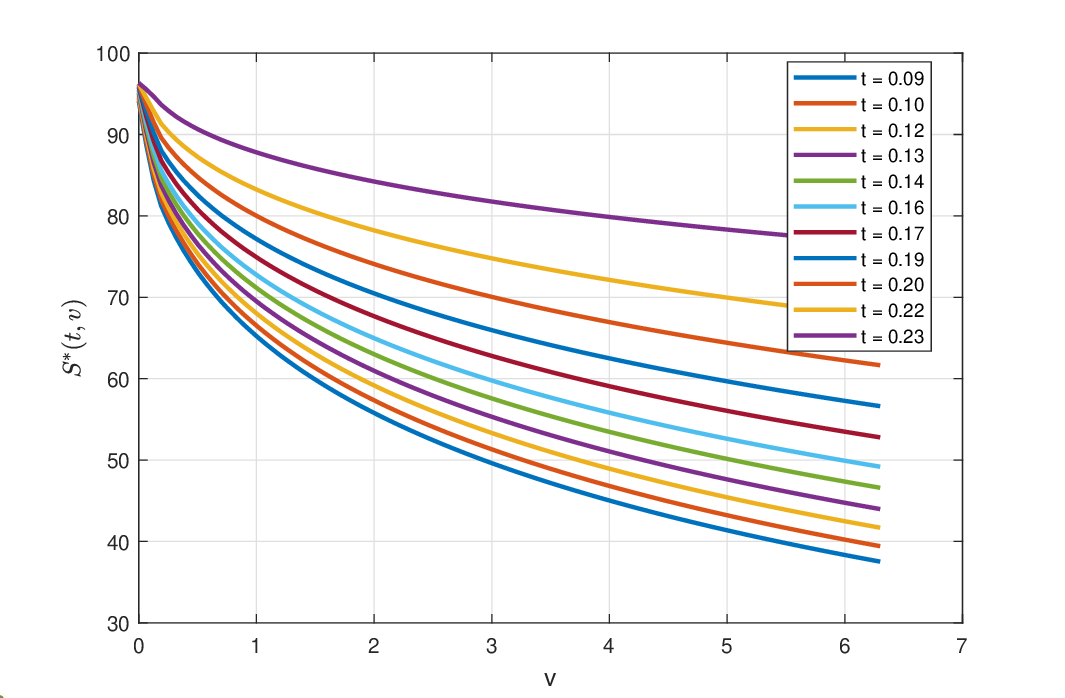}}
\end{center}
\caption{Early exercise boundaries for an American Put option under the time-homogeneous Heston model with $S = 100, K = 100, v_0 = 0.5$ and model parameters from Table~\ref{FinputData} but constant: a) the entire ES, b) projections of the ES for various time points $t_i$.}
\label{amer1eb}
\end{figure}

\begin{figure}[!htp]
\begin{center}
\hspace*{-0.3in}
\subfloat[]{\includegraphics[width=0.5\textwidth]{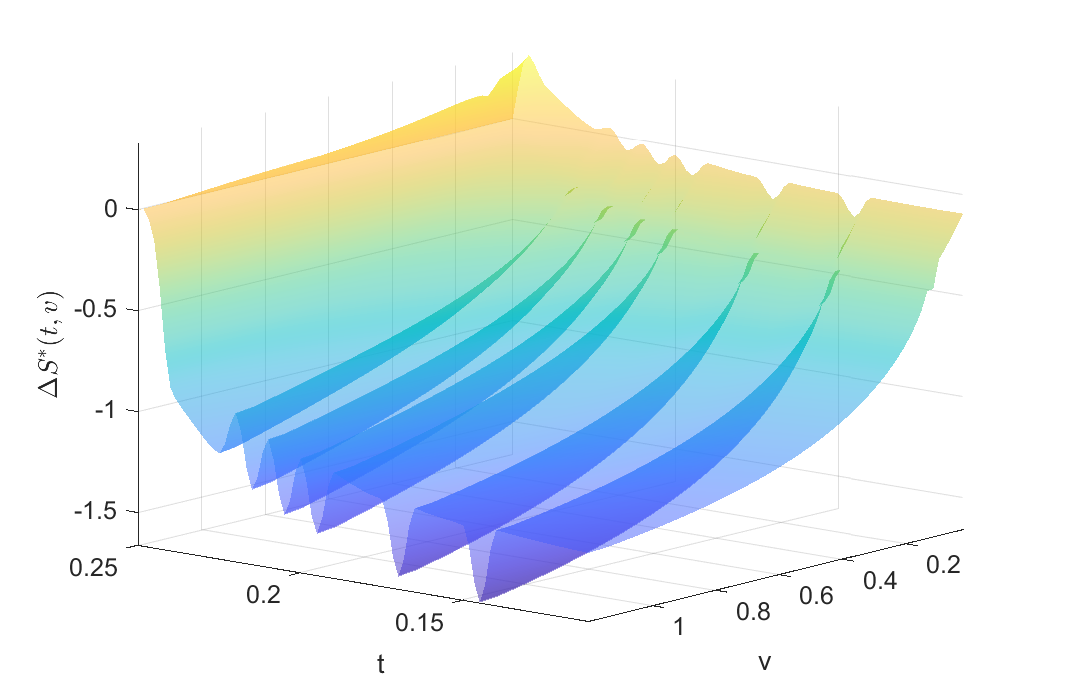}}
\subfloat[]{\includegraphics[width=0.5\textwidth]{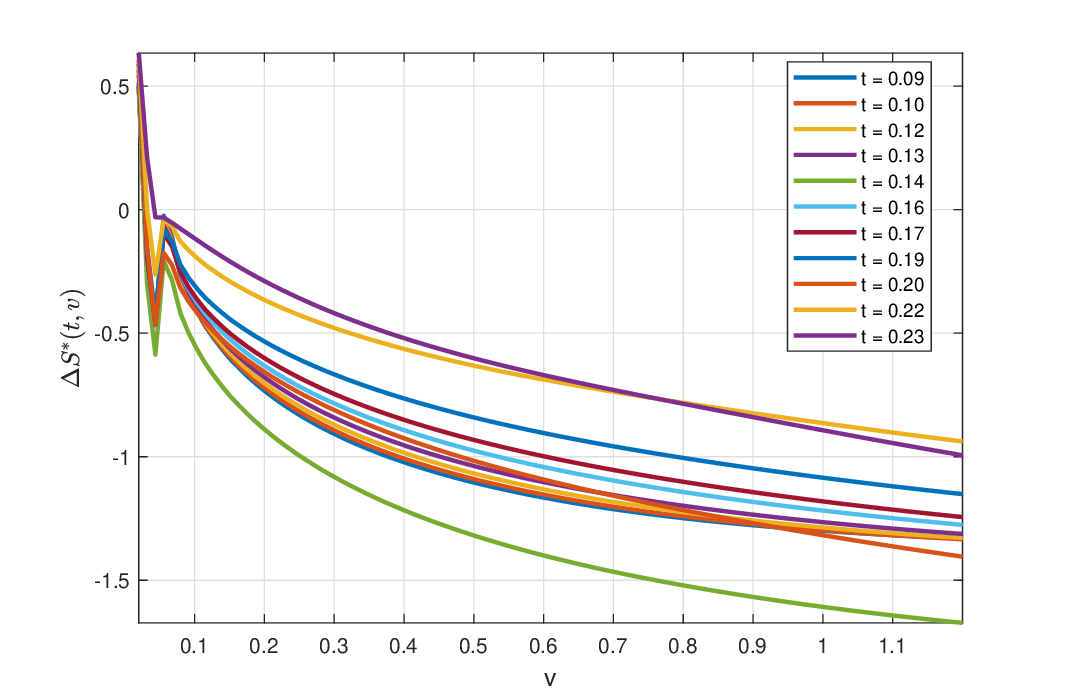}}
\end{center}
\caption{Difference in the early exercise boundaries for an American Put option under the time-inhomogeneous and time-homogeneous Heston models with $S = 100., K = 100., v_0 = 0.5$ and model parameters from Table~\ref{FinputData} a) the entire ES, b) projections of the ES for various time points $t_i$.}
\label{amer2eb}
\end{figure}

Table~\ref{compFD} presents a comparison between these results and those obtained using the FD solver described in \cref{pdeSol}.
\begin{table}[htbp]
\centering
\scalebox{0.88}{
\begin{tabular}{|l|r|r|r|r|r|l|}
\toprule
\rowcolor[rgb]{ .792,  .929,  .984} \multirow{2}[4]{*}{} & \multicolumn{2}{c|}{\textbf{t-homogeneous}} & \multicolumn{2}{c|}{\textbf{t-inhomogeneous}} & \multicolumn{1}{c|}{\multirow{2}{*}{\textbf{Elapsed time,sec}}} & \multicolumn{1}{c|}{\multirow{2}{*}{\textbf{Grid}}} \\
\cmidrule{2-5}    \rowcolor[rgb]{ .792,  .929,  .984}       & \multicolumn{1}{l|}{\textbf{European}} & \multicolumn{1}{l|}{\textbf{American}} & \multicolumn{1}{l|}{\textbf{European}} & \multicolumn{1}{l|}{\textbf{American}} &       &  \\
\hline
\rowcolor[rgb]{ .557,  .851,  .451} IECOS  & 13.6633 & 13.6749 & 13.5647 & 13.5742 & 2.1   & Nt = 51, Nx = Nv = 64 \\
\hline
FD1   & 13.6187 & 13.6413 & 13.5925 & 13.6181 & 0.91 - 0.96 & Nt = 51, Nx = 101, Nv = 79 \\
\hline
\rowcolor[rgb]{ .557,  .851,  .451} FD2   & 13.6425 & 13.6651 & 13.5755 & 13.6006 & 7.2 - 7.8 & Nt = 101, Nx = 201, Nv = 101 \\
\hline
FD3   & 13.6611 & 13.6837 & 13.5732 & 13.5980 & 27    & Nt = 201, Nx = 201, Nv = 201 \\
\bottomrule
\end{tabular}%
}
\caption{Comparison of results from the IECOS and FD solvers across several FD grid parameter sets. Also, for a time-homogeneous case, the DSINC method yields the same European price, however, its American price of 13.6845 demonstrates closer agreement with the FD reference benchmark than the IECOS result.}
\label{compFD}%
\end{table}%
Recall that for time-homogeneous parameters, the IECOS approach almost exactly (up to 1.e-7) reproduces the European Put price obtained via the COS method.

As illustrated in the table, the FD solution lacks accuracy on sparse grids. While it eventually converges toward the IE solution on denser grids, achieving this level of precision requires substantial computational overhead. These results highlight a distinct performance advantage for the IECOS method. This disparity is even more pronounced in time-inhomogeneous cases. The performance gap arises from a fundamental difference in how the models handle time-dependent parameters: by nature, the FD approach approximates the long-term mean $\theta(t)$ as a piecewise constant function, whereas in the IECOS approach we treat it as a continuous function, leading to a more refined and stable representation.

\paragraph{Test 5.}

\begin{table}[htbp]
\centering
\scalebox{0.85}{
\begin{tabular}{|l|r|r|r|r|r|l|}
\toprule
\rowcolor[rgb]{ .792,  .929,  .984} \multirow{2}{*}{\textbf{Method}} & \multicolumn{2}{c|}{\textbf{t-homogeneous}} & \multicolumn{2}{c|}{\textbf{t-inhomogeneous}} & \multicolumn{1}{c|}{\multirow{2}{*}{\textbf{Time, sec}}} & \multicolumn{1}{c|}{\multirow{2}{*}{\textbf{Grid}}} \\
\cmidrule{2-5}    \rowcolor[rgb]{ .792,  .929,  .984}       & \multicolumn{1}{l|}{\textbf{European}} & \multicolumn{1}{l|}{\textbf{American}} & \multicolumn{1}{l|}{\textbf{European}} & \multicolumn{1}{l|}{\textbf{American}} &       &  \\
\hline
IECOS   & 13.6633 & 13.6749 & 13.5647 & 13.5742 & 2.1   & $N_t{=}51,\ N_x{=}N_v{=}64$ \\
\hline
\rowcolor[rgb]{ .557,  .851,  .451} DSINC & 13.6633 & 13.6845 & 13.5757 & 13.6001 & 2.0 & $N_t{=}24,\ N_v{=}16,\ N_\omega{=}600$ \\
\hline
FD1   & 13.6187 & 13.6413 & 13.5925 & 13.6181 & 0.91 - 0.96 & $N_t{=}51,\ N_x{=}101,\ N_v{=}79$ \\
\hline
\rowcolor[rgb]{ .557,  .851,  .451} FD2   & 13.6425 & 13.6651 & 13.5755 & 13.6006 & 7.2 - 7.8 & $N_t{=}101,\ N_x{=}201,\ N_v{=}101$ \\
\hline
FD3   & 13.6611 & 13.6837 & 13.5732 & 13.5980 & 27    & $N_t{=}201,\ N_x{=}201,\ N_v{=}201$ \\
\bottomrule
\end{tabular}%
}
\caption{Comparison of the IECOS, DSINC, and FD policy-iteration ADI methods for an American Put under the time-homogeneous and time-inhomogeneous Heston model ($S=K=100$, $v_0=0.5$, $\kappa=0.5$, $T=0.25$; remaining parameters as in Table~\ref{FinputData}). FD3 is the densest (reference) grid.}
\label{compDSINC}%
\end{table}%

\begin{table}[htbp]
\centering
\begin{tabular}{|l|r|r|r|r|r|}
\toprule
\rowcolor[rgb]{ .792,  .929,  .984} \textbf{Method \textbackslash{} S} & \textbf{8} & \textbf{9} & \textbf{10} & \textbf{11} & \textbf{12} \\
\hline
\rowcolor[rgb]{ .792,  .929,  .984} \textbf{MCS-IT} & \cellcolor[rgb]{ 1,  1,  1}2.0788 & \cellcolor[rgb]{ 1,  1,  1}1.3339 & \cellcolor[rgb]{ 1,  1,  1}0.7962 & \cellcolor[rgb]{ 1,  1,  1}0.4486 & \cellcolor[rgb]{ 1,  1,  1}0.2433 \\
\hline
\rowcolor[rgb]{ .792,  .929,  .984} \textbf{our FD} & \cellcolor[rgb]{ 1,  1,  1}2.0777 & \cellcolor[rgb]{ 1,  1,  1}1.3331 & \cellcolor[rgb]{ 1,  1,  1}0.7956 & \cellcolor[rgb]{ 1,  1,  1}0.4480 & \cellcolor[rgb]{ 1,  1,  1}0.2426 \\
\hline
\rowcolor[rgb]{ .792,  .929,  .984} \textbf{IECOS} & \cellcolor[rgb]{ 1,  1,  1}2.0811 & \cellcolor[rgb]{ 1,  1,  1}1.3286 & \cellcolor[rgb]{ 1,  1,  1}0.8048 & \cellcolor[rgb]{ 1,  1,  1}0.4449 & \cellcolor[rgb]{ 1,  1,  1}0.2473 \\
\hline
\rowcolor[rgb]{ .557,  .851,  .451} \textbf{DSINC} & \cellcolor[rgb]{ 1,  1,  1}2.0779 & \cellcolor[rgb]{ 1,  1,  1}1.3327 & \cellcolor[rgb]{ 1,  1,  1}0.7954 & \cellcolor[rgb]{ 1,  1,  1}0.4480 & \cellcolor[rgb]{ 1,  1,  1}0.2426 \\
\bottomrule
\end{tabular}%
\caption{American Put prices for the \cite{haentjens2015adi} test ($\kappa=5$, $\theta=0.16$, $\xi=0.9$, $\rho=0.1$, $r_d=0.1$, $r_f=0$, $T=0.25$, $K=10$, $v_0=0.25$): the MCS-IT benchmark and our FD policy-iteration ADI solver compared with the IECOS and DSINC methods.}
\label{IntHoutDSINC}%
\end{table}%

Table~\ref{compDSINC} compares the three methods on the benchmark American Put ($S=K=100$). In the time-homogeneous case the IECOS and DSINC European values
coincide to about $10^{-7}$, both equal to $13.6633$. In the time-inhomogeneous
case the DSINC European value $13.5757$ lies closer to the dense FD grids, which
give $13.5732$ to $13.5755$, than the IECOS value $13.5647$. The American value
separates the methods more sharply. The IECOS early-exercise premium is acceptable
under time-homogeneous parameters, where its American price differs from the
densest grid FD3 by about $9\times 10^{-3}$, but it degrades markedly once the
parameters are time-dependent, where it is off by about $2.4\times 10^{-2}$.
DSINC matches FD3 to about $8\times 10^{-4}$ in the homogeneous case and about
$2\times 10^{-3}$ in the inhomogeneous case, roughly an order of magnitude closer
to the reference than IECOS. It does so in about $2$~s, that is, at IECOS speed and
roughly $13\times$ faster than FD3, which takes $27$~s. The FD solver reaches
comparable accuracy only on its densest grid, at substantial cost. On coarse
grids it is fast but inaccurate.

A comprehensive study on the valuation of American Put options under the Heston model is presented in \cite{haentjens2015adi}, where the authors benchmarked several modern finite difference schemes against existing results in the literature. We adopt a subset of their test cases to evaluate the accuracy and computational performance of our method against these well-established benchmarks.

Table~\ref{IntHoutDSINC} cross-validates against one of these test of  \cite{haentjens2015adi}. DSINC reproduces both the MCS-IT benchmark and our
FD policy-iteration ADI prices to within about $5\times 10^{-4}$ across the full
moneyness range $S\in\{8,\dots,12\}$, including the deep-in-the-money case $S=8$,
whereas the IECOS prices deviate by up to about $9\times 10^{-3}$. DSINC attains
this in about $1.5$~s, comparable to the reported FD timing of about $33$~ms per
time step.

Another test case from \cite{haentjens2015adi} examines regimes in which the option value lies close to its intrinsic value. In this setting, IECOS produces unsatisfactory results that deviate significantly from the true values and are thus impractical, whereas DSINC maintains a high level of accuracy.

\begin{table}[htbp]
  \centering
  \scalebox{0.9}{
    \begin{tabular}{|r|r|r|r|r|r|r|r|r|r|r|r|}
    \toprule
    \rowcolor[rgb]{ .58,  .863,  .973} \multicolumn{1}{|c|}{$\bm{S_0}$} & \multicolumn{1}{c|}{$\bm{K}$} & \multicolumn{1}{c|}{$\bm{T}$} & \multicolumn{1}{c|}{$\bm{r_d-r_f}$} & \multicolumn{1}{c|}{$\bm{\kappa}$} & \multicolumn{1}{c|}{$\bm{\theta}$} & \multicolumn{1}{c|}{$\bm{\xi}$} & \multicolumn{1}{c|}{$\bm{\rho}$} & \multicolumn{1}{c|}{$\bm{v_0}$} & \multicolumn{1}{c|}{\textbf{ref.}} & \multicolumn{1}{c|}{\textbf{DSINC\_M>1}} & \multicolumn{1}{c|}{\textbf{abs\_err}} \\
    \midrule
    \rowcolor[rgb]{ .514,  .886,  .557} 8     & 10    & 0.25  & 0.1   & 5     & 0.16  & 0.9   & 0.1   & 0.0625 & 2     & 2     & 0.00E+00 \\
    \hline
    \rowcolor[rgb]{ .514,  .886,  .557} 9     & 10    & 0.25  & 0.1   & 5     & 0.16  & 0.9   & 0.1   & 0.0625 & 1.1076 & 1.10752 & -8.33E-05 \\
    \hline
    \rowcolor[rgb]{ .514,  .886,  .557} 10    & 10    & 0.25  & 0.1   & 5     & 0.16  & 0.9   & 0.1   & 0.0625 & 0.5199 & 0.520404 & 5.04E-04 \\
    \hline
    \rowcolor[rgb]{ .514,  .886,  .557} 11    & 10    & 0.25  & 0.1   & 5     & 0.16  & 0.9   & 0.1   & 0.0625 & 0.2135 & 0.214013 & 5.13E-04 \\
    \hline
    \rowcolor[rgb]{ .514,  .886,  .557} 12    & 10    & 0.25  & 0.1   & 5     & 0.16  & 0.9   & 0.1   & 0.0625 & 0.082 & 0.0821013 & 1.01E-04 \\
    \hline
    \rowcolor[rgb]{ 1,  1,  0} 8     & 10    & 0.25  & 0.1   & 5     & 0.16  & 0.9   & 0.1   & 0.25  & 2.0785 & 2.0794 & 8.96E-04 \\
    \hline
    \rowcolor[rgb]{ 1,  1,  0} 9     & 10    & 0.25  & 0.1   & 5     & 0.16  & 0.9   & 0.1   & 0.25  & 1.3336 & 1.33384 & 2.44E-04 \\
    \hline
    \rowcolor[rgb]{ 1,  1,  0} 10    & 10    & 0.25  & 0.1   & 5     & 0.16  & 0.9   & 0.1   & 0.25  & 0.7959 & 0.796246 & 3.46E-04 \\
    \hline
    \rowcolor[rgb]{ 1,  1,  0} 11    & 10    & 0.25  & 0.1   & 5     & 0.16  & 0.9   & 0.1   & 0.25  & 0.4482 & 0.448584 & 3.84E-04 \\
    \hline
    \rowcolor[rgb]{ 1,  1,  0} 12    & 10    & 0.25  & 0.1   & 5     & 0.16  & 0.9   & 0.1   & 0.25  & 0.2427 & 0.243014 & 3.14E-04 \\
    \hline
    \rowcolor[rgb]{ .757,  .941,  .784} 90    & 100   & 0.25  & 0.04  & 1.15  & 0.0348 & 0.39  & -0.64 & 0.0348 & 10.0039 & 10    & -3.90E-03 \\
    \hline
    \rowcolor[rgb]{ .757,  .941,  .784} 100   & 100   & 0.25  & 0.04  & 1.15  & 0.0348 & 0.39  & -0.64 & 0.0348 & 3.2126 & 3.20262 & -9.98E-03 \\
    \hline
    \rowcolor[rgb]{ .757,  .941,  .784} 110   & 100   & 0.25  & 0.04  & 1.15  & 0.0348 & 0.39  & -0.64 & 0.0348 & 0.9305 & 0.927059 & -3.44E-03 \\
    \bottomrule
    \end{tabular}%
}
\caption{Comparison of the results in \cite{haentjens2015adi} (Tables~2,3,4) in the "ref." column with those produced by the DSINC method.}
\label{tab:addlabel}%
\end{table}%

To evaluate robustness beyond individual test cases, we tested all three methods on a suite of 432 American Put parameter sets, using the densest finite difference (FD) grid as a benchmark. The performance contrast is striking. The IECOS method yields an absolute error exceeding $10^{-2}$ in 84\% of the cases, with a median error of $7.2 \times 10^{-2}$ and a worst-case error of 1.46. In contrast, the DSINC method maintains an error below $10^{-2}$ in 71\% of the cases and never exceeds $8 \times 10^{-2}$, achieving a median error of $4.3 \times 10^{-3}$. This represents a median accuracy improvement of a factor of 12, scaling to roughly a 250-fold improvement in the tail of the distribution. The IECOS error is concentrated in the low-Feller, high-vol-of-vol regime ($2\kappa\theta/\xi^2 \lesssim 1$), where the near-singular variance density induces Gibbs ringing in the cosine expansion. The localized DSINC basis is insensitive to this effect and remains accurate throughout. Finally, both spectral methods price a single case in about 2 seconds, circumventing the significantly higher computational cost of the FD reference.

\subsection{Short FF contract}

We compute the EB for the short FF contract using the same method as for an American Put, given their structural similarities aside from payoff and exercise window.

\paragraph{Test 1.} We set $S = K = 80$ chosen to approximate an RUB/USD exchange rate. Also, we change $v_0$ to $0.05$ which better represent typical variances of these contracts.
\begin{figure}[!htp]
\begin{center}
\hspace*{-0.3in}
\subfloat[]{\includegraphics[width=0.5\textwidth]{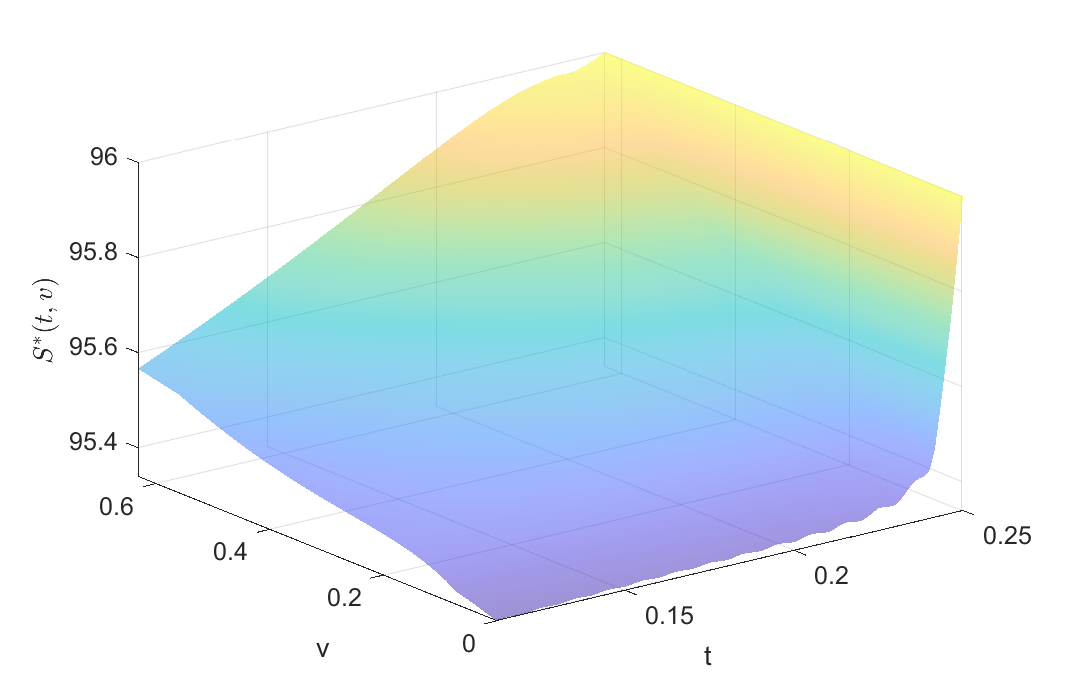}}
\subfloat[]{\includegraphics[width=0.5\textwidth]{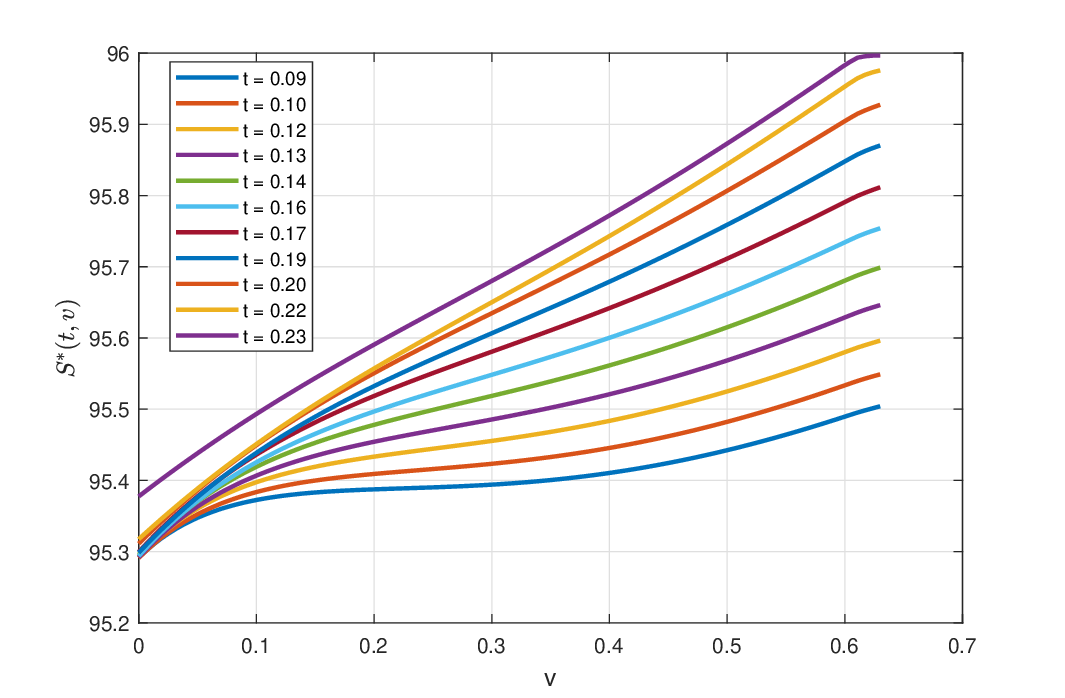}}
\end{center}
\caption{Early exercise boundaries for a short FF contract under the time-homogeneous Heston model with $S = 80, K = 80, v_0 = 0.5$ with model parameters from Table~\ref{FinputData}, but $v_0 = 0.05$; a) the entire ES, b) projections of the ES for various time points $t_i$.}
\label{eb1}
\end{figure}

\begin{figure}[!htp]
\begin{center}
\hspace*{-0.3in}
\subfloat[]{\includegraphics[width=0.5\textwidth]{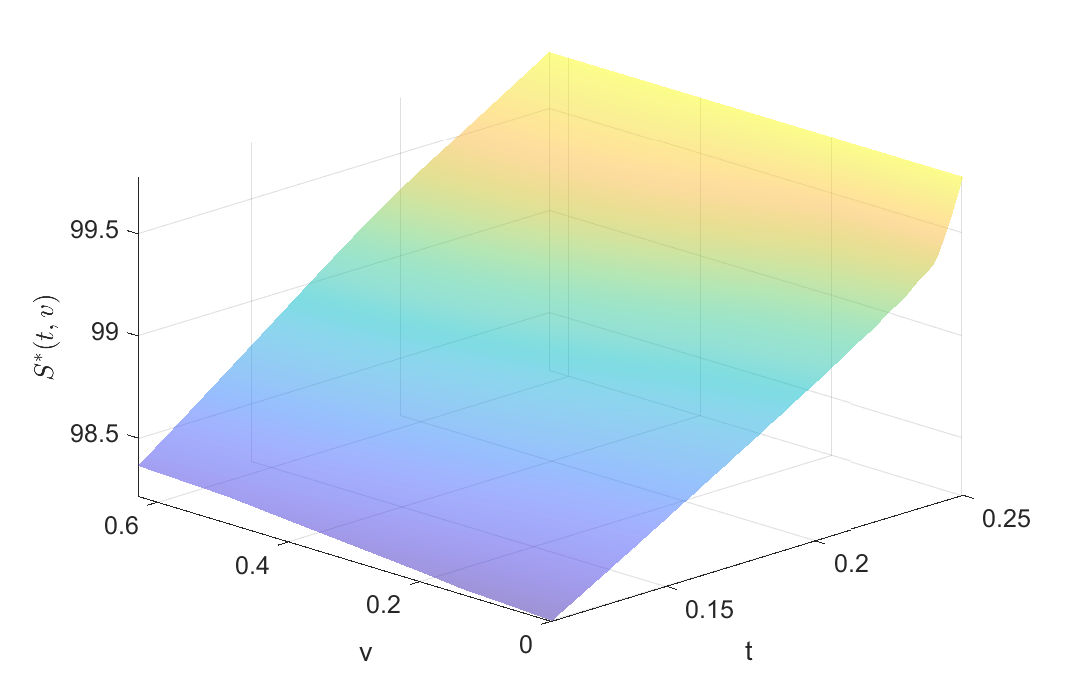}}
\subfloat[]{\includegraphics[width=0.5\textwidth]{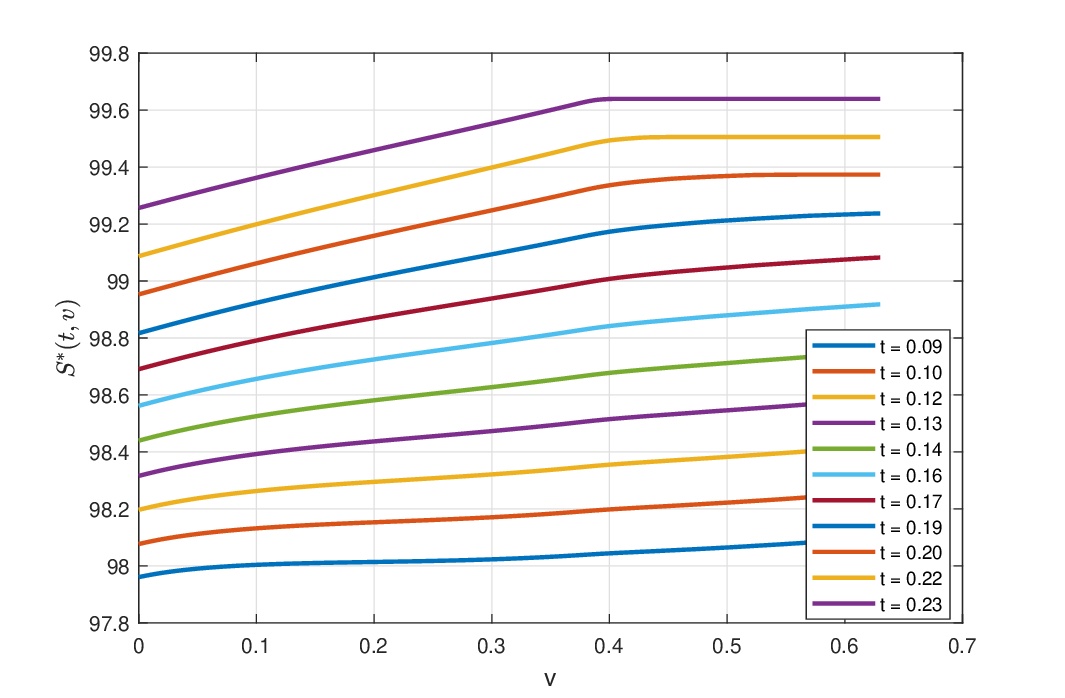}}
\end{center}
\caption{Early exercise boundaries for a short FF contract under the time-inhomogeneous Heston model with $S = 80, K = 80, v_0 = 0.5$ with model parameters from Table~\ref{FinputData}, but $v_0 = 0.05$; a) the entire ES, b) projections of the ES for various time points $t_i$.}
\label{eb2}
\end{figure}

Fig.~\ref{eb1}(a) presents thus computed EB for the short FF contract with payoff $K - S$ and constant parameters as in the first test for American options. Here, $t \in [T_1, T^-]$. Again, $S^*(t_i, v)$ for multiple time points $t_i$ are highly nonlinear as it can be seen in Fig.~\ref{eb1}(b). To notice, in this test \emph{fsolve} Matlab solver is more efficient than \emph{lsqnonlin} to find a reasonable solution.

The short FF price is -0.0966, which consists of the European price -0.0993 plus the EEP of 0.00269. The premium due to an American feature of the contract is 271 pips which is higher than typical (50-100 pips) premium values for maturities around 3 months.

\paragraph{Test 2.} Fig.~\ref{eb2}(a) and Fig.~\ref{eb2}(b) present same results for time-inhomogeneous Heston model with parameters from Table~\ref{FinputData}. It can be seen that here time dependence plays a critical role and significantly changes the shape of the EB. The short FF price is -0.1115, which consists of the European price -0.1134 plus the EEP of 0.00184. The premium due to an American feature of the contract is 162 pips.

\paragraph{Test 3.} In this test, we switch the rate curves: the domestic interest rate now follows the behavior of $r_f(t)$ in Fig.~\ref{FinputData}, while $r_d(t)$ takes the role of $r_f(t)$. The EB computed in this test with constant parameters is presented in Fig.~\ref{eb3}(a) and exhibits an interesting behavior, which arises because at maturity the ratio $r_d(T)/r_f(T) \gg 1$ and it is not limited by the strike value $K$ as for American options. Accordingly, the EB increases with backward time.

The short FF price is 0.11114, which consists of the European price 0.09931 plus
the EEP of 0.0121.
\begin{figure}[!htp]
\begin{center}
\hspace*{-0.3in}
\subfloat[]{\includegraphics[width=0.5\textwidth]{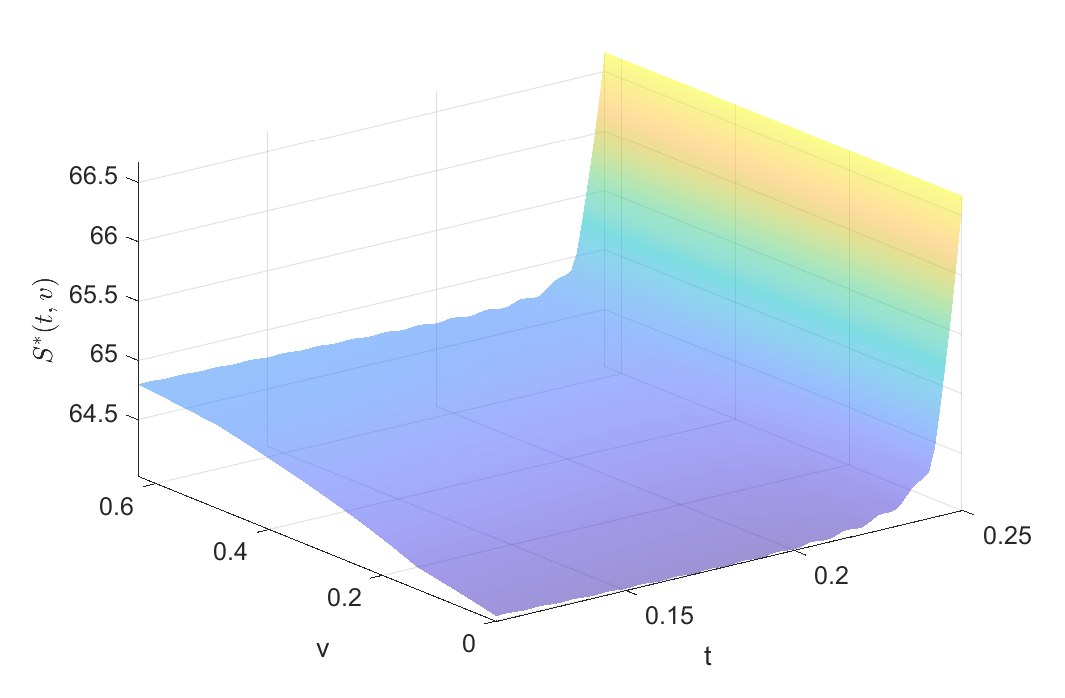}}
\subfloat[]{\includegraphics[width=0.5\textwidth]{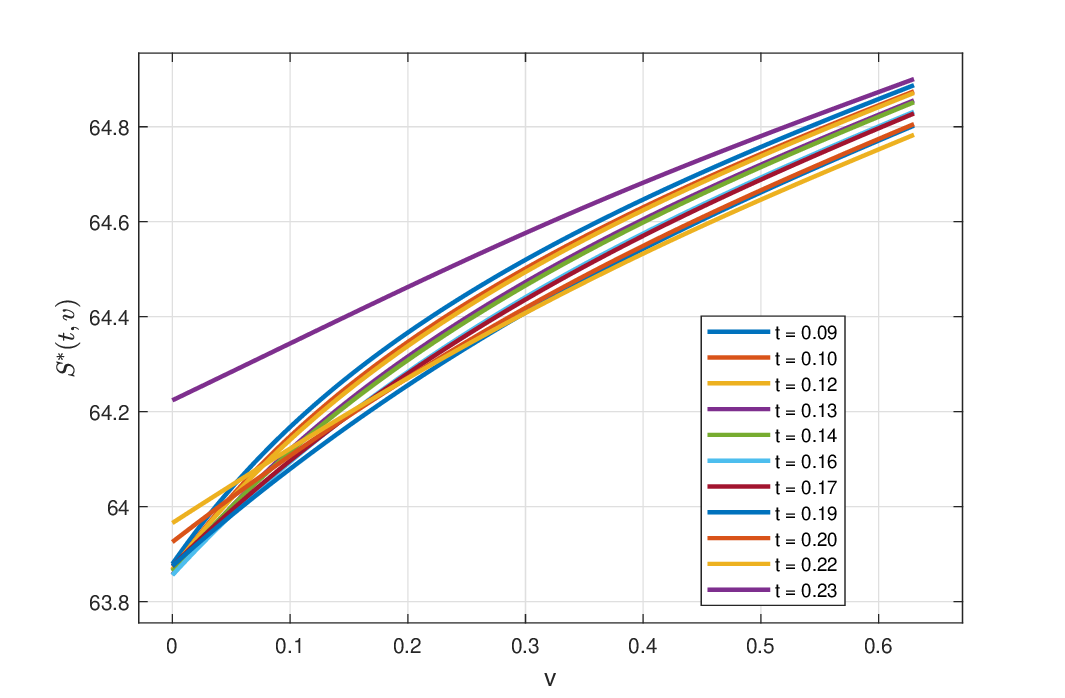}}
\end{center}
\caption{Early exercise boundaries for a short FF contract under the time-homogeneous Heston model with $S = 80, K = 80, v_0 = 0.05$ with other model parameters from Table~\ref{FinputData} but switched rate curves; a) the entire ES, b) projections of the ES for various time points $t_i$.}
\label{eb3}
\end{figure}

Similarly to the previous findings, the EB is not a linear function of the variance $v_t$ (despite for a certain time interval it can be approximated by a linear function), in contrast to the assumptions in \cite{tzavalis2003pricing, Chiarella2005} (who, however, did not consider FF contracts at all). Also, in contrast to previous tests, at a fixed time the EB can increases with $v_t$ until it flattens at high variances, and further is expected to go to zero based on the boundary conditions.

\paragraph{Test 4.} To further investigate this, we set $\kappa = 0.05, v_0 = 0.5$, use time-dependent parameters and bump $\xi$ by 1.2 over all time intervals, so $\Fe = 0.01$. The EB computed in this test with time-dependent parameters is presented in Fig.~\ref{eb4}(a) and Fig.~\ref{eb4}(b) presents $S^*(t_i, v)$ for multiple time points $t_i$ derived from this experiment. The results demonstrate higher nonlinearity of $S^*(t_i, v)$ especially pronounced at higher values of $v$.

The short FF price is -0.1134, which coincide with the European price, so EEP is 0.
\begin{figure}[!htp]
\begin{center}
\hspace*{-0.3in}
\subfloat[]{\includegraphics[width=0.5\textwidth]{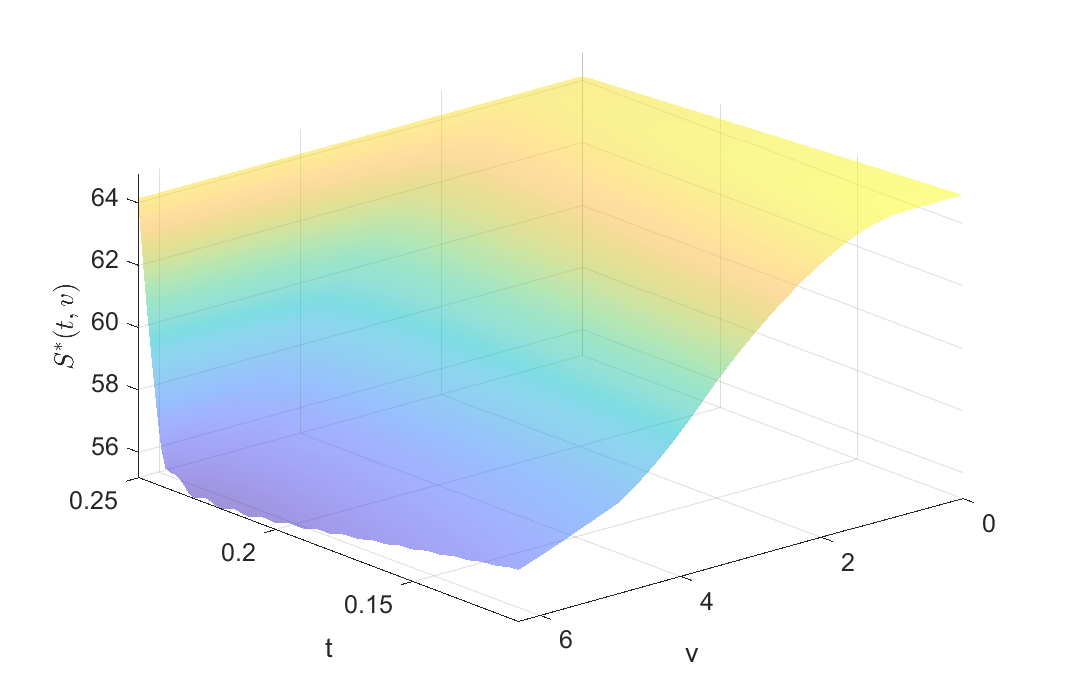}}
\subfloat[]{\includegraphics[width=0.5\textwidth]{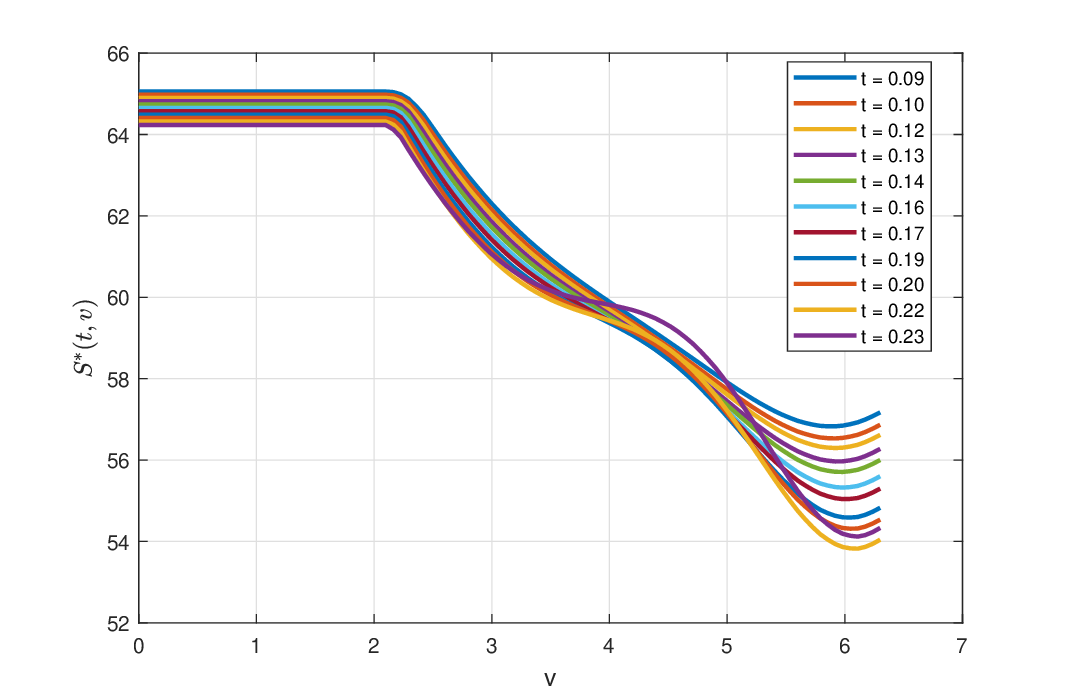}}
\end{center}
\caption{Early exercise boundaries for a short FF contract under the time-homogeneous Heston model with $S = 80, K = 80, v_0 = 0.5, \kappa = 0.05, \xi = 2.$ with other model parameters from Table~\ref{FinputData}; a) the entire ES, b) projections of the ES for various time points $t_i$.}
\label{eb4}
\end{figure}

\section{Conclusion} \label{sec:conclusion}

In this paper, we have developed a comprehensive, semi-analytical framework for pricing flexible forward contracts, which are essentially American-style options on delivery timing, within a time-inhomogeneous Heston model. By extending the integral equation approach for multidimensional diffusions to the case of time-dependent parameters, we have derived a system of Volterra integral equations that characterize both the early exercise surface and the contract value.

Our methodology combines several key innovations:

\begin{enumerate}
\item We propose a tractable time-inhomogeneous model, which is a Heston model where the long-term volatility level $\theta(t)$ is a continuous function of the time $t$, while the vol-of-vol $\xi(t)$ and correlation $\rho(t)$ are piecewise constant. This hybrid parameterization retains analytical tractability via recursive Riccati solutions for the joint CF, while capturing essential market features such as forward skew dynamics and term structure.

\item To efficiently compute the joint density, we employ two methods.

\begin{itemize}
\item The first method uses a double cosine expansion of the transition density, thereby avoiding costly numerical Fourier inversion. The expansion coefficients are expressed directly in terms of the CF, which can be computed efficiently via the recursive scheme derived in \cref{condCF}.

We also propose a fast algorithm for computing the early-exercise premium (EEP). The expectation appearing in the decomposition formula is evaluated analytically in the log-price variable and via discrete cosine transforms along the variance dimension. This reduces the computational complexity to $O(N_x N_v \log N_v)$, significantly outperforming naive quadrature.

\item The second method is a DSINC approach that addresses the limitations of the COS method when approximating $v$-densities for model parameters that induce high skewness or pronounced Gibbs oscillations. This method exploits the representation of the joint density using the CF conditional on a variance path (similar to \cite{romano1997contingent}) together with a tilted marginal density of the CIR process. Under this approach, the EEP can be expressed as a double integral—one in time and one over a spectral variable $\omega$. Importantly, the latter is a non-oscillatory integral that can be computed efficiently using Gauss--Legendre or Gauss--Kronrod methods.
\end{itemize}

\item We have tested the approach on both American Put options and short FF contracts, demonstrating that the early exercise surface is generally a nonlinear function of the variance, contrary to the linear-in-$v$ approximation sometimes used in earlier literature. The method remains robust even in near-Black–Scholes limits (small vol-of-vol and zero variance drift) and when the Feller condition is violated.

\end{enumerate}

The results show that the integral equation method offers a compelling alternative to finite-difference or least-squares Monte Carlo techniques for this class of problems. It provides high accuracy in determining the early exercise boundary, which is a critical input for risk management, while being inherently parallelizable and allowing for flexible handling of time-dependent coefficients.

For practical applications, the model can be calibrated to the vanilla European option surface via $\theta(t)$, $\xi(t)$, and $\rho(t)$, and then used to price FF contracts and other American-style derivatives consistently with the volatility smile. Future work could extend the framework to include stochastic interest rates for longer-dated contracts, or, for American options, models which consider hybrid dividends like in \cite{Itkin2025ddj}, while preserving the efficiency of the developed approach.

\section*{Disclosure statement}

No potential conflict of interest was reported by the authors.

\section*{Funding}

No funding was received.

\section*{Disclaimer}

Opinions expressed here are author's own, and do not represent views of their employers. A standard disclaimer applies.

\section*{Acknowledgments}

Andrey Itkin thanks Nizar Touzi for useful discussions and the FRE Tandon seminar participants for their valuable feedback.

\noindent The use of LLMs in this paper has been limited to proofreading and  verification of the literature and code.

\printbibliography[title={References}]

\appendix
\appendixpage
\numberwithin{equation}{section}
\setcounter{equation}{0}

\section{Computation of integrals in \eqref{Arec}} \label{app1}

We consider the integral
\begin{equation}
J(\alpha; t) = \int_0^t \frac{e^{\alpha s}}{1 - R e^{-\gamma s}} \, ds,
\end{equation}
where $\gamma$ and $\alpha$ are constants, and $t \ge 0$.

\paragraph{Case $\bm{\gamma \neq 0}$.}
Let $\mu = -\dfrac{\alpha}{\gamma}$.  Substitute $z = e^{-\gamma s}$; then
$dz = -\gamma z\, ds$ and $ds = -dz/(\gamma z)$.  The limits become $s=0 \Rightarrow z=1$ and $s=t \Rightarrow z=e^{-\gamma t}$.  Hence
\begin{equation}
J(\alpha; t) = -\frac{1}{\gamma} \int_{z=1}^{e^{-\gamma t}} \frac{z^{\mu}}{1 - R z} \, dz.
\end{equation}

For $|R z| < 1$, we expand $(1-Rz)^{-1} = \sum_{k=0}^\infty R^k z^k$.  Interchanging the sum and integral, we obtain
\begin{equation}
J(\alpha; t) = -\frac{1}{\gamma} \sum_{k=0}^\infty R^k \int_{1}^{e^{-\gamma t}} z^{\mu+k} \, dz.
\end{equation}
The inner integral is
\begin{equation}
\int_{1}^{e^{\gamma t}} z^{\mu+k} dz = \left. \frac{z^{\mu+k+1}}{\mu+k+1} \right|_{1}^{e^{-\gamma t}} = \frac{e^{-\gamma t(\mu+k+1)} - 1}{\mu+k+1}.
\end{equation}
Thus
\begin{equation}
J(\alpha; t) = -\frac{1}{\gamma} \sum_{k=0}^\infty \frac{R^k}{\mu+k+1}
\bigl[ e^{-\gamma t(\mu+k+1)} - 1 \bigr].
\end{equation}

Recall the hypergeometric representation, \cite{as64}
\begin{equation}
{}_2F_1(1, \mu+1; \mu+2; -w) = (\mu+1) \sum_{k=0}^\infty \frac{(-w)^k}{\mu+k+1}.
\end{equation}

Multiplying and dividing by $\mu+1$, we obtain
\begin{equation}
J(\alpha; t) = -\frac{1}{\gamma(\mu+1)} \Bigg[ e^{-\gamma t(\mu+1)} \sum_{k=0}^\infty \frac{(R e^{-\gamma t})^k}{\mu+k+1}(\mu+1)
- \sum_{k=0}^\infty \frac{R^k}{\mu+k+1}(\mu+1) \Bigg],
\end{equation}
which gives the closed form
\begin{equation}
J(\alpha; t) = -\frac{1}{\gamma(\mu+1)}\Bigl[
e^{-\gamma t(\mu+1)} \, {}_2F_1(1, \mu+1; \mu+2; R e^{-\gamma t})
- {}_2F_1(1, \mu+1; \mu+2; R)
\Bigr].
\end{equation}

\paragraph{Case $\bm{\gamma = 0}$.}
Then the integral becomes
\begin{equation}
J(\alpha; t) = \frac{1}{1-R} \int_0^t e^{\alpha s} \, ds = \frac{e^{\alpha t} -1}{\alpha(1-R)}.
\end{equation}

Combining all the results yields
\begin{equation} \label{combined}
J(\alpha; t) =
\begin{cases}
-\dfrac{1}{\gamma(\mu+1)}\Bigl[ e^{-\gamma t (\mu+1)} \, {}_2F_1(1, \mu+1; \mu+2; R e^{-\gamma t}) - {}_2F_1(1, \mu+1; \mu+2; R) \Bigr], & \gamma \neq 0, \\
\frac{e^{\alpha t} -1}{\alpha(1-R)}, & \gamma = 0.
\end{cases}
\end{equation}

\paragraph{Case $\bm{\alpha = 0}$.}

In this case $\mu = 0$. When $\mu = 0$, the second line in \eqref{combined} has a finite limit; for the first line, we apply an identity from \cite{as64}
\begin{equation}
{}_2F_1(1, 1; 2; x) = - \frac{\log(1-x)}{x}.
\end{equation}

\subsection{Transformation Identity}

We start from the identity between the Gaussian hypergeometric function ${}_2F_1$ and the incomplete Beta function. For $b \neq 0$ and $|z|<1$ (or by analytic continuation),
\begin{equation}
{}_2F_1(1,\,b;\,b+1;\,z) \;=\; b \int_0^1 \dfrac{t^{\,b-1}}{1 - z t} \, dt .
\end{equation}

Changing the integration variable $u = z t$, so that $t = u/z$, $dt = du/z$, and $t^{b-1} = u^{b-1}/z^{b-1}$, we obtain
\begin{equation}
{}_2F_1(1,b;b+1;z)
= b \int_0^z \dfrac{u^{b-1}/z^{\,b-1}}{1-u} \, \dfrac{du}{z}
= b\,z^{-b} \int_0^z \dfrac{u^{\,b-1}}{1-u} \, du .
\end{equation}

The integral on the right-hand side is exactly the definition of the incomplete Beta function with parameters $p = b$ and $q = 0$, i.e.
\begin{equation}
\mathrm{B}(z;\,p,\,0) \;:=\; \int_0^z \dfrac{u^{\,p-1}}{1-u} \, du,
\end{equation}
which is an \emph{analytic continuation} for $q=0$. Hence the closed-form identity reads
\begin{equation}\label{eq:identity}
{}_2F_1\bigl(1,\,b;\,b+1;\,z\bigr) = b\,z^{-b}\; \mathrm{B}\bigl(z;\,b,\,0\bigr).
\end{equation}

\subsection{Remarks on numerical evaluation}

The following recommendations facilitate an efficient numerical implementation of these recursions:
\begin{itemize}
\item In a general case, the parameter $b$ is complex and could have both negative and positive real part.

\item Standard implementations of the incomplete Beta function (e.g., \texttt{scipy.special.betainc}) often require $p>0$, $q>0$.

\item For $b<0$ one may get back to the identity
  \begin{equation}
  \mathrm{B}(z;\,b,\,0) = \dfrac{z^b}{b}\;{}_2F_1(b,1;\,b+1;\,z),
  \end{equation}
  since implementation of the hypergeometric function does not have this limitations on the parameters values.

\item If a stable implementation of the above functions is not available, a possible alternative numerical approach is to evaluate the original hypergeometric function directly via its integral representation
  \begin{equation}
  {}_2F_1(1,b;b+1;z) = b \int_0^1 \dfrac{t^{b-1}}{1-zt}\,dt,
  \end{equation}
  using adaptive quadrature, or to employ the transformation
  \begin{equation}
  {}_2F_1(1,b;b+1;z) = (1-z)^{-1}\;{}_2F_1\Bigl(1,1;b+1;\dfrac{z}{z-1}\Bigr),
  \end{equation}
  which moves the singular point away from the unit interval.
\end{itemize}

Using \cref{solB,Arec}, the recursive computation of the Heston CF for time-dependent $\theta(t)$ and piecewise constant $\xi$ and $\rho$ reduces to evaluating only elementary functions, logarithms, and the incomplete Beta function (or its hypergeometric equivalent). Moreover, the computational cost of these special functions remains low, since the time intervals for constant $\xi$ and $\rho$ which, as mentioned in Introduction, can be naturally aligned with standard traded option tenors like $[1\text{M}, 3\text{M}, 6\text{M}, 1\text{Y}, 2\text{Y}]$, correspond to a small, fixed number of evaluations per pricing exercise.

\section{Computation of $\calJ(u| x,v,t)$ in \eqref{expect}} \label{app2}

To recall, the expectation in \eqref{decompGenN} is given in the explicit form by \eqref{expect} and reads
\begin{align} \label{calJ}
\calJ(u| x,v,t) &\equiv \EQ \Big\{ \left[ r_d(u) - r_f(u) e^{x_u} \right] \mathbf{1}_{(x_u, v_u) \in \mathcal{E}} \Big\} \\
&= \int_0^\infty dv_u \int_{-\infty}^{x^*(u,v_u)} \left[ r_d(u) - r_f(u) e^{x_u} \right]  f(x_u, v_u, u \mid x, v, t) dx_u. \nonumber
\end{align}
Substituting $f(x_u, v_u, u \mid x, v, t)$ from \eqref{tdFin} into \eqref{calJ} yields
\begin{align} \label{Jcos}
\calJ(u| x,v,t) &\approx \dfrac{4 K}{L_x L_v} \sum_{k=0}^{N_x-1}{}^{'} \; \sum_{m=0}^{N_v-1}{}^{'} \Ree\left[ e^{-\iu\left( \dfrac{k\pi a_x}{L_x} + \dfrac{m\pi a_v}{L_v} \right)} \phi\left( \dfrac{k\pi}{L_x},\; \dfrac{m\pi}{L_v} \mid x,v,t \right) \right] \\
&\cdot \int_0^\infty \cos\left( \dfrac{m\pi (v_u-a_v)}{L_v} \right) \left[ \int_{a_x}^{x^*(u,v_u)} \left[ r_d(u) - r_f(u) e^{x_u} \right]  \cos\left( \dfrac{k\pi}{L_x} \left(x_u-a_x\right) \right) dx_u \right] d v_u. \nonumber
\end{align}

Further, let us define two new functions
\begin{align} \label{newFunc}
g_k(v_u) &= \int_{a_x}^{x^*(u,v_u)} \left[ r_d(u) - r_f(u) e^{x_u} \right] \cos\left( \dfrac{k\pi}{L_x} \left(x_u-a_x\right) \right) dx_u, \\
J(k,m) &= \int_0^\infty \cos\left( \dfrac{m\pi (v_u-a_v)}{L_v} \right) g_k(v_u) dv_u, \nonumber
\end{align}
so, \eqref{Jcos} can be re-written in a compact form
\begin{equation}
\calJ(u| x,v,t) \approx \dfrac{4}{L_x L_v} \sum_{k=0}^{N_x-1}{}^{'} \; \sum_{m=0}^{N_v-1}{}^{'} \Ree\left[ e^{-\iu\left( \dfrac{k\pi a_x}{L_x} + \dfrac{m\pi a_v}{L_v} \right)} \phi\left( \dfrac{k\pi}{L_x},\; \dfrac{m\pi}{L_v} \mid x,v,t \right) \right] J(k,m).
\end{equation}
Notable, the function $g_k(v)$ can be obtained in closed form.

\subsection{Analytic computation of $g_k(v)$} \label{gkvB}

Let $\theta_k(x) = \dfrac{k\pi}{L_x} (x - a_x)$. We need to compute
\begin{equation} \label{gkv}
g_k(v) = r_d(u) \int_{a_x}^{x^*(u,v)} \cos\theta_k(x) dx - r_f(u) \int_{a_x}^{x^*(u,v)} e^{x} \cos\theta_k(x) dx.
\end{equation}
Computing the first integral in \eqref{gkv}, yields
\begin{align}
\int_{a_x}^{x^*} \cos\theta_k(x) dx &=
\begin{cases}
- \dfrac{L_x}{k \pi} \sin \theta_k(x^*), & k > 0, \\
x^*-a_x, & k = 0.
\end{cases}
\end{align}

For the second integral, we have
\begin{align}
\int_{a_x}^{x^*} e^{x} \cos\theta_k(x) dx &=
\begin{cases}
L_1  \left[e^{x^*} \left( L_x \cos \theta_k(x^*) - k \pi \sin \theta_k(x^*) \right) - L_x e^{a_x} \right], & k > 0, \\
e^{x^*} - e^{a_x}, & k = 0.
\end{cases}
\end{align}

Combining all the above results, we obtain the final expression for $g_k(v)$
\begin{align} \label{gkv}
g_k(v) &=
\begin{cases}
 L_1\left[ \alpha_c \cos \theta_k(x^*) + \alpha_s \sin \theta_k(x^*) + \alpha_0 \right], & k > 0, \\
r_d(u)(x^*-a_x) - r_f(u) \left(e^{x^*} - e^{a_x} \right), & k = 0,
\end{cases}
\\
\alpha_c &= - r_f L_x e^{x^*}, \quad
\alpha_s = k \pi \left[ r_d \left(1 + L^2_2\right) - r_f e^{x^*} \right], \quad
\alpha_0 = r_f L_x e^{a_x}, \quad
L_1 = \dfrac{L_x}{L_x^2 + k^2 \pi^2}, \quad L_2 = \frac{L_x}{k \pi}, \nonumber
\end{align}
with $x^* = x^*(u,v)$.

\subsection{Efficient computation of $J(k,m)$}

By definition in \eqref{newFunc}, $J(k,m)$  is an integral over a semi-infinite interval in $v$. In line with the COS method, we truncate this integration domain to $[a_v, a_v+L_v]$, yielding the approximation
\begin{equation} \label{Jkm}
J(k,m) \approx \int_{a_v}^{a_v+L_v} \cos\left( \dfrac{m\pi (v-a_v)}{L_v} \right) g_k(v) \, dv.
\end{equation}
The right-hand side is precisely the cosine-series coefficient of $g_k(v)$  on $[a_v, a_v+L_v]$. Consequently, for all $m = 0, \dots, N_v-1$ , the values of $J(k,m)$  can be obtained efficiently by applying a type-I discrete cosine transform (DCT) to a uniform sampling of $g_k(v)$  over this interval. The complete procedure is summarized in \cref{alg:double_cos} in \cref{sec:algo}.

The computational complexity of the algorithm is dominated by three steps:
\begin{itemize}
    \item Analytic evaluation of $g_k(v)$ for all $k,v$ which yields $O(N_x N_v)$.
    \item DCT for each $k$: $O(N_x N_v \log N_v)$.
    \item Final double summation: $O(N_x N_v)$.
\end{itemize}
Hence, the overall complexity is $O(N_x N_v \log N_v)$.

From an implementation standpoint, several practical aspects are worth noting:
\begin{enumerate}
\item The DCT-I can be computed using standard numerical libraries (e.g., FFTW, SciPy).
\item To improve accuracy, the sampling grid for $v$  can be chosen larger than the number of cosine terms $N_v$.
\item For large $N_x$, the loop over $k$ can be parallelized trivially.
\end{enumerate}

\section{Conditional PDF and CF in the Heston Model} \label{condCF1}

Recall that the Heston model has been presented in \eqref{model}. We decompose $W_t^S$ into a component correlated with $W_t^v$ and an independent component $W_t^\perp$ via
\begin{equation}
W_t^S = \rho(t)\,W_t^v + \sqrt{1-\rho^2(t)}\,W_t^\perp, \qquad W_t^\perp \perp W_t^v,
\end{equation}
and condition on the full path, $x_t$ given $(\Lambda_t, \tilde{\Lambda}_t, v_t)$ where
\begin{equation} \label{Lambda}
\Lam_t = \int_0^t \lambda_1(s) v_s ds, \qquad \tilde{\Lambda}_t = \int_0^t \lambda_2(s) v_s ds,
\end{equation}
and $\lambda_1(t), \lambda_2(t)$ are some deterministic functions to be found. Then, integrating the first SDE in \eqref{model} from $0$ to $t$ yields
\begin{equation}   \label{eq:xt_decomp}
x_t = x_0 + \int_0^t \left[r_d(k) - r_f(k) - \tfrac{1}{2}v_k \right] dk +  \int_0^t \rho(s) \sqrt{v_s} dW_s^v + \int_0^t \sqrt{1-\rho^2(s)} \sqrt{v_s} dW_s^\perp.
\end{equation}

A key structural observation is that, \emph{conditional on the filtration generated by $v$}, the log-price $x_t$ is Gaussian, \cite{romano1997contingent,BroadieKaya2006}. More precisely, let us fix the realisation of the full variance path $(v_s)_{0\leq s\leq t}$, which determines both $\Lam_t$ and the terminal value $v_t$.  Multiplying the second SDE in \eqref{model} by $\rho(t)/\xi(t)$ and integrating in time from 0 to $t$, yields
\begin{equation}   \label{eq:ito_trick}
\int_0^t \rho(s) \sqrt{v_s} dW_s^v = v_t \frac{\rho(t)}{\xi(t)} - v_0 \frac{\rho(0)}{\xi(0)} - \kappa \int_0^t \frac{\rho(s) \theta(s)}{\xi(s)} ds - \int_0^t v_s \left[\left(\frac{\rho(s)}{\xi(s)}\right)'_s - \kappa \frac{\rho(s)}{\xi(s)} \right] ds.
\end{equation}

Substituting \eqref{eq:ito_trick} into \eqref{eq:xt_decomp} and defining
\begin{equation}
\lambda_1(s) = -\left(\frac{\rho(s)}{\xi(s)}\right)'_s + \kappa \frac{\rho(s)}{\xi(s)} - \frac{1}{2},
\end{equation}
yields
\begin{align} \label{eq:mu}
x_t &= \mu(\Lam_t, v_t) + \int_0^t \sqrt{1-\rho^2(s)} \sqrt{v_s} dW_s^\perp, \\
\mu(\Lam_t, v_t) &= \mu_0(t) + \mu_1(t) v_t + \Lam_t, \nonumber \\
\mu_0(t) &= x_0 + \int_0^t \left[r_d(k) - r_f(k) \right] dk - v_0 \frac{\rho(0)}{\xi(0)} - \kappa \int_0^t \frac{\rho(s) \theta(s)}{\xi(s)} ds, \qquad
\mu_1(t) = \frac{\rho(t)}{\xi(t)}. \nonumber
\end{align}
Conditional on $(\Lambda_t, v_t)$, or equivalently on $\mathcal{F}^v$, the stochastic integral $\int_0^t \sqrt{1-\rho^2(s)} \sqrt{v_s} dW_s^\perp$ is a centered Gaussian random variable with variance $\tilde{\Lambda}_t$ (by the conditional \Ito isometry)
\begin{equation}
\mathbb{E}\!\left[\left(\int_0^t \sqrt{1-\rho^2(s)} \sqrt{v_s}\, dW_s^\perp\right)^2 \middle|\, \mathcal{F}^v\right] = \int_0^t (1-\rho^2(s)) v_s\, ds = \tilde{\Lambda}_t,
\end{equation}
where the last equality follows by setting $\lambda_2(t) = 1-\rho^2(t)$.

We, therefore, obtain the exact conditional distribution
\begin{equation} \label{eq:cond_normal}
x_t \mid \Lambda_t, \tilde{\Lambda}_t, v_t \sim \mathcal{N}\left( \mu_0(t) + \frac{\rho(t)}{\xi(t)}v_t + \Lambda_t, \tilde{\Lambda}_t \right).
\end{equation}
The Gaussian structure in \eqref{eq:cond_normal} is exact (not an approximation) and holds pathwise for every realisation of the variance process.

\subsection{Conditional CF given $(\Lambda_t, \tilde{\Lambda}_t, v_t)$}

The CF of $x_t$ given the state $(\Lambda_t, \tilde{\Lambda}_t, v_t)$ follows immediately from the exact Gaussian distribution in \eqref{eq:cond_normal}. For a transform parameter $\omega \in \mathbb{R}$, we have
\begin{equation} \label{eq:CF_full_cond}
\phi(t, \omega \mid \Lambda_t, \tilde{\Lambda}_t, v_t) = \mathbb{E}\left[e^{\iu \omega x_t} \mid \Lambda_t, \tilde{\Lambda}_t, v_t \right] = \exp\left(\iu \omega \mu(\Lambda_t, v_t) - \frac{1}{2} \omega^2 \tilde{\Lambda}_t \right).
\end{equation}

In our setting, we solve the integral equation for the EB backward in time on a grid in $\tau = T-t$. At each time step $t$, we need both the European Put price and the EEP written on $x_t$ while treating the terminal variance $v_t$ as an observable (or integrating it out in a mixture representation). We, therefore, consider the CF conditioned on the endpoint $v_t = v$, marginalizing over the path functionals $\Lambda_t$ and $\tilde{\Lambda}_t$:
\begin{align} \label{eq:CF_vt_cond}
\phi(t, \omega \mid x_0, v_0, v_t = v) &= \mathbb{E} \left[ e^{\iu \omega x_t} \mid v_t = v, v_0 \right] \nonumber \\
&= \mathbb{E} \left[ \exp\left( \iu \omega \mu(\Lambda_t, v_t) - \frac{1}{2} \omega^2 \tilde{\Lambda}_t \right) \mid v_t = v, v_0 \right].
\end{align}
Substituting the decomposition of $\mu(\Lambda_t, v_t)$ from \eqref{eq:mu} into \eqref{eq:CF_vt_cond} and separating the deterministic components from the stochastic path-dependent terms yields the following factorization:
\begin{align} \label{eq:CF_factored}
\Phi(t, \omega \mid x_0, v_0, v_t=v) &= e^{i \omega \left[\mu_0(t) + \mu_1(t) v \right]} \mathbb{E} \left[ e^{i \omega \Lambda_t - \frac{1}{2} \omega^2 \tilde{\Lambda}_t} \mid v_t=v, v_0 \right] \nonumber \\
&= e^{\iu \omega \left[\mu_0(t) + \mu_1(t) v \right]} \mathbb{E} \left[ e^{-\int_0^t \gamma(s, \omega) v_s ds} \mid v_t=v, v_0 \right],
\end{align}
where the time-dependent complex frequency kernel is defined as
\begin{equation} \label{eq:gamma_omega}
\gamma(s, \omega) = \frac{1}{2}\omega^2 \lambda_2(s) - \iu \omega \lambda(s).
\end{equation}
The factorization \eqref{eq:CF_factored} reduces the problem to computing the conditional Laplace transform of a generalized integrated variance functional, evaluated at the time-varying kernel $\gamma(s, \omega)$.

By Bayes' theorem, the conditional Laplace transform appearing in \eqref{eq:CF_factored} can be expressed as \cite{Revuz_Yor1999}
\begin{equation} \label{eq:bayes}
\mathbb{E} \left[ e^{-\int_0^t \gamma(s, \omega) v_s ds} \mid v_t=v, v_0 \right] = \frac{\mathbb{E} \left[ e^{-\int_0^t \gamma(s, \omega) v_s ds} \mathbf{1}_{v_t \in dv} \mid v_0 \right]}{p_{\mathrm{CIR}}(v_t=v \mid v_0)}.
\end{equation}
The numerator in \eqref{eq:bayes} is the joint Laplace-transform density of the variance process and its path-dependent functionals under the CIR dynamics. This object remains tractable due to the affine structure of the CIR process, \cite{DuffiePanSingleton:2000}. However, unlike the constant-parameter case, the time-dependent kernel $\gamma(s, \omega)$ implies that the transform does not generally result in a standard CIR density with a constant shifted parameter. Instead, it is given by
\begin{equation} \label{eq:joint_density}
p_{\mathrm{CIR}}^{(\gamma(s, \omega))}(t, v \mid v_0) \coloneqq \mathbb{E} \left[ e^{-\int_0^t \gamma(s, \omega) v_s ds} \mathbf{1}_{v_t \in dv} \mid v_0 \right] = \exp\left( A(t, \omega) + B(t, \omega) v_0 \right),
\end{equation}
where the functions $A(s, \omega)$ and $B(s, \omega)$ are obtained by integrating the associated Riccati system backward in time $\tau$ from $\tau=0$ to $\tau=\tau(t)$. For the conditional density, this leads to
\begin{equation} \label{laplace}
\mathbb{E} \left[ e^{-\int_0^t \gamma(s, \omega) v_s ds} \mid v_t=v, v_0 \right] = \frac{p_{\mathrm{CIR}}^{(\gamma(s, \omega))}(t, v \mid v_0)}{p_{\mathrm{CIR}}(t, v \mid v_0)}.
\end{equation}
Here, the denominator $p_{\mathrm{CIR}}$ is the standard CIR transition density, and the numerator $p_{\mathrm{CIR}}^{(\gamma(s, \omega))}$ is the joint density derived from the affine propagator under time-varying coefficients.

\begin{myremark}[Solution of the Riccati system]
In the presence of time-varying $\rho(t)$ and $\xi(t)$, the "tilted" mean-reversion $\tilde{\kappa}$ is no longer a constant but a function of time. If the model parameters are assumed to be piecewise constant, the numerator in \eqref{laplace} can be computed by chaining the analytical solutions of the Riccati equations across each sub-interval, maintaining the terminal variance $v_t$ as a parameter in the inversion.
\end{myremark}

Finally, by assembling \eqref{eq:CF_factored} and \eqref{laplace}, we obtain the final form of the conditional CF:
\begin{equation} \label{finalPhicond}
\phi(t, \omega \mid x_0, v_0, v_t=v) = e^{\iu \omega [\mu_0(t) + \mu_1(t) v]} \frac{p_{\mathrm{CIR}}^{(\gamma(s, \omega))}(t, v \mid v_0)}{p_{\mathrm{CIR}}(t, v \mid v_0)}.
\end{equation}

\begin{myremark}[Use in option pricing via the COS method]
The formula \eqref{finalPhicond} is the correct object to supply to spectral pricing methods such as the COS method or SWIFT when conditioning on the terminal variance $v_t$. The marginal Heston CF $\mathbb{E}[e^{\iu \omega x_t} \mid x_0, v_0]$ should not be used in this context, as it discards the information contained in $v_t$ and conflates two distinct probability kernels. The conditional CF \eqref{finalPhicond} preserves the full Markovian structure of the $(x_t, v_t)$ pair, allowing for accurate backward induction in the early exercise boundary problem and yielding substantially more robust results in numerical schemes that exploit terminal-variance information (e.g., control-variate or stratification schemes, \cite{BroadieKaya2006}).
\end{myremark}

\section{Restoring coefficients $c_j$ in the DSINC expansion} \label{cjDeriv}

We assume the density $p(x_u, v_u | x,v)$ is approximated by a sum of DSINC functions
\begin{equation} \label{e1}
p(x_u, v_u | x,v) \approx \sum_{k=-\infty}^{\infty} c_k e^{-\chi(x-x_k)}  \sinc\left(\frac{x-x_k}{\Delta x}\right).
\end{equation}
To simplify, let’s define a "tilted" density $q(x_u) = p(x_u, v_u | x,v) e^{\chi x_u}$. Multiplying both sides of \eqref{e1} by $e^{\chi x_u}$ yields
\begin{equation}
q(x_u) \approx \sum_{k=-\infty}^{\infty} (c_k e^{\chi x_k}) \sinc\left(\frac{x-x_k}{\Delta x}\right).
\end{equation}

Now, take the Fourier transform ($\mathcal{F}$) of both sides. By the Whittaker-Shannon Interpolation Formula, a function that is represented as a sum of Sinc functions is strictly band-limited in the frequency domain. The Fourier transform of $\sinc(\frac{x-x_k}{\Delta x})$ is a shifted rectangular (Box) function
\begin{equation}
\mathcal{F}\left[\sinc\left(\frac{x-x_k}{\Delta x}\right)\right] = \Delta x e^{\iu \omega x_k} \bm{1}_{\left[-\frac{\pi}{\Delta x}, \frac{\pi}{\Delta x}\right]}(\omega).
\end{equation}

Applying this to the entire sum, we obtain
\begin{equation}
\hat{q}(\omega) = \sum_{k=-\infty}^{\infty} (c_k e^{\chi x_k}) \Delta x e^{\iu \omega x_k} \bm{1}_{\left[-\frac{\pi}{\Delta x}, \frac{\pi}{\Delta x}\right]}(\omega).
\end{equation}

Note that $\hat{q}(\omega)$ as the Fourier transform is exactly the CF $\Psi$ evaluated with a complex shift
\begin{equation}
\hat{q}(\omega) = \int p(x_u, v_u | x,v)) e^{\chi x_u} e^{\iu \omega x_u} dx_u = \int p(x_u, v_u | x,v)) e^{\iu(\omega - \iu\chi) x_u } dx_u = \Psi(\omega - \iu\chi).
\end{equation}
Inside the band $\omega \in [-\frac{\pi}{\Delta x}, \frac{\pi}{\Delta x}]$, we have
\begin{equation} \label{e2}
\Psi(\omega - \iu\chi) = \Delta x \sum_{k=-\infty}^{\infty} (c_k e^{\chi x_k}) e^{\iu \omega x_k}.
\end{equation}
This is a Fourier series representation of the function $\Psi(\omega - \iu\chi)$ on the interval $[-\frac{\pi}{\Delta x}, \frac{\pi}{\Delta x}]$.

Even though the Sinc functions are not orthogonal in $x$-space, the complex exponentials $e^{\iu \omega x_k}$ are orthogonal in $\omega$-space over the interval $[-\frac{\pi}{\Delta x}, \frac{\pi}{\Delta x}]$. Therefore, to isolate the $j$-th coefficient, we multiply \eqref{e2} by $e^{-\iu \omega x_j}$ and integrate over the band
\begin{equation}
\int_{-\frac{\pi}{\Delta x}}^{\frac{\pi}{\Delta x}} \Psi(\omega- \iu\chi) e^{-\iu \omega x_j} d\omega = \Delta x \sum_{k=-\infty}^{\infty} (c_k e^{\chi x_k}) \int_{-\frac{\pi}{\Delta x}}^{\frac{\pi}{\Delta x}} e^{\iu \omega (x_k - x_j)} d\omega.
\end{equation}
The integral on the right is zero unless $k=j$, in which case it equals the length of the interval, $\frac{2\pi}{\Delta x}$. Therefore,
\begin{equation}
\int_{-\frac{\pi}{\Delta x}}^{\frac{\pi}{\Delta x}} \Psi(\omega - \iu\chi) e^{-\iu \omega x_j} d \omega = \Delta x (c_j e^{\chi x_j}) \frac{2\pi}{\Delta x} = 2\pi c_j e^{\chi x_j}.
\end{equation}
Rearranging for $c_j$, yields
\begin{equation} \label{cjFin}
c_j = \frac{e^{-\chi x_j}}{2\pi} \int_{-\frac{\pi}{\Delta x}}^{\frac{\pi}{\Delta x}} \Psi(\omega - \iu\chi) e^{-\iu \omega x_j} d \omega.
\end{equation}

\section{Solving Riccati equations in \eqref{ricNew} via matrix propagators} \label{appMoebius}

As established in \eqref{PsiAff}, the CF $\Psi_{\mathrm{CIR}}^{(\gamma(u,\omega))}(q)$ admits an exponential affine representation
\begin{equation}
\Psi_{\mathrm{CIR}}^{(\gamma(u,\omega))}(q) = e^{\mathcal{A}(u, q) + \mathcal{B}(u, q) v},
\end{equation}
where the coefficients $\mathcal{A}(u,q)$ and $\mathcal{B}(u,q)$ solve the system of Riccati equations \eqref{ricNew}. For the generalized time-dependent Heston model, the scalar coefficient $\mathcal{B}(\tau)$ satisfies the following nonlinear quadratic initial value problem:
\begin{equation} \label{eq:riccati_scalar}
\frac{d\mathcal{B}(\tau)}{d\tau} = Q(\tau) + P(\tau)\mathcal{B}(\tau) + R(\tau)\mathcal{B}^2(\tau), \quad \mathcal{B}(0) = \Lambda_i(u, \omega),
\end{equation}
where the time-dependent coefficients $Q, P, R$ are determined by the model parameters and the generalized frequency kernel $\gamma(u, \omega)$:
\begin{align}
Q(\tau) &= -\gamma(\tau, \omega) = -\left[ \frac{1}{2}\omega^2 \lambda_2(\tau) - i \omega \lambda_1(\tau) \right], \\
P(\tau) &= -\kappa(\tau), \qquad R(\tau) = \frac{1}{2}\xi^2(\tau). \nonumber
\end{align}

Based on \eqref{Jfin}, we require the computation of $\Psi_{\mathrm{CIR}}^{(\gamma(u,\omega))}(q)$ for $q=-i\Lambda_1(u,\omega)$ and $q=-i\Lambda_2(u,\omega)$. Since both frequencies vary with the "time-to-maturity" $u$, a standard marching scheme for $\mathcal{B}$ would necessitate a full re-integration for every horizon $u_k \in \{u_1, \dots, u_N\}$. This would result in a computational complexity of $O(N^2)$, creating a significant bottleneck for the DSINC method.

To decouple the evolution of the state from the terminal-dependent initial conditions, we apply the linearization technique known as Radon's Lemma \cite{reid1972riccati}. We introduce the transformation $\mathcal{B}(\tau)=Y(\tau)/X(\tau)$, mapping the nonlinear scalar equation into a $2\times2$ linear system:
\begin{equation} \label{eq:hamiltonian}
\frac{d}{d\tau} \begin{pmatrix} X(\tau) \\ Y(\tau) \end{pmatrix} = \mathbb{M}(\tau) \begin{pmatrix} X(\tau) \\ Y(\tau) \end{pmatrix}, \quad \text{with } \mathbb{M}(\tau) = \begin{pmatrix} 0 & -R(\tau) \\ Q(\tau) & P(\tau) \end{pmatrix}.
\end{equation}
The solution at time $\tau=u$ is governed by the fundamental solution matrix (or propagator) $\mathbf{W}(u)$, such that:
\begin{equation}
\begin{pmatrix} X(u) \\ Y(u) \end{pmatrix} = \mathbf{W}(u) \begin{pmatrix} X(0) \\ Y(0) \end{pmatrix} = \begin{pmatrix} w_{11}(u) & w_{12}(u) \\ w_{21}(u) & w_{22}(u) \end{pmatrix} \begin{pmatrix} 1 \\ \Lambda_i(u, \omega) \end{pmatrix}.
\end{equation}
The matrix $\mathbf{W}(u)$ is the path-ordered exponential of $\mathbb{M}(\tau)$. Crucially, $\mathbf{W}(u)$ depends solely on the running time parameters and the frequency $\omega$, remaining independent of the ``moving'' initial condition $\Lambda_i(u, \omega)$.

For a time-discretized grid where coefficients are assumed piecewise constant over intervals $\Delta\tau$, the local propagator $\mathbf{A}=e^{\mathbb{M}\Delta\tau}$ is computed analytically using the eigenvalues of $\mathbb{M}$:
\begin{equation}
\lambda_{\pm} = \frac{P \pm \sqrt{P^2 - 4RQ}}{2} = \frac{P \pm d}{2}, \quad d = \sqrt{P^2 - 4RQ}.
\end{equation}
The components of the local propagator matrix $\mathbf{A}$ are given by:
\begin{align}
w_{11} &= e^{\frac{P}{2}\Delta\tau} \left[ \cosh\left(\frac{d\Delta\tau}{2}\right) - \frac{P}{d}\sinh\left(\frac{d\Delta\tau}{2}\right) \right], \\
w_{12} &= -\frac{2R}{d} e^{\frac{P}{2}\Delta\tau} \sinh\left(\frac{d\Delta\tau}{2}\right), \nonumber \\
w_{21} &= \frac{2Q}{d} e^{\frac{P}{2}\Delta\tau} \sinh\left(\frac{d\Delta\tau}{2}\right), \nonumber \\
w_{22} &= e^{\frac{P}{2}\Delta\tau} \left[ \cosh\left(\frac{d\Delta\tau}{2}\right) + \frac{P}{d}\sinh\left(\frac{d\Delta\tau}{2}\right) \right].\nonumber
\end{align}

By substituting the components of the cumulative propagator $\mathbf{W}(u_k)=\prod_{j=1}^{k} \mathbf{A}_j$ into the ratio $Y/X$, the solution for $\mathcal{B}(u_k)$ is recovered via a M\"{o}bius transformation:
\begin{equation} \label{eq:moebius_final}
\mathcal{B}(u_k) = \frac{w_{21}(u_k) + w_{22}(u_k) \Lambda_i(u_k, \omega)}{w_{11}(u_k) + w_{12}(u_k) \Lambda_i(u_k, \omega)}.
\end{equation}
This formulation yields a significant algorithmic advantage: the propagator $\mathbf{W}$ is advanced in time once in $O(N)$ steps, and for each $u_k$, the boundary-specific frequency $\Lambda_i(u_k, \omega)$ is directly inserted into \eqref{eq:moebius_final}. Consequently, the complexity of the Riccati solver is reduced from $O(N^2)$ to $O(N)$, ensuring numerical efficiency for the time-inhomogeneous case.

\clearpage
\section{Various algorithms} \label{sec:algo}

\begin{algorithm}[!htb]
\caption{Recursive propagation for $B(\tau)$}
\label{alg:backward_recursion}
\KwIn{Complex $u_1, u_2$, real $\kappa$, grid $0 = \tau_0 < \dots < \tau_N=t$, constants $\xi_n, \rho_n$ for $n=1,\dots,N$}
\KwOut{Piecewise function $B(\tau)$ for $\tau \in [0, t]$}
\SetAlgoLined
\DontPrintSemicolon

$B_{\text{prev}} \gets i u_2$\; \tcp*{Initial condition}

\For{$n = 1$ \KwTo $N$}{
    \tcp{Coefficients for interval $n$}
    $a_n \gets \dfrac12 \xi_n^2$\;
    $b_n \gets i u_1 \rho_n \xi_n - \kappa$\;
    $c_n \gets -\dfrac12 (u_1^2 + i u_1)$\;

    \tcp{Roots}
    $\Delta_n \gets b_n^2 - 4 a_n c_n$,\, $d_n \gets \sqrt{\Delta_n}$\;
    $\lambda_1 \gets \dfrac{-b_n + d_n}{\xi_n^2},\; \lambda_2 \gets \dfrac{-b_n - d_n}{\xi_n^2},\; \gamma = \lambda_1 - \lambda_2$\;

    \tcp{Matching at $\tau_{n-1}$}
    \eIf{$|\lambda_1 - \lambda_2| < \epsilon$}{
        \tcp{Repeated roots}
        $U_n \gets B(\tau_{n-1}) - \lambda_{n}$\;
        \For{$\tau \in [\tau_{n-1}, \tau_n)$}{
            $s \gets \tau - \tau_{n-1}$\;
            $B(\tau) \gets \lambda_n + \dfrac{U_n}{1 - a_n U_n s}$\;
        }
    }{
        \tcp{Distinct roots}
        $R_n \gets \dfrac{B(\tau_{n-1}) - \lambda_{n,2}}{B(t_{n-1}) - \lambda_{n,1}}$\;
        \For{$\tau \in [\tau_{n-1}, \tau_n)$}{
            $s \gets \tau - \tau_{n-1}$\;
            $B(\tau) \gets \dfrac{\lambda_{n,2}  - \lambda_{n,1} R_n e^{-\gamma s}} {1 - R_n e^{-\gamma s}}$\;
        }
    }

    $B_{\text{prev}} \gets B(\tau_n)$\; \tcp*{Update for next interval}
}
\end{algorithm}

\begin{algorithm}[!htb]
\caption{Penalized policy (Howard) iteration for ADI--MCS timestepping}
\label{alg:penalized_policy_mcs}
\KwIn{
Previous solution $\mathbf{u}^{\,n-1}$,
Operators $A_0^{\,n},A_1^{\,n},A_2^{\,n}$ and
$A_0^{\,n-1},A_1^{\,n-1},A_2^{\,n-1}$,
Boundary vectors $\mathbf{b}^{\,n},\mathbf{b}^{\,n-1}$,
Payoff vector $\boldsymbol{\phi}$, penalty parameter $p$,
Timestep $\Delta t^n$, parameter $\theta$,
Tolerance $\mathrm{tol}$, maximum policy iterations $M_{\max}$
}
\KwOut{Solution vector $\mathbf{u}^n$}
\SetAlgoLined
\DontPrintSemicolon

\textbf{Initialization:}\;
$\mathbf{u}^{n,0} \gets \mathbf{u}^{n-1}$\;
$\Pi^{n,0} \gets \Pi(\mathbf{u}^{n-1})$\;
\tcp{Policy: $\Pi_i=1$ if $\phi_i>\mathbf{u}_i$, else $0$}

\tcp{MCS predictor}
$\mathbf{Y}_0 \gets
\mathbf{u}^{n-1}
+ \Delta\tau^n\left(A^n\mathbf{u}^{n-1}+\mathbf{g}^{n-1}\right)$\;

\tcp{First ADI correction}
Solve
\[
\left(I-\theta\Delta\tau^n A_1^n\right)\mathbf{Y}_1
=
\mathbf{Y}_0
+\theta\Delta\tau^n\left(\mathbf{g}_1^n-\mathbf{g}_1^{n-1}\right)
\]\;

\tcp{Second ADI correction}
Solve
\[
\left(I-\theta\Delta\tau^n A_2^n\right)\mathbf{Y}_2
=
\mathbf{Y}_1
+\theta\Delta\tau^n\left(\mathbf{g}_2^n-\mathbf{g}_2^{n-1}\right)
\]\;

\tcp{Explicit MCS corrections}
$\tilde{\mathbf{Y}}_0 \gets
\mathbf{Y}_0
+ \sigma\Delta\tau^n\left(A_0^n\mathbf{Y}_2+\mathbf{g}_0^n-\mathbf{g}_0^{n-1}\right)$\;

$\hat{\mathbf{Y}}_0 \gets
\tilde{\mathbf{Y}}_0
+ \mu\Delta\tau^n\left(A^n\mathbf{Y}_2+\mathbf{g}^n-\mathbf{g}^{n-1}\right)$\;

\tcp{Final ADI solve before policy iteration}
Solve
\[
\left(I-\theta\Delta\tau^n A_1^n\right)\hat{\mathbf{Y}}_1
=
\hat{\mathbf{Y}}_0
+\theta\Delta\tau^n\left((A_1^{n}-A_1^{n-1})\mathbf{u}^{n-1}+\mathbf{g}_1^n-\mathbf{g}_1^{n-1}\right)
\]\;

\For{$m = 1,2,\ldots,K_{\max}$}{
    \tcp{Build diagonal policy matrix}
    $P^{n,m-1} \gets \mathrm{diag}(\Pi^{n,m-1})$\;

    \tcp{Policy-evaluation (linear solve)}
    Solve
    \[
    \left(I-\theta\Delta\tau^n A_2 + p P^{n,m-1}\right)\mathbf{u}^{n,m}
    =
    \hat{\mathbf{Y}}_1
    + \theta\Delta\tau^n\left(\mathbf{g}_2^n-A_2^{n-1}\mathbf{u}^{n-1}-\mathbf{g}_2^{n-1}\right)
    + p P^{n,m-1}\boldsymbol{\phi}
    \]\;

    \tcp{Policy improvement}
    $\Pi^{n,m} \gets \Pi(\mathbf{u}^{n,m})$\;

    \tcp{Stopping criteria}
    \If{
    $\Pi^{n,m} = \Pi^{n,m-1}$
    \textbf{or}
    $\displaystyle
    \max_{i,j}
    \frac{|\mathbf{u}^{n,m}_{i,j}-\mathbf{u}^{n,m-1}_{i,j}|}
         {\max\{1,|\mathbf{u}^{n,m}_{i,j}|\}}
    < \text{tol}$
    }{
        \textbf{break}\;
    }
}

$\mathbf{u}^n \gets \mathbf{u}^{n,m}$\;
\end{algorithm}

\begin{algorithm}[!htb]
\small
\DontPrintSemicolon
\SetAlgoLined
\caption{Compute the Double-COS Approximation}
\label{alg:double_cos}

\SetKwBlock{Params}{Parameters:}{}
\Params{
    $N_x, N_v$ \tcp*{Number of cosine terms}
    $L_x, L_v$ \tcp*{Truncation ranges}
    $a_x, a_v$ \tcp*{Lower truncation bounds}
}

\BlankLine
\tcp{Precompute CF $\phi_{k,m}$}
\For{$k = 0$ \KwTo $N_x-1$}{
    \For{$m = 0$ \KwTo $N_v-1$}{
        $\phi_{k,m} \gets \phi\left( \dfrac{k\pi}{L_x}, \dfrac{m\pi}{L_v} \mid x,v,u \right)$\;
    }
}

\BlankLine
\tcp{Compute $J(k,m)$ via DCT-I}
\For{$k = 0$ \KwTo $N_x-1$}{
    \tcp{Sample $g_k(v)$}
    \For{$j = 0$ \KwTo $N_v-1$}{
        $v_j \gets a_v + j \cdot \dfrac{L_v}{N_v-1}$\;
        $x^*_j \gets x^*(v_j, u)$\;
        $g_k(v_j) \gets$ compute using \eqref{gkv}\;
    }
    \BlankLine
    \tcp{Apply DCT-I}
    $\{J(k,m)\}_{m=0}^{N_v-1} \gets \text{DCT-I}\left( \{g_k(v_j)\}_{j=0}^{N_v-1} \right)$\;
}

\BlankLine
\tcp{Compute the final approximation}
$I \gets 0$\;
\For{$k = 0$ \KwTo $N_x-1$}{
    \For{$m = 0$ \KwTo $N_v-1$}{
        $w_k \gets k == 0\,\, ?\,\, 0.5\, :\, 1$\;
        $w_m \gets m == 0\,\, ?\,\, 0.5\, :\, 1$\;
        $\text{term} \gets \Re\left[ e^{-i\left( \dfrac{k\pi a_x}{L_x} + \dfrac{m\pi a_v}{L_v} \right)} \phi_{k,m} \right] \cdot J(k,m)$\;
        $I \gets I + \dfrac{4}{L_x L_v} \cdot w_k \cdot w_m \cdot \text{term}$\;
    }
}
\Return{$I$}\;
\normalsize
\end{algorithm}

\end{document}